\documentclass{article}
\usepackage{booktabs}
\usepackage{braket}
\usepackage{color}
\usepackage{enumitem}
\usepackage[margin=1.0in]{geometry}
\usepackage[backref]{hyperref}
\usepackage{marginnote}
\usepackage{multirow}
\usepackage{subcaption}
\usepackage{tikz}
\usepackage{amsmath}
\usepackage{amsthm}
\usepackage{amsfonts}
\usepackage{amssymb}
\usepackage[figuresleft]{rotating}
\usepackage{tablefootnote}
\usepackage{float}
\usepackage{dsfont}
\usetikzlibrary{quantikz2}

\usepackage{listings}

\definecolor{codegreen}{rgb}{0,0.6,0}
\definecolor{codegray}{rgb}{0.5,0.5,0.5}
\definecolor{codepurple}{rgb}{0.58,0,0.82}
\definecolor{backcolour}{rgb}{0.95,0.95,0.92}

\newfloat{lstfloat}{htbp}{lop}
\floatname{lstfloat}{Listing}
\hypersetup{
    colorlinks=true,
    linkcolor=blue,
    filecolor=magenta,      
    urlcolor=cyan,
    citecolor=red,
    pdftitle={qec},
    pdfpagemode=FullScreen,
}

\newtheorem{theorem}{Theorem}
\newtheorem{lemma}[theorem]{Lemma}
\newtheorem{proposition}[theorem]{Proposition}

\newtheorem{postulate}{Postulate}

\theoremstyle{definition}
\newtheorem{definition}[theorem]{Definition}
\newtheorem{example}[theorem]{Example}

\theoremstyle{remark}

\DeclareMathOperator{\sym}{Sym}
\DeclareMathOperator{\cnot}{\texttt{CNOT}}
\DeclareMathOperator{\swapop}{\texttt{SWAP}}

\DeclareMathOperator{\real}{Re}
\DeclareMathOperator{\imag}{Im}

\title{Introduction to Quantum Error Correction with Stabilizer Codes}
\author{Zachary P. Bradshaw\footnote{\textbf{Corresponding Author:} zak@qodexquantum.ai}, Jeffrey J. Dale, and Ethan N. Evans}
\date{QodeX Quantum, Inc.}

\begin{document}

\maketitle

\begin{abstract}
    We give an introduction to the theory of quantum error correction using stabilizer codes that is geared towards the working computer scientists and mathematicians with an interest in exploring this area. To this end, we begin with an introduction to basic quantum computation for the uninitiated. We then construct several examples of simple error correcting codes without reference to the underlying mathematical formalism in order to develop the readers intuition for the structure of a generic code. With this in hand, we then discuss the more general theory of stabilizer codes and provide the necessary level of mathematical detail for the non-mathematician. Finally, we give a brief look at the elegant homological algebra formulation for topological codes. As a bonus, we give implementations of the codes we mention using OpenQASM, and we address the more recent approaches to decoding using neural networks. We do not attempt to give a complete overview of the entire field, but provide the reader with the level of detail needed to continue in this direction.
\end{abstract}
\section{Introduction}

The invention of the quantum computer \cite{nielsen2010} is expected to be a driver of innovation in the 21st century, solving a subset of classically intractable problems efficiently. These devices are not expected to provide a dramatic improvement over existing classical methods for \textit{every} problem, but typically operate well on tasks that ask for a global feature of some complex mathematical object. Perhaps unsurprisingly, they simulate quantum systems well \cite{berry2006,buluta2009,cirac2012,daley2022,feynman1982,georgescu2014,johnson2014,leimkuhler2004,lloyd1996,low2019,trabesinger2012}, a task that is difficult for classical devices due to the sheer size of the Hilbert spaces which model simple physical systems. With each addition of a qubit, the Hilbert space associated to a quantum computer grows in complex dimension by a multiplicative factor of 2, alleviating this issue. Thus, quantum computers have the potential to efficiently solve the related problems of drug and material simulation \cite{babbush2018,cao2018,clinton2024,kumar2024}. In addition to quantum simulation, quantum computers are expected to perform well on optimization tasks, efficiently solving problems in areas such as logistics, finance, and industrial design \cite{azzaoui2021,bayerstadler2021,egger2020,herman2023,orus2019,perdomo2019,weinberg2023}. A major advance in this direction was the Quantum Approximate Optimization Algorithm (QAOA) \cite{blekos2024,farhi2014,zhou2020} for solving combinatorial optimization problems. Also worth mentioning is the recent progress in quantum machine learning \cite{biamonte2017,cerezo2022,schuld2015}, a field which attempts to combine the recent success of classical machine learning models with quantum theory in order to break through limitations in current models due to resource deficiency. Much of this work makes use of a variational quantum circuit \cite{cerezo2021,chen2020,schuld2020} built of parameterized unitary operators which are trained using some classical optimization procedure, so that the training process is done in a hybrid fashion. There has been recent progress in constructing quantum versions of classical models which are known to perform well, such as the self-attention-based transformer architecture responsible for the success of large language models \cite{di2022,evans2024,kerenidis2024,li2024}.

Although there are many expected applications for quantum devices, they generally operate with a lot of overhead, creating a trade-off between the cost of running the device and the performance boost they produce. This overhead comes largely in the form of errors which are introduced when the controlled quantum system interacts with its environment. The experimentalist attempts to suppress these errors at the hardware level using any of a number of techniques dependent on the architecture of the device. The subject of this review is instead quantum error correction \cite{devitt2013,gottesman2002,knill1997,lidar2013,roffe2019}, which attempts to detect and correct errors to the individual qubits in a system as they occur. We distinguish our approach from other reviews on the same topic by directing our review towards the interested mathematicians and computer scientists who do not have a background in quantum theory. To this end, we include an overview of the basics of quantum theory, including the circuit model of quantum computation, and we develop the necessary background in group theory to understand the stabilizer code formalism for the computer scientist who has not had the luxury of taking a first course in the subject. We furthermore implement the simple error correcting codes discussed in this work on IBM's real quantum devices, providing example code for each case. In an effort to stay as hardware-neutral as possible, we do not confine ourselves to using a quantum computing
library that is tied to a specific company.
Although these libraries are highly convenient for those that already have a target platform in mind, our
focus is on maximizing flexibility to accommodate as many readers as possible, as well as remaining
future-proof to the best of our ability.
Thus, our code examples are written directly in OpenQASM 3.0~\cite{cross2022openqasm}, an open standard for
describing quantum circuits.
As of writing, OpenQASM code can be readily imported into
    IBM Qiskit~\cite{cross2018ibm},
    Google Cirq~\cite{cirq},
    Amazon Braket~\cite{braket},
    Microsoft Azure Quantum~\cite{azurequantum},
    and Xanadu PennyLane~\cite{bergholm2018pennylane};
though, many of these libraries support only a subset of the OpenQASM 3.0 standard.
All of our code is imported into Qiskit, which transpiles the circuit for execution on IBM quantum hardware.

For some time, it was speculated that maintaining coherence in a quantum device long enough to perform a useful calculation was unlikely \cite{unruh1995}. Moreover, the well-developed theory behind classical error correcting codes appeared to be little use in the quantum regime. After all, the typical trick of securing a message by creating redundancy through the cloning of information is provably impossible in this setting. To make matters worse, a qubit can be subject not just to a bit flip error, but also to an infinitude of additional errors, making the hope for a useful error correcting code seem bleak. However, just before the turn of the century, a flurry of codes were proposed which in principal did just this \cite{bennett1996,calderbank1996,laflamme1996,steane1996}, perhaps the most famous of which was the 9 qubit code proposed by Shor \cite{shor1995}. These works showed that there was indeed hope for a fault-tolerant quantum device and drove the great progress that has been seen in the past two decades.

A major conceptual advance that helped unify and generalize many of the earliest quantum codes was the introduction of the stabilizer formalism by Gottesman~\cite{gottesman1998}. This framework provides a powerful language for describing a large class of quantum error correcting codes using tools from finite group theory. A stabilizer code is defined by an abelian subgroup of the $n$-qubit Pauli group, where each element of the subgroup fixes the code space, and the error detection and correction properties of such a code can be understood by analyzing how various Pauli errors commute or anticommute with the stabilizer generators. Because the Pauli group forms a basis for all single- and multi-qubit errors under the operator-sum decomposition, stabilizer codes are in principle sufficient to address arbitrary errors affecting a finite number of qubits.

Building on this algebraic foundation, it was soon discovered that the geometry of qubits arranged in a lattice could be used to further improve the robustness and scalability of error correcting codes. This gave rise to the field of topological quantum error correction, in which quantum information is stored non-locally across many physical qubits in a way that is resilient to local noise~\cite{kitaev2003,dennis2002}. In the surface code, physical qubits are laid out on the edges or faces of a 2D lattice, and stabilizer checks correspond to measurements of operators defined by the local features of the lattice~\cite{fowler2012}. Errors appear as defect chains on the lattice, and decoding involves finding the most likely configuration of such chains that could have produced the observed syndrome data~\cite{dennis2002}. One of the major advantages of topological codes is their compatibility with hardware platforms that support only local qubit interactions in two dimensions. The surface code, in particular, boasts one of the highest known fault-tolerance thresholds and has become a central focus for experimental implementations of scalable quantum computing~\cite{fowler2012,raussendorf2007}. Its geometric locality enables efficient and repeated syndrome extraction, and the code distance can be increased simply by growing the size of the lattice~\cite{fowler2012}.

Despite the theoretical appeal of topological codes, the decoding problem, the task of inferring the most likely error given a syndrome, is computationally intensive. For the surface code, decoders often require significant classical computation and are difficult to adapt to hardware constraints or realistic noise models. This challenge has led to growing interest in machine learning approaches to decoding. These methods aim to learn the structure of likely errors directly from data, thereby avoiding the need to explicitly model the underlying noise distribution. Classical neural networks, such as convolutional and recurrent architectures, have been trained to predict corrections or logical operators given a syndrome pattern. In many cases, they have achieved performance comparable to or better than traditional decoders while offering greater adaptability and lower latency when deployed on specialized hardware. One of the most notable developments in this area is Google's recent introduction of AlphaQubit~\cite{bausch2024,willow2024}, a machine-learned decoder based on the very popular transformer architecture.

The remainder of this work aims to equip the reader with the necessary background to understand and evaluate the above developments. We begin with the necessary background in quantum theory to understand the challenges of correcting quantum errors in Section~\ref{sec:basicquantum}. This includes a discussion of the postulates of quantum theory as well as the circuit model of quantum computation. In Section~\ref{sec:basiccodes}, we then proceed to discuss several simple examples of quantum error correcting codes, including the 2-qubit detection code, the three qubit bit flip and phase flip codes, and the Shor code which is built by concatenating the 3 qubit codes into a 9 qubit code. Although each of these codes can be viewed as a stabilizer code, we make no reference to the underlying theory at this point, instead focusing on building the reader's intuition for how an error correcting code should work. In each of these examples, we give explicit OpenQASM code, providing the reader an opportunity to gain practical experience with the implementation of quantum circuits. We then discuss the mathematical formalism behind stabilizer codes, making sure to include a review of the necessary group theory along the way in Section~\ref{sec:stabilizer}. This theoretical discussion is complemented by an additional example of a stabilizer code beyond those presented in Section~\ref{sec:basiccodes}, including OpenQASM code capable of implementing the code on existing hardware. In Section~\ref{sec:topological-codes}, we then give a brief look at the beautiful theory of topological codes, where the code is embedded into the homological degrees of freedom of a surface defined by a lattice. As a bonus, we discuss the use of machine learning in decoding stabilizer codes in Section~\ref{sec:decoderML}.

\section{Review of Basic Quantum Theory}\label{sec:basicquantum}
\subsection{The Postulates}
We begin with the postulates of quantum theory \cite{nielsen2010,shankar2012} from which all other phenomenon in the field arise. 
The postulates themselves are not proven but rather form an axiomatic foundation for 
the theory and must therefore be accepted as fact in the same sense that a mathematician would accept the axioms of set theory. 
The difference is that quantum theory comes equipped with a physical interpretation so that predictions about observable phenomenon can be made, and elaborate, expensive experiments can be designed to check whether nature agrees with us. 
We gain confidence that our axioms are ``correct'' when we employ such tests and find that reality has not yet staged a formal protest.

As notation is half the battle, let us sort that out first. 
We will make use of the Dirac notation in which elements of a vector space 
are distinguished from operators and scalars by wrapping them with the so-called ``ket'' symbol: $\ket{\psi}$. 
Inner products between vectors will be denoted $\langle \psi\vert\phi\rangle$. 
Equivalently, the Riesz representation theorem allows us to think of this quantity as the action of a 
linear functional in the dual space on an element of the vector space. 
For notational convenience, we therefore denote elements of the dual space by the ``bra'' symbol $\bra{\psi}$ 
so that we have $\bra{\psi}(\ket{\phi})=\langle \psi\vert\phi\rangle$. We will frequently deal with tensor products
of vector spaces, which are spanned by elements of the form $\ket{\psi}_1\otimes\cdots\otimes\ket{\psi}_n$.
It will therefore benefit us greatly to introduce the compact notation $\ket{\psi}_1\cdots\ket{\psi}_n:=\ket{\psi}_1\otimes\cdots\otimes\ket{\psi}_n$, which drops the tensor product symbol altogether. We note also that the symbol $\ket{\psi}\!\!\bra{\phi}$ is interpreted as the outer product of the vectors $\ket{\psi}$ and $\ket{\phi}$; that is, this symbol defines an operator on our vector space which acts as follows: $\ket{\psi}\!\!\bra{\phi}(\ket{\chi})=\ket{\psi}\braket{\phi\vert\chi}=\braket{\phi\vert\chi}\ket{\psi}$, where in the last equality, we have used the fact that the inner product $\braket{\phi\vert\chi}$ produces a scalar value.

The first postulate of quantum theory tells us that there is a quantity $\ket{\psi}$ known as the 
state of the quantum system which takes the form of an element of a complex Hilbert space $\mathcal{H}$. 
A Hilbert space is a vector space equipped with a complete inner product 
$\langle\cdot,\cdot\rangle:\mathcal{H}\times\mathcal{H}\to\mathbb{C}$, and completeness means that all Cauchy sequences in the space
converge with respect to the distance metric 
\begin{equation}\label{eq:distance-metric}
    d(\ket{\psi},\ket{\phi})=\|\ket{\psi}-\ket{\phi}\|=\sqrt{(\bra{\psi}-\bra{\phi})(\ket{\psi}-\ket{\phi})}.
\end{equation}
Usually the requirement of completeness acts in the background, 
and so we will forgo an intensive study of this property.

\begin{postulate} The state of a closed quantum system is an element of a complex Hilbert space.
\end{postulate}

The Hilbert space structure allows us to take linear combinations (or superpositions) of states. 
That is, if $\alpha,\beta\in\mathbb{C}$ and $\ket{\psi},\ket{\phi}\in\mathcal{H}$, 
then $\alpha\ket{\psi}+\beta\ket{\phi}\in\mathcal{H}$. Combined with the next postulates, 
we will see that this property has several interesting consequences. Indeed, the ability to blend
states isn't just mathematical window dressing; it's a fundamental feature. This phenomenon appears classically in electrodynamics (or any other wave theory for that matter), where electromagnetic waves are allowed to interfere with each other. Mathematically, this interference appears as a sum of the independent electromagnetic waves. Another example is the wave produced by dropping a pebble into water. The amplitudes of two such waves add (or subtract as the case may be) in regions where they intersect, producing larger (or smaller) overall peaks in the water. The principle is the same here, only this time the wave being propagated is related to the probability that some property is observed when a measurement is made. Thus, quantum theory allows us to interfere these probability waves in the same way we would mathematically interfere any other pair of waves.

Measurements introduce an oddity into the theory. In contrast to classical 
mechanics, a measurement of a quantum system will in general disturb it, resulting 
in loss of information contained in the state that existed prior to the measurement. 
The mechanism responsible for this disturbance is the subject of heated arm-waving in 
philosophical circles, but we will have little need to address such matters here. 
For the most part, we take the pragmatic approach and accept the following postulate without regard for the disturbance mechanism.

\begin{postulate} \label{post:measure} For every observable quantity $\mathcal{O}$, 
    there is an associated self-adjoint operator $H^\mathcal{O}$. A measurement of 
    $\mathcal{O}$ results in the eigenvalue $\lambda$ of $H^\mathcal{O}$ with probability 
\begin{equation}\label{eq:post2}
P(\lambda)=\frac{\lvert\langle\lambda\vert\psi\rangle\rvert^2}{\lvert\langle\psi\vert\psi\rangle\rvert^2\lvert\langle\lambda\vert\lambda\rangle\rvert^2},
\end{equation}
where $\ket{\psi}$ is the state of the system and $\ket{\lambda}$ is an eigenvector 
associated to $\lambda$. Immediately after the measurement, the state of the system 
becomes $\ket{\lambda}/\lvert\langle\lambda\vert\lambda\rangle\rvert^2$.
\end{postulate}

The assumption that $H^\mathcal{O}$ is self-adjoint ensures that its eigenvalues are real, which is a desirable property for physically observable quantities. When our Hilbert space is finite dimensional, the operator $H^\mathcal{O}$ will have finitely many eigenvalues, so that there are only a finite number of possible measurement outcomes for the observable $\mathcal{O}$. This will almost always be the case in quantum computing, where our Hilbert space takes the form $\mathbb{C}^{2^n}$ for some $n\in\mathbb{N}$; although there are continuous variable theories for which this is not the case. Such theories may have operators with a discrete infinitude of eigenvalues or even a continuous spectrum of eigenvalues. As we are concerned with the creation of fault-tolerant quantum devices and most proposals for such devices to date are finite dimensional, we will restrict our attention to this simpler regime. The interested reader is referred to \cite{aoki2009,braunstein1998,lloyd1998} for the quantum error correction of continuous variable systems.

A common convention is to assume that all states are normalized to unity (except in some infinite dimensional settings where this is impossible). This assumption has the effect of removing the denominator in \eqref{eq:post2} without affecting the probability distribution, and there are therefore no measurable consequences for making it. We will adopt this convention throughout the remainder of this work.

Recall that the eigenvectors $\ket{\lambda}$ for $H^\mathcal{O}$ form an orthogonal basis for our Hilbert space. Thus, we may expand an arbitrary state $\ket{\psi}$ as
\begin{equation}\ket{\psi}=\sum_{\lambda}c_\lambda\ket{\lambda},
\end{equation}
for some constants $c_\lambda\in\mathbb{C}$ and where the sum is over all eigenvalues of $H^\mathcal{O}$. It follows that the probability of obtaining outcome $\lambda$ when measuring $\mathcal{O}$ is
\begin{equation}P(\lambda)=\lvert\langle\lambda\vert\psi\rangle\rvert^2=\lvert c_\lambda\rvert^2.
\end{equation}
That is, the probability distribution associated to this measurement is encoded in the coefficients of $\ket{\psi}$ when written in the corresponding eigenbasis.

Notice that Postulate~\ref{post:measure} says nothing about which self-adjoint operator should be associated to the observable $\mathcal{O}$, only that such an operator exists. The problem of determining this operator for a given classically observable quantity is called quantization, and there are many approaches that can be taken when constructing a quantization scheme \cite{hall2013}. Unfortunately, and perhaps remarkably, it was shown by Groenewold that a self-consistent quantization routine satisfying several natural properties proposed by Dirac does not exist \cite{groenewold1946}, and so to some extent we must rely on experimental evidence to guide us. The reader is directed to \cite{carosso2022} for a nice historical account of the problem of quantization. In quantum computing, this is an unnecessary detail that has no material effect on the theory. To be clear, this problem must be tackled in practice; to build a quantum device, we must rely on a dynamical interpretation of the eigenvalues of a self-adjoint operator. However, in constructing theoretical quantum computational models, this detail can be omitted by simply assuming that for any self-adjoint operator we may perform something, which we call a ``measurement'', that has the effect of collapsing our quantum state into an eigenvector of the operator. For this reason, we will often drop any mention of the observable $\mathcal{O}$, instead preferring to say that measurements are taken with respect to a self-adjoint operator $H$ or that a measurement has been taken in the $H$ basis (the basis of eigenvectors).

\begin{example} \label{example:zmeasure} Let us take as a first example the self-adjoint Pauli-$Z$ operator
\begin{equation}Z=\begin{pmatrix}1 & 0\\ 0&-1
\end{pmatrix}
\end{equation}
acting on the Hilbert space $\mathbb{C}^2$ (equipped with the standard euclidean inner product). It has eigenvalues $\lambda=\pm1$ with (orthonormal) eigenvectors $\ket{0}:=(1,0)^T$ and $\ket{1}:=(0,1)^T$. Given a state $\ket{\psi}\in\mathbb{C}^2$, we therefore have a decomposition $\ket{\psi}=\alpha\ket{0}+\beta\ket{1}$ for some constants $\alpha,\beta\in\mathbb{C}$. After a measurement in the $Z$ basis, we obtain outcome $\lambda=1$ with probability
\begin{equation} P(\lambda=1)=\lvert\langle0\vert\psi\rangle\rvert^2=\lvert\langle0\vert(\alpha\ket{0}+\beta\ket{1})\rvert^2=\lvert\alpha\rvert^2
\end{equation}
where we have used the orthonormality of the eigenbasis in the last equality. Similarly, the probability of obtaining outcome $\lambda=-1$ is $\lvert\beta\rvert^2$. If the positive eigenvalue is obtained, we know that the state of the system is now $\ket{0}$. Likewise, if the negative eigenvalue is obtained, the state of the system is now $\ket{1}$. Observe that if $\ket{\psi}=\ket{0}$ at the start of this experiment (so $\alpha=1$ and $\beta=0$), then we obtain the outcome $\lambda=+1$ with absolute certainty. Thus, whatever the observable is associated to $Z$, the state $\ket{0}$ can be said to have the definite value $\lambda=+1$ for that observable. In general, an eigenvector of a self-adjoint operator has a definite value for the associated observable, and that quantity is given by the corresponding eigenvalue.\qed
\end{example}

Quantum theory is inherently probabilistic, as outlined in Postulate~\ref{post:measure}. This is in sharp contrast to the classical theory, in which an object has a definite position and momentum at all times. Indeed, in Example~\ref{example:zmeasure} we prepared a state $\ket{\psi}$ and took a measurement, getting say the $+1$ outcome and collapsing the state to $\ket{0}$. If we were to prepare the exact same state and perform exactly the same experiment in the same way, we can still obtain the $-1$ outcome with some probability. Thus, before the measurement is made, there is no definite classical interpretation of the observable associated to the Pauli-$Z$ operator unless the state of the system happens to be an eigenvector. This challenges our classical intuition, and the reader who doubts this description is in good company \cite{einstein1935}. However, Bell's theorem \cite{bell1964} tells us that quantum theory cannot satisfy at the same time \textbf{locality} (no faster than light interactions) and \textbf{realism} (quantum systems possess a definite value for an observable before measurement). The statement of Bell's theorem requires that any hidden variables be uncorrelated to the measurement setting, which is usually justified by claiming the experimenter performing the measurement has free will. Thus, a loophole can be made by dropping this assumption. Since no experiment has yet been devised to distinguish between these possibilities, it is up to the reader to decide whether they believe the universe is local or real (or neither) or if they are willing to accept a superdeterministic theory.

In Example~\ref{example:zmeasure}, we made a measurement with respect to the $Z$ basis and saw that the state of the system collapsed to an eigenvector of $Z$, so that if another measurement is made with respect to the $Z$ basis, then we obtain precisely the same result and the state of the system remains the same eigenvector. Thus, the system has a definite value for the observable associated to $Z$. Let us see what happens when we follow the $Z$-measurement with a measurement with respect to another operator.

\begin{example} \label{example:zxzmeasure} Let $\ket{\psi}=(\ket{0}+\ket{1})/\sqrt{2}$ as in Example~\ref{example:zmeasure}. We will take three measurements, once with respect to the $Z$-basis, then with respect to the $X$-basis, and then one more with respect to the $Z$-basis. After a measurement with respect to the $Z$-basis, the state of the system is either $\ket{0}$ or $\ket{1}$, each with probability $1/2$. Suppose we now measure with respect to the self-adjoint Pauli-$X$ operator
\begin{equation}X=\begin{pmatrix}0 & 1\\ 1 & 0
\end{pmatrix}.
\end{equation}
Its eigenvalues are again $\lambda=\pm1$ but this time the eigenvectors are $\ket{+}=\frac{1}{\sqrt{2}}(1,1)^T=\frac{1}{\sqrt{2}}(\ket{0}+\ket{1})$ and $\ket{-}=\frac{1}{\sqrt{2}}(1,-1)^T=\frac{1}{\sqrt{2}}(\ket{0}-\ket{1})$. Just as with the $Z$-measurement, we can now say that the state of our system has a definite value for the observable associated to $X$. Moreover, before the $X$ measurement, our system had a definite value for the observable associated to $Z$, but if we now make another $Z$-measurement, we obtain outcomes $\lambda=\pm1$ with equal probability. Thus, the information about the observable associated to $Z$ has been destroyed by the measurement with respect to $X$. \qed
\end{example}

The surprising nature of this example is perhaps best exhibited by the Stern-Gerlach experiment \cite{gerlach1989}, in which (in a variant form) an electron is passed through a uniform magnetic field in order to measure its spin, an intrinsic form of angular momentum that all electrons possess \cite{sakurai2020}. The electron is deflected by the magnetic field either up or down due to the magnetic moment of the electron induced by its spin angular momentum. It then hits a backwall, where it interacts with a fluorescent material, allowing us to see whether it was deflected upward or downward. When the experiment is implemented, we observe two points equidistant from the original trajectory of the electron on the backwall, rather than the continuous line that would be expected if the spin were allowed to vary in magnitude. Thus, the experiment confirms that the intrinsic spin angular momentum of the electron can take only two possible values. If we now restart the experiment but rotate the magnetic field apparatus so that we are measuring the component of the spin angular momentum in an orthogonal direction, we obtain the same result but now our electrons are deflected either left or right (up or down in the orthogonal direction). Thus, we have an apparatus for measuring spin in two orthogonal directions. If we now concatenate three copies of this apparatus, the first of which measures spin in one direction, the second of which measures spin in the orthogonal direction, and the third of which measures spin in the original direction, then we now have an apparatus to perform the experiment outlined in Example~\ref{example:zxzmeasure}. After the first measurement, the spin in the original direction is known with certainty, and after the second measurement the spin in the orthogonal direction is known with certainty. Thus, it may seem that we now have knowledge of the spin in both directions, but due to the collapse of the state during the second measurement, when we make the third measurement, we find that the spin in the original direction once again follows a nondeterministic probability distribution.

We now turn to the third postulate, which tells us how a quantum system evolves in time. The classical analog of this postulate is Newton's second law, which for a massive particle takes the form of the second order ordinary differential equation
\begin{equation} F=m\frac{d^2x}{dt^2},
\end{equation}
where $x(t)$ denotes position, $m$ is the mass of the system, and $F$ denotes net force. 
Given the force and mass, there is a unique solution to this differential equation up to 
the specification of two constants, the initial position and momentum. 
Thus, given these initial conditions, Newton's law allows us to predict the entirety 
of the future trajectory of the particle. In the quantum theory, the state of a closed 
system similarly evolves according to a differential equation (called the Schr\"odinger equation), 
and a solution to this partial differential equation allows us to predict how the state 
will evolve in time with the caveat that a measurement will suddenly force a nondeterministic 
collapse of the state into an eigenvector of some self-adjoint operator.
We make sense of this by noting that when a measurement is made, the system is no longer closed.
``But wait,'' says the astute reader, ``couldn't you just include the measurement 
apparatus (and whatever else is necessary) in the system so that this larger system is closed?'' This question of whether 
Postulate~\ref{post:measure} can be derived from Postulate~\ref{post:evolve} 
is still being debated today \cite{masanes2019,kent2023,masanes2025}. For our purposes, it will be clear when one should
apply the second postulate versus the third, and so we roll up this important concern
and place it in our suitcase for another day.

\begin{postulate} \label{post:evolve}The state $\ket{\psi(t)}$ of a closed quantum system evolves according to the Schr\"odinger equation
\begin{equation}\label{eq:schrodinger}
i\hbar\frac{\partial}{\partial t}\ket{\psi(t)}=H\ket{\psi(t)},
\end{equation}
where $H$ is a self-adjoint operator known as the Hamiltonian of the system, $i$ is the imaginary unit, and $\hbar$ is known as Planck's constant.
\end{postulate}

In classical mechanics, the Hamiltonian formalism is an alternative framework to the Newtonian formalism introduced in a first course. Newton's second law is replaced by the Hamilton equations, a system of coupled linear differential equations which provide the same information \cite{taylor2005}. The Hamiltonian of the classical system is then taken to be any function in the position and momentum which produces the same dynamics as Newton's formalism, but typically it takes the form of the total energy of the system $H=\frac{p^2}{2m}+V$, where the first term is the kinetic energy and the second term is the potential energy of the system. In quantum theory the Hamiltonian is the self-adjoint operator associated to the energy observable, which can sometimes be derived through quantization of a classical system. For now, we content ourselves with the knowledge that if $H$ is time independent and self-adjoint, then a solution to \eqref{eq:schrodinger} is
\begin{equation}\ket{\psi(t)}=e^{-iHt/\hbar}\ket{\psi(0)},
\end{equation}
where $\ket{\psi(0)}$ is the initial state of the system, and we have ignored the subtleties of showing that the operator exponential is indeed well-defined (this follows from the spectral theorem in finite dimensions). Notice that this exponential is unitary. Indeed, we have $(e^{-iHt/\hbar})^\dagger=e^{iH^\dagger t/\hbar}=e^{iH t/\hbar}$, where we have used the self-adjoint property for the Hamiltonian in the last equality. Thus, $e^{-iHt/\hbar}(e^{-iHt/\hbar})^\dagger=e^{-iHt/\hbar}e^{iHt/\hbar}=\mathds{1}$, so that the operator is unitary. By Stone's theorem, every strongly continuous one-parameter unitary group can be written as such a complex operator exponential of a self-adjoint operator. Thus, in principle, given a quantum state $\ket{\psi}$, we may apply any unitary operation we like to $\ket{\psi}$ by adjusting the Hamiltonian of the system. This is an important point which we will return to in Section~\ref{sec:circuit}.

The final postulate we need to discuss is the composite system postulate, which tells us that if we have a quantum system in state $\ket{\psi}_A\in\mathcal{H}_A$ as well as a quantum system in state $\ket{\phi}_B\in\mathcal{H}_B$, then the composite system is given by the tensor product of the two systems. 

\begin{postulate} Given two quantum systems $\mathcal{H}_A$ and $\mathcal{H}_B$, the composite system is given by $\mathcal{H}_A\otimes\mathcal{H}_B$.
\end{postulate}

This opens up the door for an incredible phenomenon for which there is no classical counterpart. Indeed, not only do we have states of the form $\ket{\psi}_A\ket{\phi}_B:=\ket{\psi}_A\otimes\ket{\phi}_B$, but due to the Hilbert space structure of the composite system, linear combinations of such states are also allowed. In some cases, such a linear combination cannot be broken down into a tensor product consisting of a state in $\mathcal{H}_A$ and a state in $\mathcal{H}_B$. States with this property are called \textbf{entangled}.

\begin{example} Consider the state $\frac{1}{\sqrt{2}}(\ket{00}+\ket{11})$, where we have used the compact notation $\ket{xx}:=\ket{x}\ket{x}=\ket{x}\otimes\ket{x}$. This state is an element of $\mathbb{C}^2\otimes\mathbb{C}^2$ which cannot be separated into a product of a state in $\mathbb{C}^2$ with another state in $\mathbb{C}^2$. Indeed, suppose such a decomposition exists. Then we have
\begin{align}\frac{1}{\sqrt{2}}(\ket{00}+\ket{11})&=(\alpha\ket{0}+\beta\ket{1})(\gamma\ket{0}+\delta\ket{1})\\
& =\alpha\gamma\ket{00}+\alpha\delta\ket{01}+\beta\gamma\ket{10}+\beta\delta\ket{11}
\end{align}
for some constants $\alpha,\beta,\gamma,\delta\in\mathbb{C}$. Since $\mathbb{C}$ is a division algebra, it follows that one of $\alpha$ and $\delta$ is zero and that one of $\beta$ and $\gamma$ is zero. But none of these constants can vanish, otherwise the coefficients of $\ket{00}$ and/or $\ket{11}$ vanish, producing a contradiction. Thus, the state is entangled.\qed
\end{example}

Entanglement is a remarkable consequence of our postulates which upon measurement seems to defy our classical intuition. Indeed, suppose two observers, Alice and Bob, share a pair of electrons in the entangled spin state $\frac{1}{\sqrt{2}}(\ket{00}+\ket{11})$, where the $\ket{0}$ state denotes spin up and the $\ket{1}$ state denotes spin down in some direction that we fix. The physical manifestations of the particles are distinct, Bob is free to take one of the two electrons with him on a trip to the Andromeda galaxy while Alice keeps the other on Earth, but the state of their spins is intertwined. Before departing, they agree that Alice will first make a measurement of the system, followed by Bob (we assume a non-relativistic setting for the sake of this example). If we subscribe to the non-local Copenhagen interpretation of quantum mechanics, then when Alice measures her electron's spin, the entire state collapses into $\ket{00}$ or $\ket{11}$, each with probability $1/2$. Whatever she gets, she knows with certainty that Bob will get the same thing when he makes a subsequent measurement. Thus, Alice's measurement affects the outcome of Bob's measurement. Let us now see things from Bob's perspective. Since Bob knows beforehand that Alice will make a measurement, he can conclude that before he makes his own measurement, the state of the system is either $\ket{00}$ or $\ket{11}$ with equal probability. He knows with certainty that the state of the system is given by one of these two eigenstates, but Alice has not communicated to him which case she observed. Thus, Bob is lacking some \textit{classical} information which would determine the system completely. There is no contradiction here, Bob is merely being a good Bayesian: Alice has a prior in the form of a measurement outcome which allows her to determine the state of the system completely, while Bob does not.

Now let us instead suppose that before departing, Alice and Bob agree that they will do nothing but measure their electron's spin in the same fixed direction, but this time they do not agree to an order of events. The situation is now symmetric, so let us choose Alice's perspective with the knowledge that Bob will experience the same thing. Before her measurement, Alice does not know whether Bob has already made a measurement. Thus, the true state of the system could be $\frac{1}{\sqrt{2}}(\ket{00}+\ket{11})$, $\ket{00}$, or $\ket{11}$. Let $p$ be the probability that Bob has made a measurement. Then with probability $1-p$, the state of the system is $\frac{1}{\sqrt{2}}(\ket{00}+\ket{11})$. Meanwhile, if Bob had made a measurement, the state would be $\ket{00}$ or $\ket{11}$, each with conditional probability 1/2. It follows that the probability of the system being in $\ket{00}$ is $p/2$ and likewise for $\ket{11}$. Thus, when Alice makes a measurement, she obtains either $\ket{00}$ or $\ket{11}$, each with probability $\frac12(1-p)+\frac{p}{2}=\frac12$. Therefore, at a statistical level, it does not matter whether Alice or Bob have knowledge of the other making a measurement. If they repeat this experiment many times, they will find that they obtain each outcome with equal probability. However, at the level of each experiment, Alice may have made the first measurement, affecting Bob's, or Bob may have made the first measurement, affecting Alice's, and there is no contradiction in the overall statistics. 

In this work, we choose to view things from the lens of the Copenhagen interpretation for its simplicity. This is a matter of taste since there are no known physically realizable differences between consistent interpretations. However, we should note that there are interpretations such as the approach of Everett, colloquially known as the many-worlds interpretation \cite{hugh1957}, which do away with the state collapse entirely. Instead, Alice's measurement forces a branching of the state so that Alice becomes entangled with the original Bell state:
\begin{equation}
    \frac{1}{\sqrt{2}}\left(\ket{00}\ket{\text{Alice saw 0}}+\ket{11}\ket{\text{Alice saw 1}}\right).
\end{equation}
Thus, no matter what Bob's subsequent measurement is, it agrees with that of Alice. However we interpret the mathematics, the thing of physical importance is that Alice and Bob have consistent results. Readers interested in the various interpretations of quantum mechanics and their extensions to a relativistic setting are referred to \cite{albert1994,brown2005,maudlin2011,nikoli2005,saunders2010}.

Before closing this section, let us emphasize that the above argument does NOT tell us that the state $\frac{1}{\sqrt{2}}(\ket{00}+\ket{11})$ is equivalent to being in the state $\ket{00}$ or $\ket{11}$ each with probability $\frac{1}{2}$. Indeed, although the statistics for these systems are the same when a measurement is taken in the computational basis, they are quite different when taken in another basis! Let us rewrite these states in terms of the Bell basis
\begin{align}
    \ket{\Phi^+}&=\frac{\ket{00}+\ket{11}}{\sqrt{2}}\\
    \ket{\Phi^-}&=\frac{\ket{00}-\ket{11}}{\sqrt{2}}\\
    \ket{\Psi^+}&=\frac{\ket{01}+\ket{10}}{\sqrt{2}}\\
    \ket{\Psi^-}&=\frac{\ket{01}-\ket{10}}{\sqrt{2}}.
\end{align}
We have $\frac{1}{\sqrt{2}}(\ket{00}+\ket{11})=\ket{\Phi^+}$ and so a measurement of this state in the Bell basis will
result in the state $\ket{\Phi^+}$ with probability $P(\ket{\Phi^+})=1$. On the other hand,
\begin{equation}
    \ket{00}=\frac{\ket{\Phi^+}+\ket{\Phi^-}}{\sqrt{2}}
\end{equation}
and
\begin{equation}
    \ket{11}=\frac{\ket{\Phi^+}-\ket{\Phi^-}}{\sqrt{2}}.
\end{equation}
Thus, if the system is in one of the states $\ket{00}$ or $\ket{11}$, each with probability $\frac{1}{2}$, then the probability that a measurement in the Bell basis results in the state $\ket{\Phi^+}$ is $P(\ket{\Phi^+})=\frac{1}{2}\cdot\frac{1}{2}+\frac{1}{2}\cdot\frac{1}{2}=\frac{1}{2}$. It follows that the statistics with respect to a measurement in the Bell basis differ, and if Alice does not obtain the eigenvalue corresponding to $\ket{\Phi^+}$ upon measurement, she knows that Bob's measurement in the computational basis was already made. However, observe that a measurement with respect to the Bell basis requires access to both qubits, and it is therefore not feasible for Alice to perform this experiment to distinguish between the two scenarios when she is spacelike separated from Bob.

%%%%%%%%%%%%%%%%%%%%%%%%%%%%%%%%%%%%%%%%%%%%%%%%%%%%%%%%%%%%%%%%%%%%%%%%%%%%%%%%%%%%%%%%%%%%%%%%%%%%%%%%%%%%%%%%%%%%%%%%
\subsection{The Circuit Model}\label{sec:circuit}

\begin{figure}
    \centering
    \begin{quantikz}
        \lstick{$\ket{\psi}$} & \gate[1]{U} &
    \end{quantikz}
    \caption{
        Circuit diagram for a unitary operator $U$ applied to an initial state $\ket{\psi}$.
    }\label{fig:first-circuit}
\end{figure}
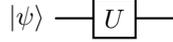

In quantum computing \cite{kaye2006,nielsen2010}, our goal is to forcibly evolve an initial quantum state in such a way that a measurement provides information about some problem of interest. 
Postulate~\ref{post:evolve} and the discussion thereafter allow us to evolve an initial quantum state by applying a sequence of unitary transformations. 
We depict this process diagrammatically in Figure~\ref{fig:first-circuit}, where the horizontal line is used to indicate the state that the unitary transformation acts on.
While this diagram resembles that of an electrical circuit, we note that there are likely no wires present in the physical apparatus implementing this diagram.
Still, we choose to call such a diagram a \textbf{quantum circuit}.

Let us introduce some of the main characters in quantum computation. 
We take as our Hilbert space $\mathbb{C}^{d}$ for some positive integer $d$. 
Usually we will take $d=2$, in which case the quantum system is called a \textbf{qubit}.
In general, for arbitrary $d$, we call the quantum system a \textbf{qudit}. The states $\ket{0}=(1,0)^T$ and $\ket{1}=(0,1)^T$ together form a basis for $\mathbb{C}^2$ known as the \textbf{computational basis}.
We saw in Example~\ref{example:zmeasure} that this basis is also the basis of eigenvectors of the Pauli-$Z$ operator.
The Pauli operators appear so often that it is worth repeating their definitions and noting some of their useful properties.
The Pauli $X$, $Y$, and $Z$ operators are given (in the computational basis) by
\begin{align}
    X&=\begin{pmatrix}
        0&1\\1&0
    \end{pmatrix}\\
    Y&=\begin{pmatrix}
        0&-i\\i&0
    \end{pmatrix}\\
    Z&=\begin{pmatrix}
        1&0\\0&-1
    \end{pmatrix}.
\end{align}
Together with the identity operator
\begin{equation}
    \mathds{1}=\begin{pmatrix}
        1&0\\0&1
    \end{pmatrix},
\end{equation}
these operators form a basis for the space of unitary operators $\mathcal{U}(\mathbb{C}^2)$, an important point that we will return to when discussing quantum errors.
It can also be checked that each of these operators is self-adjoint so that they can be regarded as observables.
Moreover, each of the Pauli operators squares to itself (in fact, this is true for any self-adjoint unitary operator) and has zero trace.

The Pauli-$X$ operator acts as the analog of a classical \texttt{NOT} gate on the computational basis. 
That is, we have $X\ket{0}=\ket{1}$ and $X\ket{1}=\ket{0}$, and for this reason, it is often called a \texttt{NOT} gate. 
The Pauli-$Z$ operator flips the sign of the relative phase factor of a qubit: $Z(\alpha\ket{0}+\beta\ket{1})=\alpha\ket{0}-\beta\ket{1}$, and for this reason it is often called a phase flip operator.
Observe that
\begin{equation}
    XZ=\begin{pmatrix}
        0&1\\1&0
    \end{pmatrix}
    \begin{pmatrix}
        1&0\\0&-1
    \end{pmatrix}=\begin{pmatrix}
        0&-1\\1&0
    \end{pmatrix}=-iY,
\end{equation}
so that the Pauli-$Y$ operator can be thought of as being built from $X$ and $Z$ operators.

In addition to the Pauli operators, a commonly used quantum operation is the Hadamard gate given (in the computational basis) by
\begin{equation}
    H=\frac{1}{\sqrt{2}}\begin{pmatrix}
        1&1\\1&-1
    \end{pmatrix}.
\end{equation}
This operator is both self-adjoint and unitary and acts by transforming the computational basis vectors into a uniform superposition of themselves.
Indeed, we have
\begin{align}
    H\ket{0}&=\frac{1}{\sqrt{2}}(\ket{0}+\ket{1})\\
    H\ket{1}&=\frac{1}{\sqrt{2}}(\ket{0}-\ket{1}).
\end{align}
This property will prove useful to us as we design quantum circuits in the future.

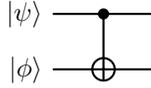
\begin{figure}
    \centering
    \begin{quantikz}
        \lstick{$\ket{\psi}$} & \ctrl{1} &\\
        \lstick{$\ket{\phi}$} & \targ{} &
    \end{quantikz}
    \caption{
        Circuit diagram for the $\cnot$ gate acting on the state $\ket{\psi}\ket{\phi}$. Here the control qubit is $\ket{\psi}$ and the target qubit is $\ket{\phi}$.
    }\label{fig:cnot}
\end{figure}

Each of the previously mentioned operators acts only on a single qubit.
In general, we will have reason to work with systems of multiple qubits, and so unitary operators acting on multiple qubits may be necessary.
To this end, we introduce first the so-called controlled-NOT or $\cnot$ operator, which operates on a two qubit system.
One of these qubits will be called the \textbf{control qubit} and the other will be called the \textbf{target qubit}.
Letting the left qubit be the control, the $\cnot$ operator acts on the four dimensional computational basis as follows:
\begin{align}
    \cnot\ket{00}&=\ket{00}\\
    \cnot\ket{01}&=\ket{01}\\
    \cnot\ket{10}&=\ket{11}\\
    \cnot\ket{11}&=\ket{10}.
\end{align}
When the control qubit is $\ket{0}$, the $\cnot$ operator acts trivially (does nothing), but when the control qubit is $\ket{1}$, the \texttt{CNOT} operator applys the $X$ operator to the target qubit.
This explains the notation; the $\cnot$ gate is a controlled version of the $X$ gate which is itself the analog of a \texttt{NOT} gate. 
The circuit notation for the $\cnot$ gate is shown in Figure~\ref{fig:cnot}. 
There is nothing special about the Pauli-$X$ gate showing up in the definition of
this operator. This construction can be extended to a controlled version of
any single qubit unitary operator $U$. When the control qubit is $\ket{0}$ the operator again does nothing,
but when the control qubit is $\ket{1}$, the operator now applies $U$
to the target qubit. The circuit notation for this more general gate is shown in Figure~\ref{fig:cU}. We use the special notation in Figure~\ref{fig:cZ} for a controlled-$Z$ gate, which the reader will notice is agnostic to which qubit is labeled the control and which is the target. This is because the controlled-$Z$ gate introduces a minus sign on the basis elements $\ket{01}$ and $\ket{10}$ and does nothing to the basis elements $\ket{00}$ and $\ket{11}$. Occasionally, we may wish to implement a controlled operation which triggers when the control qubit is in the state $\ket{0}$ rather than $\ket{1}$. This can be accomplished by putting a Pauli-$X$ gate before and after the control. Equivalently, we denote such an operation with an open circle instead of a black dot on the control as in Figure~\ref{fig:open-cU}.

\begin{figure}
    \centering
    \begin{quantikz}
        \lstick{$\ket{\psi}$} & \ctrl{1} &\\
        \lstick{$\ket{\phi}$} & \gate[1]{U} &
    \end{quantikz}
    \caption{
        Circuit diagram for the controlled $U$ gate acting on the state $\ket{\psi}\ket{\phi}$. Here the control qubit is $\ket{\psi}$ and the target qubit is $\ket{\phi}$.
    }\label{fig:cU}
\end{figure}
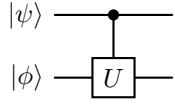

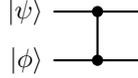
\begin{figure}
    \centering
    \begin{quantikz}
        \lstick{$\ket{\psi}$} & \ctrl{1} &\\
        \lstick{$\ket{\phi}$} & \ctrl{0} &
    \end{quantikz}
    \caption{
        Circuit diagram for the controlled-$Z$ gate acting on the state $\ket{\psi}\ket{\phi}$.
    }\label{fig:cZ}
\end{figure}

Another frequently used two qubit operator is the $\swapop$ gate, which operates by swapping the states of two qubits:
\begin{equation}
    \swapop(\ket{\psi}\ket{\phi})=\ket{\phi}\ket{\psi}.
\end{equation}
It can be shown that any permutation of $k$ things can be written as a composition
of transpositions (the swap permutation). Thus, on the Hilbert space level, we
can perform a permutation of arbitrarily many states by combining swap operators.
The circuit notation for the $\swapop$ gate is shown in Figure~\ref{fig:swap}.

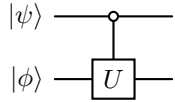
\begin{figure}
    \centering
    \begin{quantikz}
        \lstick{$\ket{\psi}$} & \ctrl[open]{1} &\\
        \lstick{$\ket{\phi}$} & \gate[1]{U} &
    \end{quantikz}
    \caption{
        Circuit diagram for the controlled $U$ gate acting on the state $\ket{\psi}\ket{\phi}$ which triggers on $\ket{0}$ instead of $\ket{1}$.
    }\label{fig:open-cU}
\end{figure}

\begin{figure}
    \centering
    \begin{quantikz}
        \lstick{$\ket{\psi}$} & \swap{1} &\\
        \lstick{$\ket{\phi}$} & \targX{} &
    \end{quantikz}
    \caption{
        Circuit diagram for the $\swapop$ gate acting on the state $\ket{\psi}\ket{\phi}$.
    }\label{fig:swap}
\end{figure}
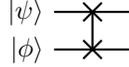

Usually, we will want to make a measurement at the end of a circuit to gain some information about the problem we are attempting to solve.
The notation for a measurement is a meter as shown in Figure~\ref{fig:measure}. 
Unless otherwise specified, all measurements are assumed to be taken with respect to the Pauli-$Z$ observable (i.e. with respect to the computational basis).

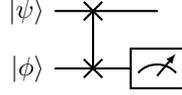
\begin{figure}
    \centering
    \begin{quantikz}
        \lstick{$\ket{\psi}$} & \swap{1} & \\
        \lstick{$\ket{\phi}$} & \targX{} & \meter{}
    \end{quantikz}
    \caption{
        Circuit diagram for the $\swapop$ gate acting on the state $\ket{\psi}\ket{\phi}$ with a measurement at the end. 
        Measurements are assumed to be in the $Z$ basis unless otherwise specified. Here we measure the second qubit.
    }\label{fig:measure}
\end{figure}

To get aquainted with the notation, let us examine a process known as \textbf{quantum teleportation}. 
Suppose Alice and Bob prepare an entangled pair of qubits in the Bell state $\ket{\Phi^+}=\frac{1}{\sqrt{2}}(\ket{00}+\ket{11})$.
They then depart, Alice holding onto one of the qubits, and Bob taking the other with him on a trip to the Andromeda galaxy.
At a later date, Alice obtains an additional qubit which is in the potentially unknown state $\ket{\psi}$.
Alice wants to send the qubit to Bob, but intergalactic shipping of qubits is expensive.
Fortunately, Alice is very clever. 
She realizes that if she makes a measurement of her half of the entangled pair she shares with Bob, it will affect the state of Bob's qubit,
and so if she entangles the qubit she wants to send to Bob with her half of the Bell state, she might be able to force Bob's qubit into the state $\ket{\psi}$,
effectively teleporting the state of her qubit to Bob. She proceeds according to the circuit in Figure~\ref{fig:teleportation}.

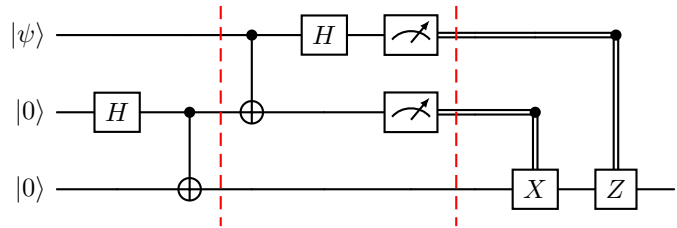
\begin{figure}
    \centering
    \begin{quantikz}
        \lstick{$\ket{\psi}$} &             & \slice{} & \ctrl{1} & \gate[1]{H} & \meter{}\slice{} & \setwiretype{c} &             & \ctrl[vertical wire=c]{2}\\
        \lstick{$\ket{0}$}    & \gate[1]{H} & \ctrl{1} & \targ{}  &             & \meter{} & \setwiretype{c} & \ctrl[vertical wire=c]{1} \\
        \lstick{$\ket{0}$}    &             & \targ{}  &          &             &          &                 & \gate[1]{X} & \gate[1]{Z} & 
    \end{quantikz}
    \caption{
        Quantum teleportation circuit. 
        The first portion of the procedure prepares the Bell state by applying a Hadamard and a $\cnot$ gate to the $\ket{00}$ state. 
        In the second portion of the circuit, Alice performs a $\cnot$ followed by a Hadamard gate between the qubit she wants to send and her half of the entangled pair.
        She then measures each of her qubits in the $Z$ basis.
        In the final portion of the circuit, Alice sends Bob two bits of information which he can then use to reconstruct the state $\ket{\psi}$.
    }\label{fig:teleportation}
\end{figure}

The first section of Figure~\ref{fig:teleportation} (the piece before the first dashed vertical line) shows the creation of the entangled pair shared by Alice and Bob.
Indeed, these two qubits start in the state $\ket{00}$ and are transformed to
\begin{align}
    \cnot (H\otimes \mathds{1})\ket{00}&=\cnot\frac{1}{\sqrt{2}}(\ket{0}+\ket{1})\ket{0}\\
    &=\frac{1}{\sqrt{2}}\cnot(\ket{00}+\ket{10})\\
    &=\frac{1}{\sqrt{2}}(\ket{00}+\ket{11}).
\end{align}
The middle section of the circuit shows the procedure performed by Alice to entangle the qubit that she wants to send to Bob with her half of the entangled pair she shares with Bob.
The state at the beginning of this section is $\frac{1}{\sqrt{2}}\ket{\psi}(\ket{00}+\ket{11})$.
We may write $\ket{\psi}=\alpha\ket{0}+\beta\ket{1}$ for some constants $\alpha,\beta\in\mathbb{C}$, so that the state becomes
$\frac{1}{\sqrt{2}}(\alpha\ket{000}+\alpha\ket{011}+\beta\ket{100}+\beta\ket{111})$.
Alice first applies a $\cnot$ gate controlled off of the qubit she wants to send to Bob and targeting the entangled qubit, producing
\begin{equation}
    \frac{1}{\sqrt{2}}(\alpha\ket{000}+\alpha\ket{011}+\beta\ket{110}+\beta\ket{101}).
\end{equation}
She then applies a Hadamard gate to the qubit she wants to send, which gives us the overall state
\begin{equation}
        \frac{1}{2}(\alpha(\ket{0}+\ket{1})\ket{00}+\alpha(\ket{0}+\ket{1})\ket{11}+\beta(\ket{0}-\ket{1})\ket{10}+\beta(\ket{0}-\ket{1})\ket{01}).
\end{equation}
Rearranging, we see that the state can be written
\begin{equation}
    \frac{1}{2}[\ket{00}(\alpha\ket{0}+\beta\ket{1})+\ket{01}(\beta\ket{0}+\alpha\ket{1})+\ket{10}(\alpha\ket{0}-\beta\ket{1})+\ket{11}(\alpha\ket{1}-\beta\ket{0})],
\end{equation}
so that when Alice subsequently makes a measurement of her two qubits in the computational basis, the state of the system collapses to one of the following with equal probability:
\begin{align}
    \ket{\phi_{00}}&=\ket{00}(\alpha\ket{0}+\beta\ket{1})\\
    \ket{\phi_{01}}&=\ket{01}(\beta\ket{0}+\alpha\ket{1})\\
    \ket{\phi_{10}}&=\ket{10}(\alpha\ket{0}-\beta\ket{1})\\
    \ket{\phi_{11}}&=\ket{11}(\alpha\ket{1}-\beta\ket{0}).
\end{align}
Notice that if the result of Alice's measurement is the state $\ket{\phi_{00}}$, then she will accomplish her goal of sending her qubit to Bob because Bob's qubit is now in the state $\alpha\ket{0}+\beta\ket{1}$.
Moreover, if Alice's measurement produces any other result, Bob's qubit can be corrected by applying $X$ and $Z$ gates until his qubit is in the state that Alice wanted to send.
Indeed, we have $\ket{\psi}=X(\beta\ket{0}+\alpha\ket{1})=Z(\alpha\ket{0}-\beta\ket{1})=ZX(\alpha\ket{1}-\beta\ket{0})$. 
Thus, to accomplish her goal, Alice sends the result of her measurement to Bob (two \textit{classical} bits of information) with instructions to apply an $X$ gate if the first bit (read from right to left) is a 1 and then to apply a $Z$ gate if the second bit is a 1,
and this is what the final piece of the circuit depicts. The double lines indicate a classical bit of information. To be abundantly clear, when Alice makes her measurements, she obtains one of the eigenvalues $\pm1$.
She then records this information in the form of two bits by identifying the $+1$ measurement with the bit state 0 and the $-1$ measurement with the bit state $1$ (hence the labeling of the $\ket{\phi_{ij}}$). In other words, Alice labels the bits with the same label as the eigenvector the state collapses into, rather than the eigenvalues. This is a convention we adopt throughout the remainder of this work.

This procedure illustrates the remarkable fact that a qubit of information can in some sense be teleported across a vast distance at the cost of two classical bits of information and some pre-established entanglement between shared qubits.
This first cost is the reason that this procedure cannot be used to communicate faster than the speed of light. 
It is not until Bob receives the result of Alice's measurement that he can correct his state to the one Alice wished to send,
and this transmission of classical information is indeed limited by the speed of light according to Einstein's well-tested theory of relativity.

In Figure~\ref{fig:teleportation}, let us see what happens if we move the controlled 
operations occuring after Alice's measurements to just before Alice's measurements as shown in Figure~\ref{fig:teleportation2}.
Physically, this would involve Alice applying an operation involving her qubits as well as Bob's qubit.
We have already seen that the state after the second part of the circuit is
\begin{equation}
    \frac{1}{2}[\ket{00}(\alpha\ket{0}+\beta\ket{1})+\ket{01}(\beta\ket{0}+\alpha\ket{1})+\ket{10}(\alpha\ket{0}-\beta\ket{1})+\ket{11}(\alpha\ket{1}-\beta\ket{0})].
\end{equation}
Now applying the controlled $X$ operation in the third part of the circuit, this state becomes
\begin{equation}
    \frac{1}{2}[\ket{00}(\alpha\ket{0}+\beta\ket{1})+\ket{01}(\beta\ket{1}+\alpha\ket{0})+\ket{10}(\alpha\ket{0}-\beta\ket{1})+\ket{11}(\alpha\ket{0}-\beta\ket{1})],
\end{equation}
and the subsequent application of the controlled $Z$ operation produces
\begin{equation}
    \frac{1}{2}[\ket{00}(\alpha\ket{0}+\beta\ket{1})+\ket{01}(\alpha\ket{0}+\beta\ket{1})+\ket{10}(\alpha\ket{0}+\beta\ket{1})+\ket{11}(\alpha\ket{0}+\beta\ket{1})].
\end{equation}
Thus, when Alice makes her measurements, no matter which result she gets, she can be assured that Bob's qubit is in the state $\ket{\psi}$ as she intended.
Of course, applying the controlled operations before her measurement would probably
require her to be in the vicinity of Bob, and so she still elects to perform the original procedure
to avoid traveling to the Andromeda galaxy. In principle, however, these two circuits are equivalent.
In fact, this phenomenon is not specific to the teleportation problem, it is true in general and called the \textbf{Principle of Deferred Measurement}:
When the result of a measurement is used as a classical conditional operation in a later portion of the circuit,
the classical conditional operation can be replaced by a controlled quantum operation
before the measurement takes place. This principle allows us to move all of our measurement
operations to the end of the circuit.

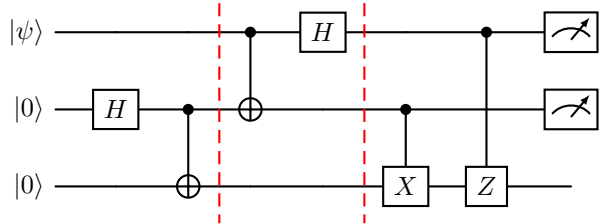
\begin{figure}
    \centering
    \begin{quantikz}
        \lstick{$\ket{\psi}$} &             & \slice{} & \ctrl{1} & \gate[1]{H}\slice{} &                           &\ctrl{2} & \meter{}\\
        \lstick{$\ket{0}$}    & \gate[1]{H} & \ctrl{1} & \targ{}  &             & \ctrl{1}                  & & \meter{} \\
        \lstick{$\ket{0}$}    &             & \targ{}  &          &             &  \gate[1]{X}   &  \gate[1]{Z}   &   
    \end{quantikz}
    \caption{
        Quantum teleportation circuit with classically conditioned operations replaced by controlled operations inside circuit.
    }\label{fig:teleportation2}
\end{figure}

\subsection{Distinguishing Quantum States}

It is often the case that we want to understand how closely related two quantum states are. A natural way to measure this is to take the distance metric \eqref{eq:distance-metric} which is induced by the inner product equipped to our Hilbert space and observe that states which are a distance zero from each other satisfy
\begin{equation}
    \sqrt{(\bra{\psi}-\bra{\phi})(\ket{\psi}-\ket{\phi})}=0,
\end{equation}
or equivalently,
\begin{equation}\label{eq:real-part-1}
    \real(\braket{\phi\vert\psi})=1.
\end{equation}
Recalling that $1\ge\lvert\braket{\phi\vert\psi}\rvert^2=(\real(\braket{\phi\vert\psi}))^2+(\imag(\braket{\phi\vert\psi}))^2$, we see that $\imag(\braket{\phi\vert\psi})$ must vanish, so that \eqref{eq:real-part-1} is equivalent to
\begin{equation}\label{eq:nonreal-part-1}
    \braket{\phi\vert\psi}=1.
\end{equation}
But recall that we only care about quantum states up to a global phase since these global phases have no measurable effects on the system. Thus, $\ket{\phi}$ and $\ket{\psi}$ could equivalently be represented by $e^{i\varphi}\ket{\phi}$ and $e^{i\theta}\ket{\psi}$, respectively, and this transforms $\braket{\phi\vert\psi}\mapsto e^{i(\theta-\varphi)}\braket{\phi\vert\psi}$. It follows that this function is not invariant under global phase transformations, and so we should search elsewhere for a measure of similarity between two quantum states. A natural candidate that solves the global phase invariance issue while simultaneously having a physical interpretation is the \textbf{fidelity}
\begin{equation}
    \mathcal{F}(\ket{\phi},\ket{\psi})=\lvert\braket{\phi\vert\psi}\rvert^2.
\end{equation}
The presence of the modulus leaves the fidelity invariant under global phase transformations, and we note that if $\ket{\phi}$ is an eigenvector for some observable, then the fidelity is precisely the probability of obtaining the eigenvalue associated to $\ket{\phi}$ when a measurement of the state $\ket{\psi}$ is made with respect to this observable.

This choice of fidelity satisfies several natural properties. Firstly, we have
\begin{equation}
    0\le\mathcal{F}(\ket{\phi},\ket{\psi})\le1
\end{equation}
with equality to 1 if and only if $\ket{\phi}=\ket{\psi}$ up to a global phase. This property ensures that the fidelity is indeed a measure of similarity between the two quantum states. Although it is not a distance metric, it does give rise to one. The second property is the invariance under a swap of the arguments. That is,
\begin{equation}
    \mathcal{F}(\ket{\phi},\ket{\psi})=\mathcal{F}(\ket{\psi},\ket{\phi}),
\end{equation}
a natural assumption for a quantity that should measure the similarity between two states. Finally, this quantity is invariant under not just global phase transformations but also any simultaneous unitary transformation of the arguments
\begin{equation}
    \mathcal{F}(U\ket{\phi},U\ket{\psi})=\lvert\braket{\phi\vert U^\dagger U\vert\psi}\rvert^2=\lvert\braket{\phi\vert\psi}\rvert^2=\mathcal{F}(\ket{\phi},\ket{\psi}).
\end{equation}
When we consider that such a transformation can be seen as a change of basis, this becomes a very natural assumption.

In the event that the state of the system is not entirely known, we can describe it as an ensemble of pure states $\mathcal{E}=\{(p_1,\ket{\psi_1}),\ldots,(p_n,\ket{\psi_n})\}$ which lists all of the possible states $\ket{\psi_i}$ that the system could be in along with the probability $p_i$ that the system is in this state. Given an ensemble of pure states, we can always extend our Hilbert space $\mathcal{H}$ to a larger one $\mathcal{H}_E$ such that the state of the system can be described by a single $\ket{\psi}\in\mathcal{H}_E$. The state $\ket{\psi}$ is called a \textbf{purification} of the ensemble. We then recover the original system by measuring $\ket{\psi}$ in a basis for the orthogonal complement $\mathcal{H}^\bot$, which collapses the system into the state $\ket{\psi_i}$ with probability $p_i$. Note that there are many purifications for a given ensemble as the following example shows.

\begin{example}Consider the ensemble $\mathcal{E}=\{(1/2,\ket{0}),(1/2,\ket{1})\}$ describing a system with Hilbert space $\mathcal{H}=\mathbb{C}^2$. By extending $\mathcal{H}$ to $\mathcal{H}_E=\mathbb{C}^4$ and writing $\ket{\psi}=\frac{1}{\sqrt{2}}(\ket{00}+\ket{11})$, we obtain a purification of the ensemble $\mathcal{E}$. Indeed, a measurement of the second qubit will force the original system into the state $\ket{0}$ or $\ket{1}$, each with probability $1/2$. Similarly, we could instead write $\ket{\psi}=\frac{1}{\sqrt{2}}(\ket{01}+\ket{10})$, and a measurement of the second qubit again reproduces the ensemble $\mathcal{E}$.
\end{example}

Now, given two ensembles $\mathcal{E}_\psi$ and $\mathcal{E}_\phi$, we may wish to know how closely related they are. At the level of the Hilbert space $\mathcal{H}$, we do not have single states $\ket{\psi}$ and $\ket{\phi}$ with which to calculate a fidelity. A solution to this problem is given by Jozsa \cite{jozsa1994} and Uhlmann \cite{uhlmann1976}, whose combined work produces the fidelity
\begin{equation}
    \mathcal{F}(\mathcal{E}_\psi,\mathcal{E}_\phi)=\max\ \lvert\braket{\psi\vert\phi}\rvert^2,
\end{equation}
where the maximum is taken over all purifications of the ensembles. That is, we compute the fidelity between two ensembles by taking the maximum fidelity between all purifications of the ensembles $\mathcal{F}(\mathcal{E}_\psi,\mathcal{E}_\phi)=\max\ \mathcal{F}(\ket{\psi},\ket{\phi})$.

As an aside, there is a lesser known generalization of these fidelities to more than two states. Consider a collection of $k$ states $\ket{\psi_1},\ldots,\ket{\psi_k}$ (not an ensemble, just a list of states that we wish to compare). Occasionally, one may wish to know how closely related these $k$ states are \textit{collectively}. Naturally, we should ask such a collective fidelity to satisfy the properties
\begin{equation}
    0\le\mathcal{F}(\ket{\psi_1},\ldots,\ket{\psi_k})\le1
\end{equation}
with equality to 1 if and only if $\ket{\psi_1}=\cdots=\ket{\psi_k}$ (up to global phases),
\begin{equation}
    \mathcal{F}(\ket{\psi_1},\ldots,\ket{\psi_k})=\mathcal{F}(\ket{\psi_{\sigma(1)}},\ldots,\ket{\psi_{\sigma(k)}})
\end{equation}
for any permutation $\sigma$ of the arguments, and the unitary invariance property
\begin{equation}
    \mathcal{F}(U\ket{\psi_1},\ldots,U\ket{\psi_k})=\mathcal{F}(\ket{\psi_1},\ldots,\ket{\psi_k}).
\end{equation}
Interestingly, the invariance under all permutations of the arguments can be relaxed to invariance under the action of a permutation group $G$ on $k$ symbols (we give an introduction to groups in Section~\ref{sec:groups} for those unfamiliar), but this possibly weaker assumption only guarantees that arrangements of the arguments produce the same fidelity when they can be connected by a permutation belonging to $G$. An example of a collective fidelity of this type is given by the quantity
\begin{equation}
    \mathcal{F}(\ket{\psi_1},\ldots,\ket{\psi_k})=\frac{1}{\lvert G\rvert}\sum_{\sigma\in G}\braket{\psi_1\vert\psi_{\sigma(1)}}\cdots\braket{\psi_k\vert\psi_{\sigma(k)}},
\end{equation}
which appears as the acceptance probability for a $G$-Bose symmetry test as outlined in \cite{laborde2023,bradshaw2023,barenco1997}, allowing for the calculation of this fidelity on a quantum device.

For $k=2$, we recover an expression containing the fidelity $\mathcal{F}(\ket{\phi},\ket{\psi})=\lvert\braket{\phi\vert\psi}\rvert^2$, which can be computed as the expected value of the well-known swap test as we now show. Consider the circuit in Figure~\ref{fig:swap-test}. We first apply a Hadamard gate to an ancillary qubit producing the state
\begin{equation}
    \frac{1}{\sqrt{2}}(\ket{0}+\ket{1})\ket{\phi}\ket{\psi},
\end{equation}
which is transformed by the controlled swap gate to
\begin{equation}
    \frac{1}{\sqrt{2}}(\ket{0}\ket{\phi}\ket{\psi}+\ket{1}\ket{\psi}\ket{\phi}).
\end{equation}
Now applying the final Hadamard gate produces the state
\begin{equation}
    \ket{0}\frac{1}{2}(\ket{\phi}\ket{\psi}+\ket{\psi}\ket{\phi})+\ket{1}\frac{1}{2}(\ket{\phi}\ket{\psi}-\ket{\psi}\ket{\phi}).
\end{equation}
Thus, the expected value for a measurement of the ancillary qubit in the Pauli-$Z$ basis is
\begin{equation}
    \langle Z\rangle = \frac{1}{4}(\bra{\phi}\bra{\psi}+\bra{\psi}\bra{\phi})(\ket{\phi}\ket{\psi}+\ket{\psi}\ket{\phi})-\frac{1}{4}(\bra{\phi}\bra{\psi}-\bra{\psi}\bra{\phi})(\ket{\phi}\ket{\psi}-\ket{\psi}\ket{\phi})=\lvert\braket{\phi\vert\psi}\rvert^2.
\end{equation}
We therefore have a circuit which computes the fidelity between two quantum states. This circuit is an example of a Hadamard test, a class which we will see arise again in the syndrome extraction procedure of the stabilizer code formalism.

\begin{figure}
    \centering
    \begin{quantikz}
        \lstick{$\ket{0}$}       & \gate[1]{H} & \ctrl{1} & \gate[1]{H} & \meter{}\\
        \lstick{$\ket{\phi}$}    &             & \swap{1} &             &  \\
        \lstick{$\ket{\psi}$}    &             & \swap{0} &             &  
    \end{quantikz}
    \caption{
        Circuit for the swap test to compute the fidelity between two states.
    }\label{fig:swap-test}
\end{figure}
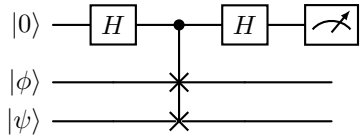

\section{Basic Quantum Error Correcting Codes} \label{sec:basiccodes}

In Section~\ref{sec:circuit}, we discussed the circuit model of quantum computation
in which the quantum computer is described as a closed system free of any noise.
The unitary transformations made in each of our quantum circuits were assumed
to be without error when acting on our qubits and our measurements were assumed
to be made perfectly. It will come at no surprise to the reader that 
the reality of the situation is in sharp contrast to this point of view; 
quantum gates are noisy and quantum systems are hard to fully conceal from the 
outside world, so that interactions with the environment are a source of noise. 
It is probably a fair assumption to make that there are in fact no true closed 
systems in the universe apart from the universe itself.

For all of these reasons, we need a method to handle errors which arise
in our quantum system. Of course, error correction has been studied already in
the context of classical computers, and so we might guess that the procedures
outlined in the classical regime can be translated to the quantum world. While
there is sometimes a correspondence between a classical error correcting code and
a quantum one \cite{calderbank1996,steane1996multiple}, we quickly run into problems even for the simple repetition codes which safeguard the information content against errors by creating redundancy through its duplication. In the quantum regime, it turns out that such a duplication process cannot exist for an arbitrary quantum state, a fact which has been termed the \textbf{no-cloning theorem}. However, we will see that an analog of the repetition code
can be constructed without cloning. Instead, we create redundancy in the quantum
system by appending additional qubits, which we call \textbf{redundancy qubits},
and entangling the state we want to preserve with these qubits. This creates a
new logical state in a higher dimensional Hilbert space which contains redundant
information about the original quantum state but does not clone it. With the encoded state 
containing redundant information, we are able to construct simple codes which detect and 
even correct errors on our qubits.

Before diving into our first quantum code, let us discuss the types of errors which can occur
on a qubit. We will first address the most common form of error in a highly closed off system,
a unitary error $E\in\mathcal{U}(\mathbb{C}^2)$. Under such an error, a state $\ket{\psi}\in\mathbb{C}^2$ 
is transformed by sending $\ket{\psi}\mapsto E\ket{\psi}$. Of course, there are
infinitely many such errors since $\mathcal{U}(\mathbb{C}^2)$ is infinite in cardinality. This is
in contrast to the classical error which can only take the form of a bit flip. It may therefore
seem that the problem of correcting errors on quantum states is hopeless. However, not all is lost.
Notice that the Pauli operators together with the identity span the space $\mathcal{U}(\mathbb{C}^2)$. 
Thus, every unitary error can be written
\begin{equation}
    E=a\mathds{1}+bX+cY+dZ,
\end{equation}
for some $a,b,c,d\in\mathbb{C}$, and we would therefore be well advised to start by correcting errors that
take the form of Pauli matrices. Notice, moreover, that $XZ=-iY$, so that a $Y$-error is nothing more
than a $Z$-error followed by an $X$-error (up to a phase). We will therefore begin by correcting the latter two errors.
Since $X\ket{0}=\ket{1}$ and $X\ket{1}=\ket{0}$, we call such an error a \textbf{bit flip error} in
analogy to the type of error which occurs in classical computing. The error caused by a $Z$ operation
transforms an arbitrary qubit as $\alpha\ket{0}+\beta\ket{1}\mapsto\alpha\ket{0}-\beta\ket{1}$, 
thereby changing the relative phase by a sign. For this reason, we call such an error a \textbf{phase
flip error}.

In this section, we examine some of the basic known quantum error correcting codes
without much regard for the underlying mathematical formalism that operates in the background.
In Section~\ref{sec:classical-rep}, we review the classical repetition codes mentioned
above and discuss their success rates. In Section~\ref{sec:hamming-bound}, we discuss the Hamming
bound and its relationship to more advanced classical codes. We then prove the no-cloning theorem in 
Section~\ref{sec:redundancy} and provide an alternative means of creating redundancy.
In Section~\ref{sec:2qubit}, we use this alternative method of
creating redundancy to create our first quantum code which detects bit flip errors, and then we
proceed to construct 3-qubit error correcting codes for both bit and phase flip errors in
Sections~\ref{sec:3qubitcode} and \ref{sec:3qubitphasecode}. In Section~\ref{sec:shorcode}, 
we show that these latter two codes can be concatenated, forming a 9-qubit code known as the Shor 
code, named for its discoverer, Peter Shor. This code corrects for both bit and phase flip errors, and 
was one of the first complete codes to be discovered. We wrap up this section by briefly discussing the quantum Hamming bound in Section~\ref{sec:quantum-hamming-bound}.

\subsection{The Classical Repetition Codes}\label{sec:classical-rep}

Let us consider the most basic classical error correcting code: the repetition code of length $n$. Repetition codes represent one of the earliest and most intuitive approaches to error correction, tracing back to the era of telegraphy, where redundancy was a natural way to mitigate noise in primitive communication systems \cite{shannon1962mathematical}.
To encode a bit $b \in \{0, 1\}$, we simply repeat it $n$ times before sending, which we denote
$b \cdots b$ ($n$ times). These encoded states are known as \textbf{codewords}.
To decode, the receiver chooses the most commonly occurring bit as the presumed original message, a strategy which is called \textbf{majority voting}.
Codes are often identified by a triple $[n,k,d]$ consisting of the total number $n$ of \textbf{physical bits} sent over the
communication channel, the number $k$ of \textbf{logical bits} that are encoded into the $n$ total physical bits, and
the \textbf{distance} $d$ of the code, the minimum number of bits that must be changed to transform one codeword to another i.e. the minimum \textbf{hamming distance} between codewords.
The repetition code of length $n$ is therefore an $[n, 1, n]$ code, as it encodes $k=1$ logical qubits into $n$ physical bits
with codewords $0\cdots0$ and $1\cdots1$, which are a distance $d=n$ apart.
We now provide an example to elucidate this process.

\begin{example}
Suppose Alice wants to send the message $m=010$ to Bob over a noisy channel using the repetition
code of length $n=3$.
She first encodes each symbol in the message by repeating it $3$ times to form the codeword
$c = 000 \ 111 \ 000$.
Now suppose Bob receives the message $c' = 000 \ 101 \ 011$.
For each block, Bob assumes the original message to be the most commonly occurring bit, and thus he
decodes the message to be $m' = 011$.
The first block of $c'$ was sent without error, and thus Bob has no issue decoding it.
The second block experienced one bit flip error, but Bob could still decode it correctly as there
were more ones than zeros.
The third block, however, experienced two bit flip errors which caused Bob to incorrectly decode
the block to a one.
\end{example}

From the example above, we can see that for an $n=3$ repetition code, only one error can be
successfully corrected; two or more errors will lead to an incorrect decoded message.
In general, a (binary) repetition code of length $n$ can correct $t=\lfloor\frac{d-1}{2}\rfloor$
errors, where $d=n$ is the distance of the code.
This can be intuitively interpreted as having fewer than half of the bits in the message corrupted.
We generally only consider repetition codes for odd $n$, as even $n$ have the potential for ties in
the number of received $0$ and $1$ bits, in which case we could detect that an error has occurred
but not correct it.

The astute reader may be wondering whether there exists a more efficient code than the simple repetition code, to which the answer is
a resounding yes. Among these are the Reed-Solomon codes~\cite{reed1960polynomial} which find use in CDs
and DVDs as well as QR codes.
The only thing repetition codes really have going for them is the ease with which they can be learned and implemented.
To quantify the amount of redundancy built into a code, we use a concept called \textbf{rate} that is defined by
\begin{equation}
    R = \frac{k}{n},
\end{equation}
where $k$ is the number of message bits being encoded and $n$ is the total number of bits sent over
the communication channel.
For repetition codes, this is always $1/n$, which is very bad considering you need $n \ge 3$ to do
any error correction.
Reed-Solomon codes, for comparison, typically have a rate of $1/2$.

While it is perfectly permissible to think about the repetition code as a means of protecting the information contained in a classical bit which can then be decoded by the receiver to recover the original information content, we have chosen to call the encoded sequence of bits a logical bit because the repetition code can be viewed as a means of constructing ``bits'' that are intrinsically more robust to noise. By encoding a single logical bit across multiple physical bits and interpreting the collective state as our fundamental unit of information, the repetition code creates a more robust information carrier that can tolerate a certain number of errors. This abstraction is crucial in the design of fault-tolerant devices, where operations are performed not on raw physical bits, which are highly error-prone, but on these encoded logical bits. In such settings, logical gates are engineered to act on encoded data, with error correction interwoven into the computational process to ensure reliable output despite imperfect hardware. Thus, the repetition code is more than an error-correcting scheme; it exemplifies how redundancy can be leveraged to construct higher-level, error-resilient building blocks for reliable computation.

This point is worth formalizing. Suppose our physical bits are independently prone to experience a bit flip error with probability $p$ as it is transmitted through a noisy channel. By implementing the repetition code of length $n$, we are able to correct errors with up to $t=\lfloor\frac{n-1}{2}\rfloor$ bit flips. Thus, the probability that a logical bit is transmitted successfully is
\begin{equation}
    p_{success} = \sum_{j=0}^t\binom{n}{j}p^j(1-p)^{n-j}.
\end{equation}
So, the \textbf{logical error rate}, the probability that an error occurs on the logical bit, is $1-p_{success}$. Thus, the logical error rate is an improvement from the physical error rate when $1-p_{success}<p$, which is equivalent to $1-p<p_{success}$. When $n=2m+1$ is odd, this condition is
\begin{equation}
    1-p<\sum_{j=0}^m\binom{2m+1}{j}p^j(1-p)^{2m+1-j}.
\end{equation}
When $p=\frac12$, the binomial distribution on the right becomes
\begin{equation}
    \frac{1}{2^{2m+1}}\sum_{j=0}^m\binom{2m+1}{j}=\frac{2^{2m}}{2^{2m+1}}=\frac12.
\end{equation}
Now consider what happens as $p$ decreases below $1/2$. The binomial distribution becomes \emph{asymmetric}, with its mass shifting toward smaller values of $m$. In other words, the distribution is skewed to the right, meaning the bulk of its probability mass lies on the left side of the distribution.
This observation establishes the inequality
\begin{equation}
1 < \sum_{j=0}^m \binom{2m+1}{j} p^j (1-p)^{2m+1-j},
\end{equation}
for $p<\frac12$. Thus, as long as $p<\frac12$, the logical error rate is smaller than the original error rate of the individual physical bits. For this reason, we call this value of $p$ the \textbf{error threshold}.

The primary challenge in designing efficient codes is in the optimization of rate versus distance.
Ideally, our code would have both a large distance and rate, in which case it corrects a large number
of errors and does so without much overhead. Unfortunately, various results have shown that this
cannot happen. There is always a trade-off between distance and rate so that high distance codes necessarily
have low rate and vice-versa. Thus, the code designer's objective is to maximize distance without too much
overhead with the understanding that an increase in distance will likely force an increase in overhead,
or a decrease in the rate.

\subsection{The Hamming Bound}\label{sec:hamming-bound}

One of the most fundamental limitations in classical coding theory is the \textbf{Hamming bound}, which places an upper limit on the number of codewords a code can have, given its length and error-correcting capability. Understanding this bound also sets the stage for the quantum analogue we will encounter later. 

Let $q$ denote the size of the alphabet used in some code $\mathcal{C}$. For example, binary codes have $q = 2$. Let $\mathcal{A}_q$ be the $q$-ary alphabet (a set of order $q$, it doesn't really matter what symbols you use), and let $n$ be the length of each codeword. Suppose $\mathcal{C} \subseteq \mathcal{A}_q^n$ is a code with minimum distance $d$, meaning that any two distinct codewords differ in at least $d$ positions. Here,
\begin{equation}
    \mathcal{A}_q^n = \{(a_1,\ldots,a_n):a_i\in \mathcal{A}_q\}
\end{equation}
denotes the direct product of $\mathcal{A}_q$ with itself $n$ times. Under maximum likelihood decoding, a received word is mapped to the codeword that is closest in Hamming distance. The decoder can therefore correct up to
\begin{equation} \label{eq:hamming_t}
    t = \left\lfloor \frac{d - 1}{2} \right\rfloor
\end{equation}
errors, since this is the largest number of symbol errors that can be corrected unambiguously. If more than $t$ errors occur, the received word may lie closer to a different codeword, resulting in a decoding error. Thus, each codeword defines a \emph{Hamming ball} of radius $t$ in $\mathcal{A}_q^n$, within which all received messages are decoded to that codeword (see Figure~\ref{fig:sphere_packing}).

\begin{figure}
    \centering
    \includegraphics[width=3in]{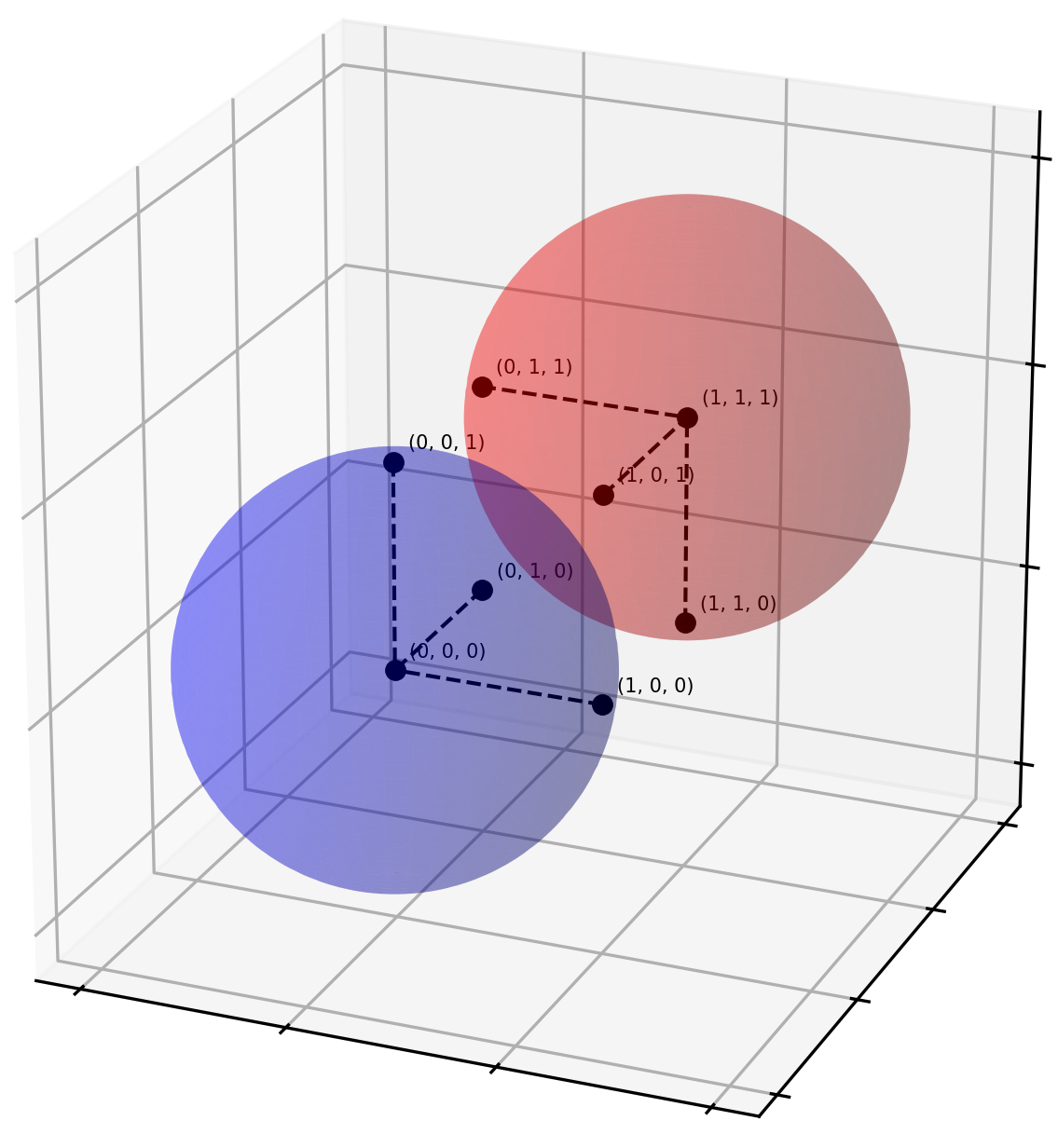}
    \caption{
        Illustration of how error correcting code design can be framed as a
        problem of optimal sphere packing.
        All received messages within the blue sphere will be corrected to a zero
        and those within the red sphere will be corrected to a one.
        Since we are operating not in $\mathbb{R}^3$ but rather $\mathbb{Z}_2^3$,
        these ``spheres'' are actually the set of all codewords within Hamming
        distance one of the center (each set of connected black dots).
    }
    \label{fig:sphere_packing}
\end{figure}

To count how many received words are decoded into a single codeword, observe that for each $k \le t$, there are $\binom{n}{k}$ ways to choose $k$ positions for errors, and for each such position, there are $(q - 1)$ incorrect symbols to substitute. Therefore, the number of $n$-letter words within Hamming distance $t$ of a given codeword is
\begin{equation} \label{eq:hamming_denom}
    \sum_{k=0}^{t} \binom{n}{k} (q - 1)^k.
\end{equation}

If a code has $M$ codewords, and the Hamming balls of radius $t$ around each codeword are disjoint (i.e., decoding regions do not overlap), then these balls must collectively fit inside $\mathcal{A}_q^n$, which contains $q^n$ total strings. Hence, we obtain the \textbf{Hamming bound}:
\begin{equation} \label{eq:hamming_bound}
    A_q(n, d) \le \frac{q^n}{\sum_{k=0}^{t} \binom{n}{k}(q - 1)^k},
\end{equation}
where $A_q(n,d)$ denotes the maximum number of codewords in a $q$-ary code of length $n$ and minimum distance $d$.

A code that saturates this inequality is called a \textbf{perfect code}. In such codes, the Hamming balls of radius $t$ centered at each codeword exactly partition the ambient space $\mathcal{A}_q^n$. Every possible received word lies within distance $t$ of exactly one codeword. While such codes are rare, notable examples include binary repetition codes with odd length, the family of binary Hamming codes \cite{hamming1950,macwilliams1977}, and the Golay codes over binary and ternary alphabets \cite{golay1949,golay1961}. Perfect codes represent the optimal trade-off between redundancy and error correction, but they exist only for very specific parameter ($n,k,d$) choices. In general, most practical codes fall short of saturating the Hamming bound but are designed to approach it as closely as possible.

Later, we will see a quantum analogue of this result known as the \textbf{quantum Hamming bound}, which similarly limits the parameters of quantum error-correcting codes. Although the logic is more subtle due to quantum phenomena such as superposition and entanglement, the classical Hamming bound provides essential conceptual groundwork.

\subsection{Redundancy in Quantum Systems}\label{sec:redundancy}

We have seen that error correction of classical bits can be accomplished by simply
copying the bit we wish to communicate and sending all copies to the receiver
with the understanding that they will make a majority vote guess at the true 
value of this bit. In this way, redundancy is added to the system that has the 
effect of decreasing the probability of receiving an erroneous bit. The more 
redundancy is added, the better the odds of success are so long as the probability
of an error is below a certain threshold.

We may therefore guess that a similar approach can be taken in quantum information
theory; perhaps we can copy the state of one qubit onto another qubit and then send
them both to a receiver who does something like a majority vote. This is problematic
for one big reason; suppose that the state we want to copy isn't known.
A measurement of this state would likely destroy it, giving us only limited tomographic
information in the process. If the state can be prepared repeatedly, the problem
is solved, but we would like a procedure that works for \emph{any} quantum state.
Thus, we are in need of a unitary transformation $U$ acting on two qubits which copies
the arbitrary state of the first qubit onto the second qubit. That is, we would like
$U$ to satisfy
\begin{equation}\label{eq:clone}
    U(\ket{\psi}\ket{\phi})=\ket{\psi}\ket{\psi}
\end{equation}
for every $\ket{\psi}\in\mathcal{H}$ and for some $\ket{\phi}\in\mathcal{H}$. The
following theorem tells us that such an operation does not exist.

\begin{theorem}[No-cloning]
    There does not exist a unitary operator $U$ satisfying \eqref{eq:clone} for all
    $\ket{\psi}\in\mathcal{H}$ and some $\ket{\phi}\in\mathcal{H}$.
\end{theorem}
\begin{proof}
    Suppose such an operator $U$ and such a state $\ket{\phi}$ exist and recall that unitary operators satisfy $U^\dagger U=\mathds{1}$. 
    We are free to choose $\ket{\psi}$, so let $\ket{\psi_1}$ and $\ket{\psi_2}$ be distinct non-orthogonal states. Then we have
    \begin{align}
        \braket{\psi_1\vert\psi_2}&=\braket{\psi_1\vert\psi_2}\braket{\phi\vert\phi}\\
        &=\bra{\psi_1}\bra{\phi}\cdot\ket{\psi_2}\ket{\phi}\\
        &=\bra{\psi_1}\bra{\phi}U^\dagger U\ket{\psi_2}\ket{\phi}\\
        &=(U\ket{\psi_1}\ket{\phi})^\dagger(U\ket{\psi_2}\ket{\phi})\\
        &=(\ket{\psi_1}\ket{\psi_1})^\dagger\ket{\psi_2}\ket{\psi_2}\\
        &=\bra{\psi_1}\bra{\psi_1}\cdot\ket{\psi_2}\ket{\psi_2}\\
        &=\braket{\psi_1\vert\psi_2}^2.
    \end{align}
    Thus, $\braket{\psi_1\vert\psi_2}$ is either 0 or 1. If the inner product is zero, then the states are orthogonal, a contradiction. If the inner product is one, then the states differ only by a global phase and are therefore not distinct, a contradiction. Thus, the operator $U$ does not exist.
\end{proof}

The no-cloning theorem tells us that if we want to create redundancy in our quantum state,
we will have to consider some other means of doing so. Let us now show how to accomplish this
by entangling the input state with additional qubits to create a logical state in a 
higher dimensional Hilbert space which contains the information of the original state 
but spread out between the additional qubits.

Suppose Alice has a qubit in the state $\ket{\psi}=\alpha\ket{0}+\beta\ket{1}$ that she wants to send to Bob over a noisy channel.
To do so safely, she encodes her qubit into a higher dimensional Hilbert space called the \textbf{code space}.
With a single additional qubit, Alice might encode her state as in Figure~\ref{fig:cnot-encode}. 
By appending a $\ket{0}$ to her state and applying the $\cnot$ gate, she is making the following identification
\begin{equation}
    \alpha\ket{0}+\beta\ket{1}\mapsto\alpha\ket{00}+\beta\ket{11}.
\end{equation}
We see that the coefficients defining $\ket{\psi}$ are still present in the state,
but now that the state is part of a larger Hilbert space, it might be possible to
make a measurement which tells us whether an error occurred without destroying this 
information. As we show in the next section, this is indeed the case. In the context of fault-tolerance, the redundancy is added to create a logical qubit which is more resistant to error than the individual physical qubits it is made from, thereby allowing for quantum computations to be performed in a fault-tolerant manner.

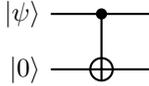
\begin{figure}
    \centering
    \begin{quantikz}
        \lstick{$\ket{\psi}$} & \ctrl{1} & \\
        \lstick{$\ket{0}$} & \targ{} &
    \end{quantikz}
    \caption{
        Circuit diagram for sample encoding procedure with one additional qubit.
    }\label{fig:cnot-encode}
\end{figure}
\subsection{The 2-Qubit Code}
\label{sec:2qubit}

We begin our first practical dive into QEC with the 2-qubit detection code,
which introduces redundancy into a quantum state via entanglement and then
takes a measurement of an entangled ancillary system in order to deduce whether
a single Pauli $X$ error has occurred.
Of key importance is that this measurement does not collapse the
state of the logical qubit, and thus the qubits being sent are still usable by the receiver.
This code is unable to correct errors, but rest assured we will introduce a
number of codes that can later.

As we have already learned, we cannot simply introduce redundancy into a
quantum state by repetition (as in the classical repetition code) without
violating the no-cloning theorem. We must therefore find another 
way to add redundancy to $\ket{\psi}$ before transmission. Consider again the simple 
quantum circuit in Figure~\ref{fig:cnot-encode}, which applies the \texttt{CNOT} operator controlled by the
state $\ket{\psi}$ and targeting another qubit in the $\ket{0}$ state.
Let $\ket{\psi} = \alpha\ket{0} + \beta\ket{1}$. We will denote the encoded state by $\ket{\psi}_L$ and call it the \textbf{logical state} to distinguish it from the state $\ket{\psi}$ being represented. We observed previously that
the resulting logical state of the system is
\begin{equation} \label{eq:2qubit:redundancy}
    \ket{\psi}_L = \alpha\ket{00} + \beta\ket{11}.
\end{equation}
The above circuit that transforms the input state to its corresponding logical state
is called the \textbf{encoder}.

We draw a few key observations from \eqref{eq:2qubit:redundancy}.
First, we have not cloned $\ket{\psi}$; there is no way to recover two copies
of $\ket{\psi}$ from \eqref{eq:2qubit:redundancy}.
Second, there is now a source of redundancy that can be exploited for error
detection.
If the received state has components in either the $\ket{01}$ or $\ket{10}$
basis states, we know that the state was affected by noise during transmission.
Unfortunately, this code provides no way to tell which qubit was affected.
The only course of action for the receiver in this case is to request that the
sender prepare another quantum state and send it again, in the hopes that both qubits are in agreement this time.

The question remains of how to determine whether the physical qubits are in
agreement.
Simply measuring each qubit is problematic for two reasons:
    (1) directly measuring either qubit collapses the state and no
        longer allows the receiver to take advantage of the quantum properties
        of the system, and
    (2) measurement on either would yield a classical bit of $0$ with
        probability $\lvert\alpha\rvert^2$ and $1$ with probability $\lvert\beta\rvert^2$, 
        thus we still wouldn't know with certainty whether an error in transmission 
        had occurred.
We want to obtain this information while maintaining the quantum
properties of the system.

What we want is to take a single measurement which produces a bit
that distinguishes between when both qubits in $\ket{\psi}_L$ are the same and when they are different.
This bit is called an \textbf{error syndrome} and the process of obtaining it is called \textbf{syndrome extraction}. The process is analogous to having a doctor diagnose an ailment, hence the terminology.
Computer scientists likely recognize that the \texttt{XOR} operation would
accomplish this task, albeit with quantum arguments.
In motivating a quantum circuit that does an \texttt{XOR}-like operation, we
invoke Maslow's Theorem, which states that when the only tool you have is a
hammer, everything starts to look like a nail. That's right, we're going to use
\texttt{CNOT} gates again.

Let's think about what \texttt{CNOT} gates do, and how we might be able to use
them to create an \texttt{XOR} operation.
If $\ket{m}$ and $\ket{n}$ are single qubit computational basis states, then the \texttt{CNOT} 
operation controlled on qubit $\ket{m}$ and applied to $\ket{n}$ gives us the following:
\begin{equation}
    \texttt{CNOT} \ket{m} \ket{n} = \ket{m} \ket{n \oplus m}
\end{equation}
where $\oplus$ is addition modulo two (equivalent to \texttt{XOR}).
So, effectively, what a \texttt{CNOT} gate is doing in this case is transforming
$\ket{n} \mapsto \ket{n\oplus m}$.
This is close to what we want, but it changes the state of one of its arguments.
Leveraging an ancillary qubit to store the result of the \texttt{CNOT}
operation, we obtain the circuit in Figure~\ref{fig:3.4:xor:circuit}.

\begin{figure}
    \centering

    \begin{subfigure}{0.4\textwidth}
        \centering
        \begin{quantikz}
            \lstick{$\ket{m}$} & \ctrl{2} &          &          & \\
            \lstick{$\ket{n}$} &          & \ctrl{1} &          & \\
            \lstick{$\ket{0}_A$}    & \targ{}  & \targ{}  & \meter{} & \setwiretype{c}
        \end{quantikz}
        \vspace{0.25em}
        \caption{Circuit diagram.}
        \label{fig:3.4:xor:circuit}
        \vspace{1.15em}
    \end{subfigure}
    \begin{subfigure}{0.55\textwidth}
        \centering
        \begin{tabular}{ccccc} \toprule
            $\ket{m}$      & $\ket{n}$      & First \texttt{CNOT} & Second \texttt{CNOT} & Measurement \\ \midrule
            $\ket{0}$      & $\ket{0}$      & $\ket{0}$           & $\ket{0}$            & 0           \\
            $\ket{0}$      & $\ket{1}$      & $\ket{0}$           & $\ket{1}$            & 1           \\
            $\ket{1}$      & $\ket{0}$      & $\ket{1}$           & $\ket{1}$            & 1           \\
            $\ket{1}$      & $\ket{1}$      & $\ket{1}$           & $\ket{0}$            & 0           \\ \bottomrule
        \end{tabular}
        \caption{``Truth table'' showing measurement obtained for each computational basis state.}
        \label{fig:3.4:xor:table}
    \end{subfigure}
    
    \caption{Using \texttt{CNOT} gates to produce behaviour similar to the classical \texttt{XOR} operator.}
    \label{fig:3.4:xor}
    
\end{figure}

Consider what happens to the ancillary qubit $\ket{0}_A$ when $\ket{m}$ and 
    $\ket{n}$ are in each of the basis states $\ket{0}$ and $\ket{1}$ using the table in
    Figure~\ref{fig:3.4:xor:table}.
We list the choice of the computational basis vector in the first two columns, the
    state of the ancillary qubit after each \texttt{CNOT} gate in the middle two columns,
    and the result of a measurement of the ancillary qubit in the final column.
Behold, the measurement column is identical to the truth table for the classical
\texttt{XOR} operation!
Moreover, we can show that if the computational basis states $\ket{m}\ket{n}$ are
replaced by the logical qubit $\ket{\psi}_L$, this state is undisturbed by the
measurement of the ancillary qubit.
Indeed, with $\ket{\psi}_L$ as in \eqref{eq:2qubit:redundancy}, the application
of the \texttt{CNOT} operations produces the state
\begin{equation}
    \alpha\ket{000}+\beta\ket{110}=(\alpha\ket{00}+\beta\ket{11})\ket{0},
\end{equation}
which is the same as doing nothing at all.
Thus, a measurement of the ancillary qubit will produce the 0 outcome and the
logical qubit remains undisturbed.
Let us see what happens when an error occurs before the \texttt{XOR} analog is applied.

Putting together the encoder and the syndrome extraction procedure, the final
circuit for the 2-qubit detection code is shown in Figure~\ref{fig:2qubitcode}.
First, we form the logical state $\ket{\psi}_L$ by adding redundancy to
$\ket{\psi}$ via a controlled-not operation.
Then, the error operator $E \in \{\mathds{1}, X_1, X_2, X_1 X_2\}$ potentially
degrades $\ket{\psi}_L$.
Finally, we perform syndrome extraction by measuring the ancillary qubit
$\ket{0}_A$, understanding that if we obtain the outcome 1, then an error has occurred.
To see this, observe that after the encoding, the error transforms the logical qubit
in one of the following ways:
\begin{align}
    \alpha\ket{00}+\beta\ket{11}&\stackrel{X_1}{\longmapsto} \alpha\ket{10}+\beta\ket{01}\\
    \alpha\ket{00}+\beta\ket{11}&\stackrel{X_2}{\longmapsto} \alpha\ket{01}+\beta\ket{10}\\
    \alpha\ket{00}+\beta\ket{11}&\stackrel{X_1X_2}{\longmapsto} \alpha\ket{11}+\beta\ket{00}.
\end{align}
In the first case, the \texttt{CNOT} gates in the syndrome extraction produce
the state
\begin{equation}
    \alpha\ket{101}+\beta\ket{011}=(\alpha\ket{10}+\beta\ket{01})\ket{1},
\end{equation}
which produces the outcome 1 upon measurement of the ancillary qubit. Similarly,
the second case produces the outcome 1 upon measurement. However, notice that in
the last case, which corresponds to a bit flip error on both physical qubits, the
state after the \texttt{CNOT} gates in the syndrome extraction is
\begin{equation}
    (\alpha\ket{11}+\beta\ket{00})\ket{0},
\end{equation} 
so that a measurement of the ancillary qubit produces the same value that we
get when no error occurs. Thus, a simultaneous bit flip error on both physical
qubits will go undetected by our simple code.
% Table~\ref{tab:2qubitprobability} enumerates all possible ways in which the modeled error channel can act on sent qubits. We assume that Pauli-$X$ gates act independently on each qubit with probabilities $p_1$ and $p_2$, respectively.

\begin{figure}

    \begin{subfigure}{0.55\textwidth}
        \centering
        \tikzset{
            noisy/.style={cloud,fill=white,draw=gray,line width=2pt,inner xsep=-4pt,inner ysep=-5pt,cloud puffs=10,aspect=0.5} 
        }
        \begin{quantikz}
            \lstick{$\ket{\psi}$} & \ctrl{1} &[0.3cm] \gate[2,style={noisy},label style=black]{\text{Noise}} & & \ctrl{2} & & & & \\
            \lstick{$\ket{0}$} & \targ{} & & & & \ctrl{1} & & & \\
            \setwiretype{n} & & & \lstick{$\ket{0}_A$} & \targ{} \setwiretype{q} & \targ{}  & \meter{} & \setwiretype{c} & 
        \end{quantikz}
        \caption{Circuit Diagram}
        \vspace{5em}
    \end{subfigure}
    \hspace{1em}
    \begin{subfigure}{0.4\textwidth}
    \centering
        \begin{lstlisting}[language=C++]
OPENQASM 3.0;
include "stdgates.inc";

input angle theta;

// Initialize qubits
qubit[2] rand;
qubit[2] psi;
qubit S;

// Entangle psi0 and psi1
cx psi[0], psi[1];

// Simulate noise channel
h rand;
ry(theta) rand;
cx rand, psi;

// Extract error syndrome
cx psi[0], S;
cx psi[1], S;\end{lstlisting}
        \caption{OpenQASM implementation}
    \end{subfigure}

    \vspace{1em}
    
    \caption{
        The two-qubit error detection circuit used to detect, but not correct, a single Pauli $X$ error when transmitting a quantum state $\ket{\psi}$ via an ancilla qubit $\ket{0}_A$.
    }
    \label{fig:2qubitcode}

\end{figure}

% {'00000': 3074, '01011': 327, '10101': 351, '00011': 57, '00101': 25, '01010': 9, '00001': 58, '11110': 38, '01110': 5, '00010': 31, '00110': 15, '00111': 8, '10000': 6, '01001': 19, '00100': 22, '01000': 18, '10100': 4, '01100': 2, '10110': 5, '01111': 7, '10001': 5, '10111': 3, '11101': 1, '11001': 1, '01101': 1, '10010': 1, '11100': 1, '10011': 2}
\begin{table}
    \centering
    \begin{tabular}{ccccccl}
        \multicolumn{2}{c}{$\ket{\psi}_L$}\vspace{0.1em} & $S$ & Count & $P_{obs}$ & $P_{true}$ & Analysis                 \\ \toprule
            0 &     0 &   0 &  3074 &     0.750 &      0.810 & No transmission errors; syndrome of 0.                      \\
            1 &     0 &   1 &   351 &     0.086 &      0.090 & Bit flip error on first physical qubit, error detected.     \\
            0 &     1 &   1 &   327 &     0.080 &      0.090 & Bit flip error on second physical qubit, error detected.    \\
            1 &     1 &   0 &    38 &     0.009 &      0.010 & Bit flip error on both physical qubits, error not detected. \\
            \multicolumn{3}{r}{Other:}
                            &   306 &     0.075 &      0.000 & Errors occurred outside of the communication channel.
    \end{tabular}
    
    \caption{
        Empirical probabilities of each combination of $X_1$ and $X_2$ error.
        Values were obtained by running the circuit in Figure~\ref{fig:2qubitcode} on IBM Torino for 4096 shots.
    }
    \label{tab:2qubit}
\end{table}

Perhaps the most important metric of this code is the probability that the
syndrome will be one given that an error has occurred, i.e., the
probability that an error is detectable.
We have
\begin{align*}
    P(S=1 | E \ne \mathds{1})   &= \frac{P(S=1 \cap E \ne \mathds{1})}{P(E \ne \mathds{1})} \\
                                &= \frac{1}{P(E \ne \mathds{1})}
                                   \sum_{\hat{E} \ne \mathds{1}}
                                        P(S=1 | E = \hat{E}) P(E = \hat{E})\\
                                &= \frac{1\cdot P(E = X_1)+1\cdot P(E = X_2)+0\cdot P(E = X_1X_2)}{P(E \ne \mathds{1})} ,                                       
\end{align*}
where we have used the facts that the probability that $S=1$ given that $\hat{E}=X_1$ or $X_2$ is 1,
and the probability that $S=1$ given that $\hat{E}=X_1X_2=1$ is 0 (see Table~\ref{tab:2qubit}). 
Together with the probabilities $P(E=\hat{E})$ that can be read off of Table~\ref{tab:2qubit},
and the fact that $P(E\ne \mathds{1})=1-P(E=\mathds{1})=1-(1-p_1)(1-p_2)=p_1+p_2-p_1p_2$,
we have
\begin{align}
    P(S=1 | E \ne \mathds{1})&=\frac{(1-p_1)p_2+p_1(1-p_2)}{p_1+p_2-p_1p_2}\\
    &=\frac{p_1+p_2-2p_1p_2}{p_1+p_2-p_1p_2}\\
    &= 1 - \frac{p_1p_2}{p_1 + p_2 - p_1p_2}. \label{eq:2qubit:pdetectable}
\end{align}
If $p_1=p_2=1$, then both errors occur and $P(S=1\vert E\ne\mathds{1})=0$ as expected. If $p_1=1$, then
the error on the first qubit occurs with certainty and $P(S=1\vert E\ne\mathds{1})=1-p_2$, which says
that the probability of detecting an error given that one occurred is equivalent to the probability that
an error does not occur on the second qubit. This is indeed what we would expect since the error will go
undetected if the second qubit also experiences a bit flip. Similarly, if the bit flip on the second qubit occurs
with certainty, then $P(S=1\vert E\ne\mathds{1})=1-p_1$.

Just as for classical codes, we categorize quantum codes according to the number $n$ of physical qubits, the number $k$ of logical qubits, and the distance $d$, which we define as the minimum weight of a Pauli word that transforms one code word to another. Here, the weight $w(A)$ of an operator $A$ is the number of non-identity Pauli operators in $A$. To distinguish a quantum code from a classical code, we write $[[n,k,d]]$ instead of $[n,k,d]$. We also define the rate of the code to be $R=k/n$ just as in the classical case. Thus, the two qubit code has rate $R=1/2$ but distance $d=2$ since it takes two Pauli-$X$ gates to get from $\ket{00}$ to $\ket{11}$. A quantum code with distance $d$ can detect errors with weight up to $d-1$ and correct errors with weight up to $\lfloor\frac{d-1}{2}\rfloor$, and this agrees with our analysis of the 2-qubit code which detects only single qubit bit flips and corrects nothing. We say that the 2-qubit code is a $[[2,1,2]]$ code.

\subsection{The 3-Qubit Bit Flip Code}\label{sec:3qubitcode}

For the purpose of this code, we will again assume that only bit flip errors are capable
of occuring, choosing to take into account phase flip errors in the next section. 
Thus, we assume that for an arbitrary qubit with quantum state 
$\ket{\psi}=\alpha\ket{0}+\beta\ket{1}$ transmitted over our communication 
channel, an error $E=X$ occurs with probability $p$, leaving us with the state 
$X\ket{\psi}=\beta\ket{0}+\alpha\ket{1}$. We will further assume that when 
multiple qubits are transmitted, the probability of an error occurring on each 
qubit is independent of the other qubits. This is a safe assumption if we send
each qubit independently through the channel.

We must first describe the encoder taking us from $\ket{\psi}$ to a logical state
$\ket{\psi}_L$ made of three physical qubits. We append two qubits in the $\ket{0}$
state to the redundancy register and apply a $\cnot$ gate targeting each of the
redundancy qubits, both controlled off of the data qubit as shown in 
Figure~\ref{fig:3qubitbit} just before the noise channel. Thus, the data state $\ket{\psi}$ is transformed to
the logical state
\begin{equation}
    \ket{\psi}_L:=\alpha\ket{000}+\beta\ket{111},
\end{equation}
which we note is not the cloned state $\ket{\psi}\ket{\psi}\ket{\psi}$. There are three physical qubits forming this logical state, and we see that three Pauli-$X$ operations are needed to transform $\ket{000}$ to $\ket{111}$, so this is a $[[3,1,3]]$ code with rate $R=1/3$.

When a bit flip error occurs, the basis vectors spanning the logical qubit state
are no longer the all 0 state or the all 1 state. Just as in the two qubit code, 
we want to capitalize on this observation by designing our syndrome extraction operation
so that it checks whether the basis vectors appearing in the state of the logical
qubit satisfy this constraint using $\cnot$ gates targeting an ancillary register.
To this end, we append two ancillary qubits and apply four $\cnot$ gates as shown
in the syndrome extraction portion of Figure~\ref{fig:3qubitbit}. If the reader cannot see the motivation behind this choice of syndrome extraction, they can rest assured that motivation will be given when we discuss stabilizer codes in Section~\ref{sec:stabilizer}. The current section is more about developing intuition for the general structure of a quantum code. 

\begin{figure}[t]

    \begin{subfigure}{\textwidth}
        \centering
        \tikzset{
            noisy/.style={cloud,fill=white,draw=gray,line width=2pt,inner xsep=-4pt,inner ysep=-5pt,cloud puffs=10,aspect=0.5} 
        }
        \begin{quantikz}[column sep=0.4cm]
            \lstick{$\ket{\psi}$} & \ctrl{1} & \ctrl{2} &[0.3cm] \gate[3,style={noisy},label style=black]{\text{Noise}}
                                                          &[0.5cm] \gategroup[5,steps=5,style={dotted}]{Syndrome Extraction}
                                                                                       & \ctrl{3}                &          & \ctrl{4} &          & & \gategroup[5,steps=7,style={dotted}]{Error Correction}
                                                                                                                                                      & \targ{}   &          &           &          &           &          &[0.1cm]   & \\
            \lstick{$\ket{0}$}    & \targ{}  &          & &                            &                         & \ctrl{2} &          &          & & &           &          & \targ{}   &          &           &          &          & \\
            \lstick{$\ket{0}$}    &          & \targ{}  & &                            &                         &          &          & \ctrl{2} & & &           &          &           &          & \targ{}   &          &          & \\
            \setwiretype{n}       &          &          & & \lstick{$\ket{0}_A \quad$} & \targ{} \setwiretype{q} & \targ{}  &          &          & & & \ctrl{-3} &          & \ctrl{-2} & \gate{X} & \ctrl{-1} & \gate{X} & \meter{} & \\
            \setwiretype{n}       &          &          & & \lstick{$\ket{0}_A \quad$} & \setwiretype{q}         &          & \targ{}  & \targ{}  & & & \ctrl{-4} & \gate{X} & \ctrl{-3} & \gate{X} & \ctrl{-2} &          & \meter{} &
        \end{quantikz}
        \caption{Circuit diagram}
        \vspace{5em}
    \end{subfigure}
    
    \vspace{-4em}

    \begin{subfigure}{\textwidth}
    \centering
        \begin{minipage}{0.4\textwidth}
            \begin{lstlisting}[language=C++]
OPENQASM 3.0;
include "stdgates.inc"

// Initialization
qubit[3] psi;
qubit[2] S;

// Encoder
cx psi[0], psi[1];
cx psi[0], psi[2];

// [Psi subject to noise]
            \end{lstlisting}
            \vspace{1em}
        \end{minipage}
        \hspace{2em}
        \begin{minipage}{0.4\textwidth}
            \begin{lstlisting}[language=C++,firstnumber=14]
// Syndrome extraction
cx psi[0], S[0];
cx psi[1], S[0];
cx psi[0], S[1];
cx psi[2], S[1];

// Error correction
ctrl(2) @ x S[0], S[1], psi[0];
x S[1];
ctrl(2) @ x S[0], S[1], psi[1];
x S[0];
x S[1];
ctrl(2) @ x S[0], S[1], psi[2];
x S[0];\end{lstlisting}
        \end{minipage}
        \caption{OpenQASM implementation}
    \end{subfigure}
    
    \caption{
        The three-qubit error correction code used to detect Pauli $X$ errors when transmitting a quantum state $\ket{\psi}$.
    }
    \label{fig:3qubitbit}

\end{figure}

\begin{table}
    \centering
    \begin{tabular}{ccccccl}
        \toprule
        $\ket{\psi}_L$ & $\ket{\hat{\psi}}_L$ & $S$ & Count & $P_{obs}$ & $P_{true}$ & Analysis \\[0.1cm] 
                                                                                       \midrule \\[-0.4cm]
        % =======================================================================================================
        $\ket{000}$ & $\ket{000}$ & 00 & 2271 & 0.554 \hspace{0.4em} & 0.729 & No transmission errors. \\[0.15cm]
        % =======================================================================================================
        $\ket{001}$ & $\ket{000}$ & 01 &  233 &
                                            \multirow{3}{*}{$\left.\begin{array}{l}
                                                0.057 \\
                                                0.060 \\
                                                0.057 \\
                                            \end{array}\right\rbrace$} & \multirow{3}*{0.081} & \multirow{3}*{One error; detectable and correctable.} \\
        $\ket{010}$ & $\ket{000}$ & 10 &  244 & & & \\
        $\ket{100}$ & $\ket{000}$ & 11 &  233 & & & \\[0.15cm]
        % =======================================================================================================
        $\ket{011}$ & $\ket{111}$ & 11 &   15 & 
                                            \multirow{3}{*}{$\left.\begin{array}{l}
                                                0.004 \\
                                                0.005 \\
                                                0.003 \\
                                            \end{array}\right\rbrace$} & \multirow{3}*{0.009} & \multirow{3}*{Two errors; detectable but not correctable.} \\
        $\ket{101}$ & $\ket{111}$ & 10 &   21 & & & \\
        $\ket{110}$ & $\ket{111}$ & 01 &   11 & & & \\[0.15cm]
        % =======================================================================================================
        $\ket{111}$ & $\ket{111}$ & 00 &    3 & 0.001\hspace{0.7em} & 0.001 & Three errors, neither detectable nor correctable. \\[0.15cm]
        % =======================================================================================================
        \multicolumn{3}{r}{Other:}     & 1065 & 0.260\hspace{0.7em} & 0.000 & Errors occured outside of noise channel. \\ \bottomrule
    \end{tabular}
    
    \caption{
        Empirical probabilities of each combination of $X_1$, $X_2$, and $X_3$ error in the 3-qubit bit flip code.
        Values were obtained by running the circuit in Figure~\ref{fig:3qubitbit} on IBM Torino for 4096 shots. Here we have used a simple error model described in Appendix~\ref{sec:appendix:impl:noise}.
    }
    \label{tab:3qubit_bit}
\end{table}

If no error occurs, then the result of applying
the $\cnot$ gates is the state
\begin{equation}
    \alpha\ket{00000}+\beta\ket{11100}=\ket{\psi}_L\ket{00},
\end{equation}
and so a measurement of the final two qubits will return the 00 outcome with certainty,
and the state of the logical qubit will be undisturbed.
Let us see what happens if an error on a single qubit occurs, starting with the
error $E=X_1$. The state before the syndrome extractor is then
\begin{equation}
    X_1\ket{\psi}_L\ket{00}=\alpha\ket{10000}+\beta\ket{01100},
\end{equation}
and after the $\cnot$ gates are applied, this state becomes
\begin{equation}
    \alpha\ket{10011}+\beta\ket{01111}=X_1\ket{\psi}_L\ket{11}.
\end{equation}
Thus, a measurement of the ancillary qubits produces the syndrome 11 with certainty,
and the state of the logical qubit is undisturbed by the measurement. Similarly, 
if instead the error is $E=X_2$, then the state of the system just before the
syndrome extraction is
\begin{equation}
    X_2\ket{\psi}_L\ket{00}=\alpha\ket{01000}+\beta\ket{10100},
\end{equation}
and after the $\cnot$ gates are applied, this state becomes
\begin{equation}
    \alpha\ket{01010}+\beta\ket{10110}=X_2\ket{\psi}_L\ket{10}.
\end{equation}
Thus, a measurement of the ancillary qubits produces the syndrome 10 with certainty,
and the state of the logical qubit is yet again undisturbed by this measurement.
Moreover, the fact that this syndrome is different than the syndrome obtained with
the $X_1$ error allows us to distinguish between these two errors. Finally, if the
error is given by $E=X_3$, then the state just before syndrome extraction is
\begin{equation}
    X_3\ket{\psi}_L\ket{00}=\alpha\ket{00100}+\beta\ket{11000},
\end{equation}
and after the $\cnot$ gates are applied, this state becomes
\begin{equation}
    \alpha\ket{00101}+\beta\ket{11001}=X_3\ket{\psi}_L\ket{01}.
\end{equation}
Thus, a measurement of the ancillary qubits produces the syndrome 01 with certainty,
and the state of the logical qubit is again undisturbed by this measurement. Moreover,
all three single qubit errors are detected with unique syndromes, allowing us to
distinguish between each case. Now that we can detect single qubit errors, we can
correct for them by applying the inverse of the error operator (in this case itself).
If we measure the 00 syndrome, we do nothing. If we measure the 11, 10, or 01 syndromes,
we apply the $X_1$, $X_2$, and $X_3$ gates, respectively, transforming the logical
qubit back to the state before the error occurred. The process of undoing the errors
using the measured syndrome is called \textbf{decoding}. This process is depicted in Figure~\ref{fig:3qubitbit}, where we have applied the principle of deferred measurement to move all measurements to the end. We will continue to do so throughout the remainder of this work.

Notice that there are only four possible syndromes that can be extracted from two
qubits, and we have used all of them. Thus, if there are two or more simultaneous
bit flip errors on the logical state, they cannot be distinguished from the single
qubit errors. Take, for example, the error $E=X_1X_2$, for which the state just
before syndrome extraction is
\begin{equation}
    X_1X_2\ket{\psi}_L\ket{00}=\alpha\ket{11000}+\beta\ket{00100}.
\end{equation}
After the $\cnot$ gates in the syndrome extractor are applied, the state of the
system becomes
\begin{equation}
    \alpha\ket{11001}+\beta\ket{00101}=X_1X_2\ket{\psi}_L\ket{01},
\end{equation}
and so a measurement of the ancillary qubits produces the syndrome 01, which matches
the syndrome produced by the $X_3$ error. Similarly, the $X_1X_3$ error produces
the 10 syndrome like the $X_2$ error and the $X_2X_3$ error produces the 11 syndrome
like the $X_1$ error. Finally, if all three bit flips occur, then we find that the
00 syndrome is produced, and so we won't detect the error at all. 

Of course, the probability that two or more errors occur is smaller than the
probability that a single error occurs, and so we expect that if we assume that
all detected errors are single qubit errors and then correct them as discussed earlier, 
we will find that the majority of the time, the error has indeed been corrected. 
Let us formalize this discussion.

For a single qubit, let $p$ denote the probability that a bit flip occurs. Then
the probability that exactly one bit flip occurs in the logical state is 
\begin{equation}
    \sum_{i=1}^3P(E=X_i)=3p(1-p)^2,
\end{equation} 
and the probability that no error occurs on any qubit is 
\begin{equation}
    P(E=\mathds{1})=(1-p)^3.
\end{equation} 
Thus, if we assume
all errors are single bit flip errors when correcting, the total probability of successfully correcting 
the true error is 
\begin{equation}
    p_{success}=(1-p)^3+3p(1-p)^2=1-3p^2+2p^3.
\end{equation} Notice
that as our channel becomes less noisy, the probability of success approaches 1
as we would expect. If we would like our code to correct an error more
than $50\%$ of the time, we will need $p<1/2$, and if we want our code to correct
an error more than $99\%$ of the time, we will need $p$ less than $\approx0.0589$.
In any case, the original physical error rate without correction was $p$ and the error rate with correction is 
\begin{equation}
    1-p_{success}=3p^2-2p^3.
\end{equation} The difference
is therefore $p-3p^2+2p^3$, which is positive when $p<\frac12$. Thus, the 3 qubit code offers a reduction
in the error rate whenever the original physical error rate satisfies $p<\frac12$. That is, the threshold for this code is $p_{th}=\frac12$.

So far, we have assumed that the only possible error is a product of unitary
bit flip errors. If the environment interacts with our logical qubit, the nature
of this interaction is likely not described by a unitary transformation on the
three qubit Hilbert space. Let us take as an example the operation
\begin{equation}
    (1-\epsilon)\mathds{1}+\epsilon X,
\end{equation}
which might arise in this way, and suppose that it acts on the first qubit. The state
after the error is therefore
\begin{equation}
    ((1-\epsilon)\mathds{1}+\epsilon X_1)\ket{\psi}_L\ket{00}
    =(1-\epsilon)\ket{\psi}_L\ket{00}+\epsilon(\alpha\ket{100}+\beta\ket{011})\ket{00},
\end{equation}
and after applying the $\cnot$ operations in the syndrome extractor, this state becomes
\begin{equation}
    (1-\epsilon)\ket{\psi}_L\ket{00}+\epsilon(\alpha\ket{100}+\beta\ket{011})\ket{11}.
\end{equation}
Thus, upon measurement, we will obtain either the syndrome 00 with probability
$P(00)=(1-\epsilon)^2/(1-2\epsilon+2\epsilon^2)$ or the syndrome 11 with probability
$P(11)=\epsilon^2/(1-2\epsilon+2\epsilon^2)$. In the first case, the state of the
logical qubit after the measurement is $\ket{\psi}_L$ and so no error has occurred.
In the second case, the state of the system becomes $X_1\ket{\psi}_L$, and the $X_1$ 
error can be corrected. Thus, the syndrome extraction forces the state of the
system to choose whether or not an error has occurred so that the same decoder that
we described in the case of the unitary error can be applied here if necessary.

In fact, this phenomenon holds more generally. If an interaction with the environment occurs,
causing an error of the form
\begin{equation}\label{eq:continuous-error}
    c_0I+c_1X_1+c_2X_2+c_3X_3,
\end{equation}
then the state after the error is
\begin{equation}
    c_0\ket{\psi}_L\ket{00}+c_1X_1\ket{\psi}_L\ket{00}+c_2X_2\ket{\psi}_L\ket{00}+c_3X_3\ket{\psi}_L\ket{00},
\end{equation}
and the result of applying the $\cnot$ gates in the syndrome extraction procedure is the state
\begin{equation}
    c_0\ket{\psi}_L\ket{00}+c_1X_1\ket{\psi}_L\ket{11}+c_2X_2\ket{\psi}_L\ket{10}+c_3X_3\ket{\psi}_L\ket{01}.
\end{equation}
Thus, a measurement of the ancillary qubits produces the syndrome 00 with probability $\lvert c_0\rvert^2/\|\Vec{c}\|^2$,
where $\|\Vec{c}\|^2=\sum_{i=0}^3\lvert c_i\rvert^2$. Similarly, the syndromes 11, 10, and 01 are produced with
probabilities $\lvert c_1\rvert^2/\|\Vec{c}\|^2$, $\lvert c_2\rvert^2/\|\Vec{c}\|^2$, and $\lvert c_3\rvert^2/\|\Vec{c}\|^2$,
respectively. In each case, the state of the system collapses to the single qubit error state corresponding to each syndrome.
That is, the act of measuring forces the system to choose which error ``really'' occurred, and therefore transforms
the continuous spectrum of errors defined by \eqref{eq:continuous-error} into a discrete one. For this reason, we call this phenomenon
the \textbf{discretization of errors}. A better model of the environment would also take into account two and three qubit bit flips in the sum, but the same thing happens; the syndrome extraction forces the system to choose one error from those listed in the sum.

\subsection{The 3-Qubit Phase Flip Code}\label{sec:3qubitphasecode}

Just as we constructed a 3-qubit code which corrects for single qubit bit flip errors
in Section~\ref{sec:3qubitcode}, we can similarly construct a 3-qubit code which
corrects for single qubit phase flip errors. Thus, we assume that for an arbitrary
qubit with quantum state $\ket{\psi}=\alpha\ket{0}+\beta\ket{1}$ transmitted over
our communication channel, an error $E=Z$ occurs with probability $p$, leaving us
with the state $Z\ket{\psi}=\alpha\ket{0}-\beta\ket{1}$. Just as with the bit flip
case, we make the assumption that multiple qubits experience errors independently
of one another.

The encoder sending $\ket{\psi}$ to the logical state $\ket{\psi'}_L$\footnote{
    We will need to distinguish the state after the encoding in the 3-qubit bit
    flip code from the state after the encoding in the 3-qubit phase flip code,
    hence our choice of $\ket{\psi'}_L$ rather than simply $\ket{\psi}_L$.
} made of three
physical qubits is constructed as follows. We first append two qubits in the $\ket{0}$
state to the redundancy register. Next, we apply two $\cnot$ gates targeting
the redundancy qubits, both controlled off of the data qubits just as we did for the
bit flip error code in Figure~\ref{fig:3qubitbit}. This transforms the state 
$\ket{\psi}=\alpha\ket{0}+\beta\ket{1}$ into the state $\ket{\psi}_L=\alpha\ket{000}+\beta\ket{111}$.
Next, we apply a Hadamard gate to all three qubits, producing the logical state
\begin{equation}
    \ket{\psi'}_L:=\alpha\ket{+++}+\beta\ket{---},
\end{equation}
where we recall that $\ket{+}=\frac{1}{\sqrt{2}}(\ket{0}+\ket{1})$ and 
$\ket{-}=\frac{1}{\sqrt{2}}(\ket{0}-\ket{1})$. The circuit for this encoder is
shown just before the noise channel in Figure~\ref{fig:3qubitphase}. Since we need three Pauli-$Z$ operations to map $\ket{+++}$ to $\ket{---}$, the distance of this code is three. Thus, like the 3-qubit bit flip code, this is a $[[3,1,3]]$ code with rate $R=1/3$.

Notice how a phase flip acts on the $\ket{+}$ and $\ket{-}$ states: we have
$Z\ket{+}=\ket{-}$ and $Z\ket{-}=\ket{+}$. Thus, in the $\{\ket{+},\ket{-}\}$
basis, the Pauli-$Z$ operator acts like the Pauli-$X$ operator does in the
$\{\ket{0},\ket{1}\}$ basis. In other words, a phase flip error looks like a bit
flip error after Hadamard gates have been applied. Indeed, we have $HXH=Z$, and
this allows us to switch between the two pictures. In our encoder, we started with 
the original bit flip encoder and then changed pictures by applying the Hadamard 
gates to each qubit. Similarly, we can reuse the bit flip syndrome extraction procedure
in Figure~\ref{fig:3qubitbit}. With the knowledge that after the encoding,
we are in the phase flip picture, we first transform back to the bit flip picture by
applying Hadamard gates, we then apply the bit flip syndrome extraction procedure,
and finally transform back to the phase flip picture by applying Hadamard gates.
This syndrome extraction procedure is shown in Figure~\ref{fig:3qubitphase}.

\begin{figure}[t]

    \begin{subfigure}{\textwidth}
        \centering
        \tikzset{
            noisy/.style={cloud,fill=white,draw=gray,line width=2pt,inner xsep=-4pt,inner ysep=-5pt,cloud puffs=10,aspect=0.5} 
        }
        \begin{quantikz}[column sep=0.3cm]
            \lstick{$\ket{\psi}$} & \ctrl{1} & \ctrl{2} & \gate{H} &[0.4cm] \gate[3,style={noisy},label style=black]{\text{Noise}}
                                                                   &[0.6cm] \gate{H} \gategroup[5,steps=6,style={dotted}]{Syndrome Extraction}
                                                                   & \ctrl{3} & & \ctrl{4} & &  \gate{H} & & \ctrl{3} \gategroup[5,steps=7,style={dotted}]{Error Correction} &
                                                                                                 & & & & &  &[0.1cm] & \\
            \lstick{$\ket{0}$}    & \targ{} & & \gate{H} & & \gate{H} & & \ctrl{2} & & &  \gate{H} & & & & \ctrl{2} & & & &  & &  \\
            \lstick{$\ket{0}$}    & & \targ{} & \gate{H} & & \gate{H} & & & & \ctrl{2} &  \gate{H} & & & & & & \ctrl{1} & &  & & \\
            \setwiretype{n}       & & & & & \lstick{$\ket{0}_A \ \ \quad$} & \targ{} \setwiretype{q} & \targ{} & & & & & \ctrl{-3} & & \ctrl{-2} & \gate{X} & \ctrl{-1} & \gate{X} & & \meter{} & \\
            \setwiretype{n}       & & & & & \lstick{$\ket{0}_A \ \ \quad$} & \setwiretype{q} & & \targ{} & \targ{} & & & \ctrl{-4} & \gate{X} & \ctrl{-3} & \gate{X} & \ctrl{-2} & & & \meter{} &
        \end{quantikz}
        \caption{Circuit diagram}
        \vspace{5em}
    \end{subfigure}
    
    \vspace{-4em}

    \begin{subfigure}{\textwidth}
    \centering
        \begin{minipage}{0.4\textwidth}
            \begin{lstlisting}[language=C++]
OPENQASM 3.0;
include "stdgates.inc"

// Initialization
qubit[3] psi;
qubit[2] S;

// Encoder
cx psi[0], psi[1];
cx psi[0], psi[2];
h psi;

// [Psi subject to noise]
            \end{lstlisting}
            \vspace{1.9em}
        \end{minipage}
        \hspace{2em}
        \begin{minipage}{0.4\textwidth}
            \begin{lstlisting}[language=C++,firstnumber=15]
// Syndrome extraction
h psi;
cx psi[0], S[0];
cx psi[1], S[0];
cx psi[0], S[1];
cx psi[2], S[1];
h psi;

// Error correction
ctrl(2) @ z S[0], S[1], psi[0];
x S[1];
ctrl(2) @ z S[0], S[1], psi[1];
x S[0]; x S[1];
ctrl(2) @ z S[0], S[1], psi[2];
x S[0];
h psi;\end{lstlisting}
        \end{minipage}
        \caption{OpenQASM implementation}
    \end{subfigure}
    
    \caption{
        The three-qubit error correction code used to detect phase flip errors when transmitting a quantum state $\ket{\psi}$.
    }
    \label{fig:3qubitphase}

\end{figure}

\begin{table}
    \centering
    \begin{tabular}{ccccccl}
        \toprule
        $\ket{\psi}_L$ & $\ket{\hat{\psi}}_L$ & $S$ & Count & $P_{obs}$ & $P_{true}$ & Analysis \\[0.1cm] 
                                                                                       \midrule \\[-0.4cm]
        % =======================================================================================================
        $\ket{000}$ & $\ket{000}$ & 00 & 2244 & 0.548 \hspace{0.4em} & 0.729 & No transmission errors. \\[0.15cm]
        % =======================================================================================================
        $\ket{001}$ & $\ket{000}$ & 01 &  229 &
                                            \multirow{3}{*}{$\left.\begin{array}{l}
                                                0.056 \\
                                                0.042 \\
                                                0.055 \\
                                            \end{array}\right\rbrace$} & \multirow{3}*{0.081} & \multirow{3}*{One error; detectable and correctable.} \\
        $\ket{010}$ & $\ket{000}$ & 10 &  174 & & & \\
        $\ket{100}$ & $\ket{000}$ & 11 &  224 & & & \\[0.15cm]
        % =======================================================================================================
        $\ket{011}$ & $\ket{111}$ & 11 &   18 & 
                                            \multirow{3}{*}{$\left.\begin{array}{l}
                                                0.004 \\
                                                0.006 \\
                                                0.006 \\
                                            \end{array}\right\rbrace$} & \multirow{3}*{0.009} & \multirow{3}*{Two errors; detectable but not correctable.} \\
        $\ket{101}$ & $\ket{111}$ & 10 &   24 & & & \\
        $\ket{110}$ & $\ket{111}$ & 01 &   24 & & & \\[0.15cm]
        % =======================================================================================================
        $\ket{111}$ & $\ket{111}$ & 00 &    2 & $<$ 0.001 \hspace{0.7em} & 0.001 & Three errors, neither detectable nor correctable. \\[0.15cm]
        % =======================================================================================================
        \multicolumn{3}{r}{Other:}     & 1157 & 0.282\hspace{0.7em} & 0.000 & Errors occurred outside of noise channel. \\ \bottomrule
    \end{tabular}
    
    \caption{
        Empirical probabilities of phase errors on each combination qubits in the 3-qubit phase flip code.
        Values were obtained by running the circuit in Figure~\ref{fig:3qubitphase} on IBM Torino for 4096 shots. Here we have used a simple error model described in Appendix~\ref{sec:appendix:impl:noise}.
    }
    \label{tab:3qubit_phase}
\end{table}

When no error occurs, the state of the logical qubit as it enters into the syndrome
extractor is
\begin{equation}
    \ket{\psi'}_L=\alpha\ket{+++}+\beta\ket{---}.
\end{equation}
The Hadamard gates then transform this into the state
\begin{equation}
    \ket{\psi}_L=\alpha\ket{000}+\beta\ket{111},
\end{equation}
which we have already shown is left invariant by the $\cnot$ gates that follow.
We therefore obtain the syndrome 00 with certainty, and the state is transformed
back to the original logical qubit state $\ket{\psi'}_L$ by the final Hadamard gates.

If a phase flip occurs on the $i$-th qubit, the state before the syndrome extraction
becomes
\begin{equation}
    Z_i\ket{\psi'}_L=Z_iH_1H_2H_3\ket{\psi}_L,
\end{equation}
so that the action of the first set of Hadamard gates in the syndrome extractor
produces the state
\begin{equation}
    H_1H_2H_3Z_iH_1H_2H_3\ket{\psi}_L=X_i\ket{\psi}_L,
\end{equation}
where we have used the fact that $H_iZ_iH_i=X_i$ and $H_j^2=\mathds{1}$. The phase flip error on $\ket{\psi'}_L$
has been transformed into a \textit{bit} flip error on $\ket{\psi}_L$. Thus, the
remainder of the syndrome extractor performs exactly the same as the bit flip version
with the exception that at the end, the logical qubit is transformed back to the
phase flip picture. It is worth being explicit for clarity.
If there is a phase flip on the first physical qubit, we get
the syndrome 11 and the state of the logical qubit remains $Z_1\ket{\psi'}_L$ after 
syndrome extraction. If there is a phase flip on the second physical qubit,
we obtain the syndrome 10 and the state of the logical qubit remains $Z_2\ket{\psi'}_L$
after syndrome extraction. If there is a phase flip on the third physical qubit,
we obtain the syndrome 01 and the state of the logical qubit remains $Z_3\ket{\psi'}_L$
after syndrome extraction. Thus, our code is capable of detecting and correcting
single qubit phase flip errors. Just as in the bit flip case, there is no hope 
for correction of phase flip errors on more than one qubit in addition to the single qubit errors
using this code for the simple reason that we have used up all possible syndromes 
obtainable with two bits of information.

Let us again discuss the efficiency of the code. The argument proceeds in just the
same way as the bit flip case. If $p$ is the probability that a phase flip error
occurs, then the probability that exactly one phase flip occurs in the logical state
is
\begin{equation}
    \sum_{i=1}^3P(E=Z_i)=3p(1-p)^2,
\end{equation}
and the probability that no error occurs on any qubit is
\begin{equation}
    P(E=\mathds{1})=(1-p)^3.
\end{equation}
Thus, if when we obtain a syndrome not equal to 00, we attempt to correct the error
by assuming it is a single qubit error (the most likely case), the probability that
we have successfully corrected the true error is 
\begin{equation}
    p_{success}=(1-p)^3+3p(1-p)^2=1-3p^2+2p^3.
\end{equation}
This is exactly the same result that we obtained in the bit flip case. The logical error
rate of the code is therefore
\begin{equation}
    1-p_{success}=3p^2-2p^3,
\end{equation}
and the error rate without correction is still $p$.
Thus, we still find that the difference is positive for $p<\frac12$, showing that the 3-qubit
phase flip code again lowers the error rate whenever the original physical error rate is $p<\frac12$. In other words, the threshold for the 3-qubit phase flip code is $p_{th}=\frac12$.

For the bit flip code, we saw that if the environment interacted with our system
in a specific way, the resulting error could still be handled by our code, which
essentially forces the system to decide whether or not an error has occurred. 
Similarly, if the environment interacts with our logical qubit according to the 
non-unitary operation
\begin{equation}
    (1-\epsilon)\mathds{1}+\epsilon Z
\end{equation}
our phase flip code will act in the same way. Indeed, let us assume, for example,
that the $Z$ acts on the first qubit. Then the state after the error 
is
\begin{equation}
    ((1-\epsilon)\mathds{1}+\epsilon Z_1)\ket{\psi'}_L\ket{00}
    =(1-\epsilon)\ket{\psi'}_L\ket{00}+\epsilon(\alpha\ket{-++}+\beta\ket{+--})\ket{00},
\end{equation}
and after applying the first round of Hadamard gates in the syndrome extractor,
this state is transformed to the bit flip picture, where the subsequent application
of the $\cnot$ gates in the syndrome extractor produces the state
\begin{equation}
    (1-\epsilon)\ket{\psi}_L\ket{00}+\epsilon(\alpha\ket{100}+\beta\ket{011})\ket{11}.
\end{equation}
Thus, upon measurement, we will obtain either the syndrome 00 with probability
$P(00)=(1-\epsilon)^2/(1-2\epsilon+2\epsilon^2)$ or the syndrome 11 with probability
$P(11)=\epsilon^2/(1-2\epsilon+2\epsilon^2)$. In the first case, the state of the
logical qubit after the measurement is $\ket{\psi}_L$, which gets transformed back
to $\ket{\psi'}_L$ by the remaining Hadamard gates, and so no error has occurred.
In the second case, the state of the system becomes $X_1\ket{\psi}_L$, which is transformed
to the state $Z_1\ket{\psi'}_L$ by the remaining Hadamard gates. 
Thus, the syndrome extraction forces the state of the system to choose whether 
or not a phase flip error has occurred on the first qubit, and the same decoder 
that we described in the case of the unitary error can be applied here if necessary.
We note that just as in the bit flip case, this phenomenon generalizes to a continuous spectrum
of phase flip errors of the form
\begin{equation}
    c_0\mathds{1}+c_1Z_1+c_2Z_2+c_3Z_3,
\end{equation}
and even to sums containing products of phase flip errors.

\subsection{The Shor Code}\label{sec:shorcode}

The first full single qubit error correcting code to be proposed was the nine qubit code
introduced by Peter Shor \cite{shor1995}. This code protects against any single qubit error and is
constructed by concatenating codes; that is, by using the output of a first code
as the input of a second code. The two codes we will concatenate are the 3-qubit
bit and phase flip codes of the previous sections.

We begin by encoding our given state $\ket{\psi}=\alpha\ket{0}+\beta\ket{1}$ 
using the phase flip encoding in Figure~\ref{fig:3qubitphase}. This transforms 
$\ket{\psi}$ into the 3-qubit state $\ket{\psi'}_L=\alpha\ket{+++}+\beta\ket{---}$.
Next we apply the bit flip encoding in Figure~\ref{fig:3qubitbit} to each of the
three physical qubits in $\ket{\psi'}_L$, giving us a total of nine physical qubits
in our logical qubit. The bit flip encoding maps each $\ket{+}$ state to the 3-qubit
GHZ state
\begin{equation}
    \ket{GHZ}=\frac{1}{\sqrt{2}}(\ket{000}+\ket{111})
\end{equation}
and each $\ket{-}$ state to the phase flipped GHZ state
\begin{equation}
    \ket{\widehat{GHZ}}=\frac{1}{\sqrt{2}}(\ket{000}-\ket{111}).
\end{equation}
Thus, the full encoding is given by
\begin{equation}
    \ket{\psi}=\alpha\ket{0}+\beta\ket{1}\longmapsto\frac{\alpha}{2\sqrt{2}}
    (\ket{000}+\ket{111})^{\otimes3}+\frac{\beta}{2\sqrt{2}}(\ket{000}-\ket{111})^{\otimes3},
\end{equation}
and the circuit diagram is shown just before the noise channel in Figure~\ref{fig:shorcodea}. Let us denote the
state of the logical qubit after encoding by $\ket{\Psi}_L$. Note that the distance of the code is $d=3$ since a Pauli-$Z$ operation on exactly one qubit of the three qubits in each of the three tensor product components of $\frac{1}{2\sqrt{2}}(\ket{000}+\ket{111})^{\otimes 3}$ produces $\frac{1}{2\sqrt{2}}(\ket{000}-\ket{111})^{\otimes 3}$. Thus, the Shor code is a $[[9,1,3]]$ code with rate $R=1/9$.

\begin{figure}
    \centering
    \begin{subfigure}{\textwidth}
        \centering
        \includegraphics[height=4.1in]{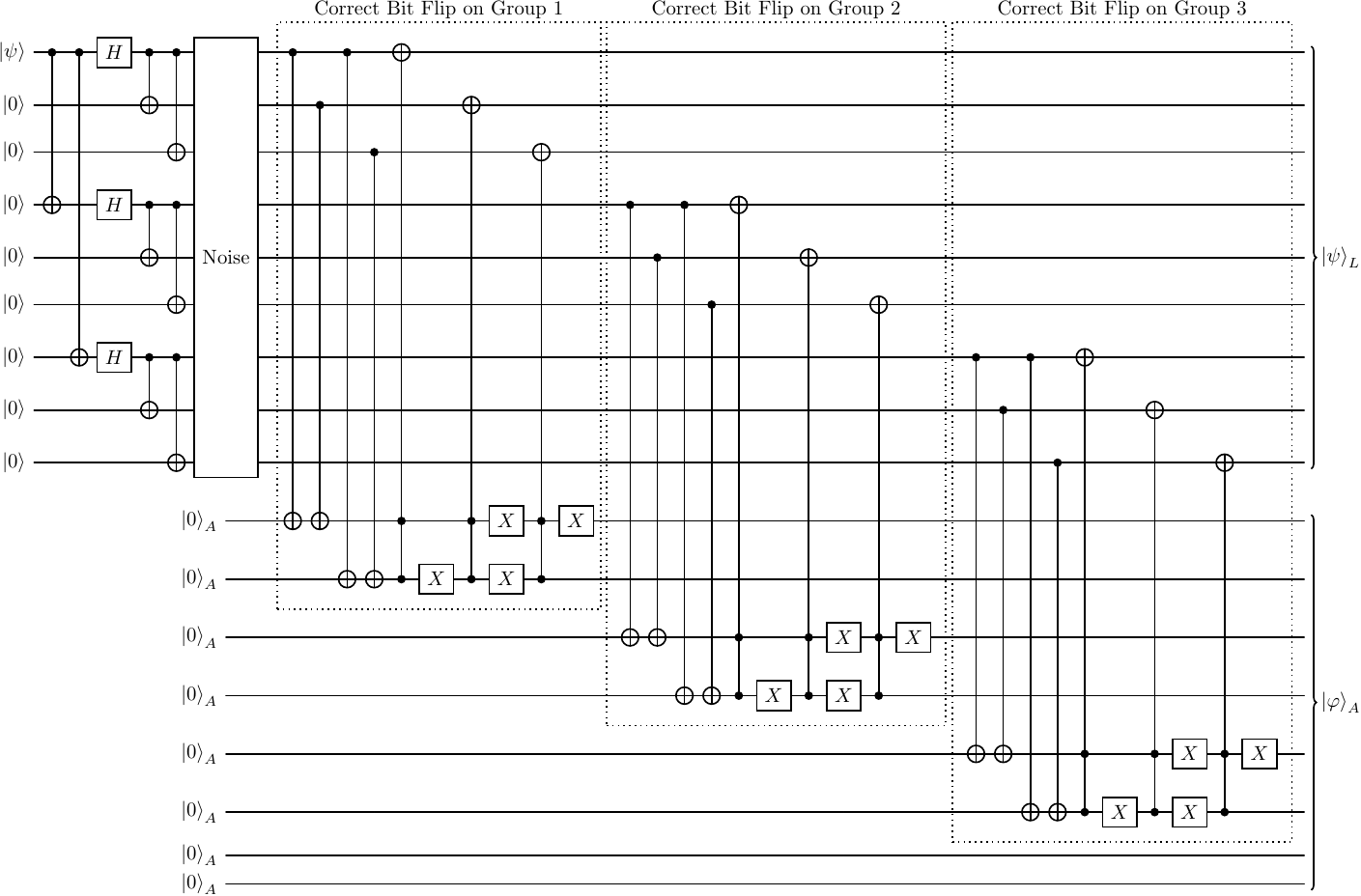}
        \caption{Circuit diagram (first half).}
        \label{fig:shorcodea}
    \end{subfigure}
    
    \vspace{1em}
    
    \begin{subfigure}{\textwidth}
        \centering
        \includegraphics[height=4.1in]{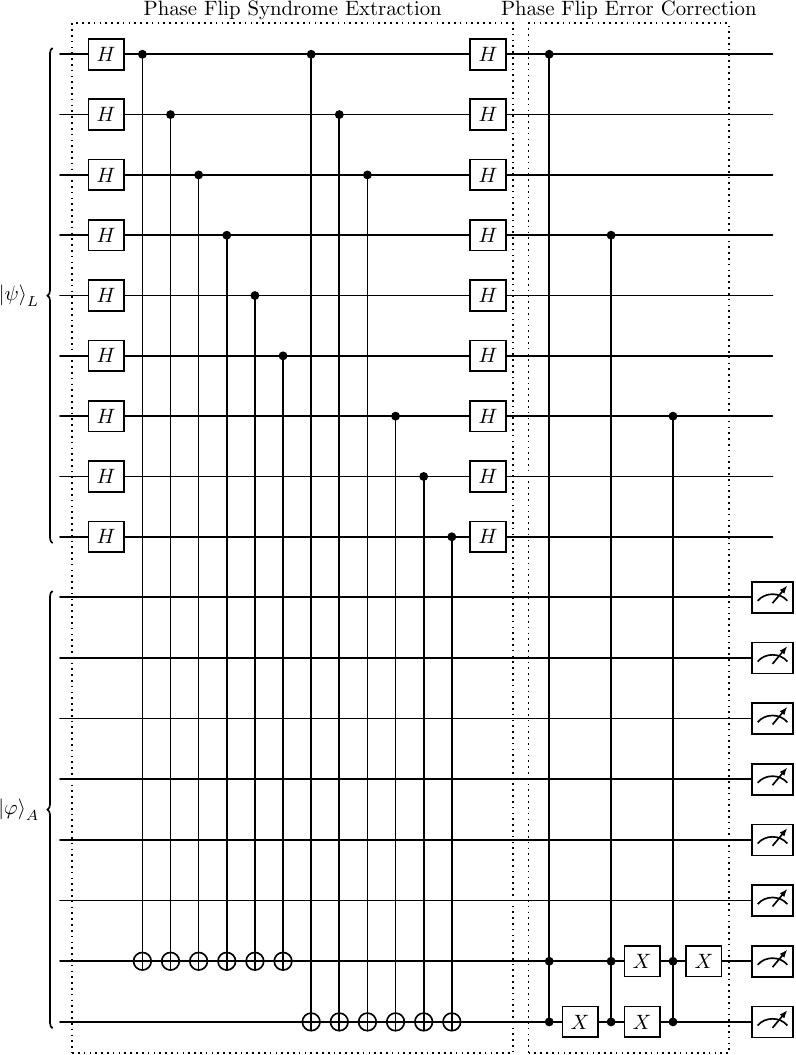}
        \caption{Circuit diagram (second half).}
        \label{fig:shorcodeb}
    \end{subfigure}

    \caption{The 9-qubit Shor code.}
    \label{fig:shorcode}
    
\end{figure}

\begin{table}
    \centering
    \begin{tabular}{ccccccl}
        \toprule
        Noise & $\ket{\varphi}_A^{(1,2)}$ & $\ket{\varphi}_A^{(7,8)}$ & Count & $P_{obs}$ & $P_{true}$ \\[0.1cm] 
                                                                                       \midrule \\[-0.4cm]
        % =======================================================================================================
        $\mathbb{I}$ & 00 & 00 & 1598 & $1.95 \cdot 10^{-1}$\hspace{0.7em} & $6.56 \cdot 10^{-1}$ \\[0.15cm]
        % =======================================================================================================
        $X_1$ & 11 & 00 & 38 &
                                            \multirow{4}{*}{$\left.\begin{array}{l}
                                                4.64 \cdot 10^{-3} \\
                                                5.13 \cdot 10^{-3} \\
                                                1.95 \cdot 10^{-3} \\
                                                2.20 \cdot 10^{-3} \\
                                            \end{array}\right\rbrace$} & \multirow{4}*{$7.29 \cdot 10^{-2}$} \\
        $X_3$ & 01 & 00 & 42 & & & \\
        $Z_1$ & 00 & 11 & 16 & & & \\
        $Z_3$ & 00 & 11 & 18 & & & \\[0.15cm]
        % =======================================================================================================
        $X_1 X_3$ & 10 & 00 & 1 & 
                                            \multirow{6}{*}{$\left.\begin{array}{l}
                                                1.22 \cdot 10^{-4} \\
                                                0.00 \\
                                                0.00 \\
                                                1.22 \cdot 10^{-4} \\
                                                0.00 \\
                                                0.00 \\
                                            \end{array}\right\rbrace$} & \multirow{6}*{$8.10 \cdot 10^{-3}$} \\
        $X_1 Z_1$ & 11 & 11 & 0 & & & \\
        $X_1 Z_3$ & 11 & 11 & 0 & & & \\
        $X_3 Z_1$ & 01 & 11 & 1 & & & \\
        $X_3 Z_3$ & 01 & 11 & 0 & & & \\
        $Z_1 Z_3$ & 00 & 00 & 0 & & & \\[0.15cm]
        % =======================================================================================================
        $X_1 X_3 Z_1$ & 10 & 11 & 0 & 
                                            \multirow{4}{*}{$\left.\begin{array}{l}
                                                0.00 \hspace{2.75em} \\
                                                0.00 \\
                                                0.00 \\
                                                0.00 \\
                                            \end{array}\right\rbrace$} & \multirow{4}*{$9.00 \cdot 10^{-4}$} \\
        $X_1 X_3 Z_3$ & 10 & 11 & 0 & & & \\
        $X_1 Z_1 Z_3$ & 11 & 00 & 0 & & & \\
        $X_3 Z_1 Z_3$ & 01 & 00 & 0 & & & \\[0.15cm]
        % =======================================================================================================
        $X_1 X_3 Z_1 Z_3$ & 10 & 00 & 0 & 0.00 \hspace{3.15em} & $1.00 \cdot 10^{-4}$ \\[0.15cm]
        % =======================================================================================================
        \multicolumn{3}{r}{Other:}     & 6478 & $7.91 \cdot 10^{-1}$\hspace{0.7em} & 0.00\hspace{2.8em} \\ \bottomrule
    \end{tabular}
    
    \caption{
        Empirical probabilities of bit and phase flip errors on qubits in the 9-qubit Shor code where $X_1$, $X_3$, $Z_1$, and $Z_3$ errors are introduced independently with probability $0.1$.
        Values were obtained by running the circuit in Figure~\ref{fig:shorcode} on IBM Torino for 8192 shots.
        Here, noise multiplies $\ket{\psi}_L$ and $\ket{\varphi}_A^{(i,j)}$ denotes the $i$-th and $j$-th ancilla qubit. The discrepancy between the observed and true values can be explained by noisy encoding, syndrome extraction, and measurement operations which have not been accounted for. In our construction, we assume perfect such operations with the knowledge that if the error rate is small enough, the threshold theorem guarantees that the logical rate will be lowered. Here we have used a simple error model described in Appendix~\ref{sec:appendix:impl:noise}.
    }
    \label{tab:shor}
\end{table}

Our claim is that the Shor code will correct any unitary single qubit error, so
let us check how bit and phase flip errors affect the state $\ket{\Psi}_L$ of the 
logical qubit. Suppose a bit flip error occurs on the first qubit. Notice that each
component of $\ket{\Psi}_L$ can be separated into a three-fold tensor product of GHZ states,
and the latter two GHZ states are unaffected by the action of a bit flip on the first qubit.
By applying the bit flip syndrome extraction to the first GHZ state, we can therefore
detect and correct this error. This argument carries over to a bit flip error on
any physical qubit. Thus, in the Shor code, we perform a bit flip syndrome extraction on every block of three physical qubits as in Figure~\ref{fig:shorcodea}.

Suppose that a phase flip instead occurs on the first qubit. This has the effect
of transforming the first $\ket{GHZ}$ to a $\ket{\widehat{GHZ}}$, and so we would like to
construct a way to distinguish between these two states. We will begin by applying
a Hadamard gate to every physical qubit. Observe that
\begin{align}
    H^{\otimes3}\ket{GHZ}&=\frac{1}{4}\left[(\ket{0}+\ket{1})^{\otimes3}+(\ket{0}-\ket{1})^{\otimes3}\right]\\
    &=\frac{1}{2}(\ket{000}+\ket{011}+\ket{110}+\ket{101})
\end{align}
and that
\begin{align}
    H^{\otimes3}\ket{\widehat{GHZ}}&=\frac{1}{4}\left[(\ket{0}+\ket{1})^{\otimes3}-(\ket{0}-\ket{1})^{\otimes3}\right]\\
    &=\frac{1}{2}(\ket{111}+\ket{001}+\ket{010}+\ket{100}).
\end{align}
Each component of the transformed GHZ state has an even number of 1's, and each
component of the transformed $\widehat{GHZ}$ state has an odd number of 1's. Thus, we can
distinguish between these two states by checking the number of 1's in each component.
We do so by applying a sequence of three $\cnot$ gates targeting an ancillary qubit in
the state $\ket{0}$ and controlled off of the three qubits making up the transformed GHZ or $\widehat{GHZ}$ state.
In the case of the GHZ state, an even number of $\cnot$ gates will trigger, so that
a measurement of the ancillary qubit will produce the outcome 0. In the case of
the $\widehat{GHZ}$ state, an odd number of $\cnot$ gates will trigger, so that a measurement
of the ancillary qubit will produce the outcome 1. Thus, we have the desired procedure
for distinguishing between $\ket{GHZ}$ and $\ket{\widehat{GHZ}}$. The circuit diagram is 
shown in Figure~\ref{fig:ghzdistinguish}.

\begin{figure}
    \centering
    \begin{quantikz}
        \lstick[3]{$\ket{GHZ}$\\ or \\$\ket{\widehat{GHZ}}$} & \gate[1]{H} & \ctrl{3} &  &  &  \\
         & \gate[1]{H} &  & \ctrl{2} & &  \\
         & \gate[1]{H} & & & \ctrl{1} &  \\
        \lstick{$\ket{0}$} &             & \targ{} & \targ{} & \targ{} & \meter{}
    \end{quantikz}
    \caption{
        Procedure for distinguishing between $\ket{GHZ}$ and $\ket{\widehat{GHZ}}$.
    }\label{fig:ghzdistinguish}
\end{figure}
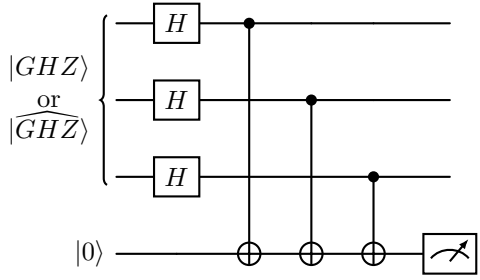

We now have a method for determining if a phase flip has occurred in each of the
three blocks of three qubits. Of course, we would like to detect an error in any
of the nine total qubits. To this end, let us introduce two ancillary qubits and
examine what happens if we combine four of the phase flip detection circuits as
shown in the phase flip syndrome extraction portion of Figure~\ref{fig:shorcodeb}.

If no error occurs, none of the $\cnot$ gates trigger, and so we obtain the syndrome
00 with certainty and the Hadamard gates cancel each other so that the state of
the logical qubit is left invariant. If there is a single phase flip error in one of
the first three physical qubits, the state of the logical qubit becomes 
\begin{equation}
    \frac{1}{\sqrt{2}}(\ket{\widehat{GHZ}}\ket{GHZ}\ket{GHZ}+\ket{GHZ}\ket{\widehat{GHZ}}\ket{\widehat{GHZ}}).
\end{equation}
In this case, one of the first three $\cnot$ gates will be triggered, 
causing the first ancillary qubit to flip. One of the seventh, eigth, or ninth 
$\cnot$ gates will also trigger causing the second ancillary qubit to flip. Thus,
a measurement of the ancillary qubits produces the syndrome 11. Moreover, the state
of the logical qubit is again unaffected by this procedure. Similarly, a phase flip
in the second three qubits will produce the logical qubit state
\begin{equation}
    \frac{1}{\sqrt{2}}(\ket{GHZ}\ket{\widehat{GHZ}}\ket{GHZ}+\ket{\widehat{GHZ}}\ket{GHZ}\ket{\widehat{GHZ}}),
\end{equation}
and one of the second three $\cnot$ gates is triggered, so that a measurement of the ancillary
qubits produces the syndrome 10. Finally, a phase flip in the last three qubits
produces the logical qubit state
\begin{equation}
    \frac{1}{\sqrt{2}}(\ket{GHZ}\ket{GHZ}\ket{\widehat{GHZ}}+\ket{\widehat{GHZ}}\ket{\widehat{GHZ}}\ket{GHZ}),
\end{equation}
and one of the final three $\cnot$ gates is triggered, so that a measurement of
the ancillary qubits produces the syndrome 01. Thus, the phase syndrome extractor in 
Figure~\ref{fig:shorcodeb} is capable of detecting all single qubit phase flips.

To be precise, the syndrome extraction procedure we have just outlined only tells 
us which block of three qubits a phase flip error has occurred in. Fortunately, due
to the way we have chosen to encode our data qubit, it does not matter which qubit
in each block experiences an error. Indeed, $Z_i\ket{GHZ}=\ket{\widehat{GHZ}}$ regardless
of the choice of $i$, and so we have $Z_iZ_j\ket{GHZ}=\ket{GHZ}$ even if $i\ne j$.
Thus, we can correct an error on any of the three qubits in a given block by applying
a Pauli-$Z$ gate to the erroneous qubit or any other qubit in the block. In Figure~\ref{fig:shorcodeb}, we have elected to apply the Pauli-$Z$ gate to the first qubit in each block when needed.

We have seen that the Shor code can detect and correct both single qubit bit and 
phase flip errors. By performing the syndrome extraction procedure for each of these
errors, we can furthermore correct the composition of a bit flip error with a phase
flip error. This is because the syndrome extraction procedure leaves the state of
the logical qubit invariant, so that a sequence of syndrome extraction operations
can be applied with no consequences. All we need is a fresh set of ancillary qubits
in the $\ket{0}$ state at each step. Now, the Shor code is only useful to us if it can lower the rate at which errors occur. If a single qubit error occurs independently on each physical qubit with probability $p$, then the probability that an error either doesn't occur or is corrected is
\begin{equation}
    p_{success}=(1-p)^9+9p(1-p)^8.
\end{equation}
Thus, the probability that the code fails (the error rate for the logical qubit) is
\begin{equation}
    p_{fail}=1-(1-p)^9-9p(1-p)^8=1-(1-p)^8(1+8p).
\end{equation}
The error threshold is the value of $p$ such that $p_{fail}<p$, which a computer readily approximates as $p\approx0.0323$. Thus, the logical error rate is less than the physical error rate when the physical error rate is below $3.23\%$. Actually, we have implicitly made several assumptions here and in each of our previous threshold calculations. We have assumed that the noise affecting the system takes the form of a sum of Pauli operations (a depolarizing channel) which acts independently on each physical qubit and that the encoder/decoder and syndrome extraction perform perfectly. For this reason, the $3.23\%$ we have computed is sometimes referred to as a pseudo-threshold. In practice, if we take into account 2-qubit gate noise, measurement noise, and so on, it turns out that the threshold of Shor's 9-qubit code drops dramatically to around $0.01\%-0.1\%$ \cite{knill1996,aliferis2005}. Later, we will construct the surface codes which have been shown to have a threshold on the order of $1\%$.

\subsection{The Quantum Hamming Bound} \label{sec:quantum-hamming-bound}

The Hamming bound from classical coding theory imposes a fundamental limit on the number of codewords a code can possess while correcting a fixed number of errors. There is an analog to this result known as the \textbf{quantum Hamming bound}, which similarly constrains the parameters of quantum error-correcting codes.

To set the stage, recall that a code encodes $k$ logical qubits into $n$ physical qubits in such a way that arbitrary errors on up to $t$ qubits can be detected and corrected. Unlike classical errors, which typically consist of symbol substitutions, quantum errors are described by operators from the Pauli group: $I$, $X$, $Y$, and $Z$. For a single qubit, there are thus 3 nontrivial error types; for $t$ qubits, errors can affect any subset of up to $t$ qubits, and for each such subset, $3^t$ distinct Pauli error combinations are possible (excluding the identity). Thus, the number of correctable error patterns should not exceed the total number of mutually orthogonal subspaces available in the full Hilbert space of $n$ qubits. Each logical codeword resides in a $2^k$-dimensional subspace, and there are $2^n$ total dimensions available. Each correctable error effectively consumes a full copy of the logical subspace. Thus, a necessary condition for the existence of an $[[n, k, d]]$ quantum code that corrects $t = \lfloor (d-1)/2 \rfloor$ errors is
\begin{equation} \label{eq:quantum-hamming-bound}
    \sum_{j=0}^t \binom{n}{j} 3^j \cdot 2^k \le 2^n.
\end{equation}
This is the quantum Hamming bound. It states that the total Hilbert space volume required to distinguish all possible correctable errors (each acting on a logical subspace of dimension $2^k$) must not exceed the dimension of the ambient space, which is $2^n$.

As in the classical case, the quantum Hamming bound provides a useful necessary (but not sufficient) condition for constructing quantum codes. It plays a particularly important role in the study of \emph{nondegenerate} quantum codes, where distinct errors map the code space to orthogonal subspaces. For degenerate codes, which can correct certain errors without needing to distinguish them, the bound may be violated. Nonetheless, the quantum Hamming bound serves as a fundamental benchmark in quantum error correction theory. It defines the quantum analogue of the classical sphere-packing limit and provides essential intuition for the trade-offs between code length, error-correcting power, and logical capacity.

\section{Basics of Stabilizer Codes} \label{sec:stabilizer}

In this section, we will introduce a generalization of the codes studied in the previous section. As with each of the previous examples, we create redundancy by embedding the state we wish to protect into a higher dimensional Hilbert space called the code space. This encoded state can then be passed through a noisy channel, and a syndrome extraction procedure can be applied to identify any errors that might have occurred. To perform syndrome extraction, we construct an abelian subgroup of the Pauli group with generators that anticommute with the errors we'd like to detect and leave the code space invariant. We then perform a Hadamard test using each of the generators of the stabilizer group, which has the effect of projecting the encoded state onto an eigenspace of the generator. As a consequence of our assumption that the logical state is left invariant by each stabilizer, this projection will not affect the state of the logical qubit even if it has been altered by an error. Moreover, the ancillary qubit in the Hadamard test keeps track of whether or not the error that occurred anticommutes with the stabilizer generator, in which case the error is detected. By applying one such test for each of the $k$ generators of the stabilizer, we are able to detect $2^k-1$ different errors (one for each possible syndrome, minus the all zero syndrome corresponding to no error). Once the errors are identified, they can be corrected by applying the inverse operation. The encoded state can then be decoded by applying the inverse of the encoder and throwing away the redundancy register.

In addition to being useful for protecting the transfer of quantum information across a noisy channel, the stabilizer formalism can be used as a foundation for a fault-tolerant quantum computer. Indeed, we think about the encoded state (made of several physical qubits) as a logical qubit which can be acted upon by logical gates built out of several smaller gates. We then detect errors on the logical qubit by applying the same syndrome extraction procedure as above and correct them by applying the inverse operation. We do not apply the inverse of the encoder because there is no reason to decode and recover the original Hilbert space; the logical qubit is now the object of interest. The syndrome extraction can be repeated indefinitely, so that errors are in principle continuously corrected as they appear, leaving us with an error resistant ``qubit''. By combining logical qubits and deriving logical operations on them, we are able to program our new ``logical device'' in the same way we would program a quantum device with truly error resistant qubits. 

The general stabilizer code requires some knowledge of group theory to understand, and so we provide a sufficient background in Sections~\ref{sec:groups} and \ref{sec:groupactions}. The theory of groups is much deeper than the short review presented here, and the interested reader is referred to any of a number of great texts on abstract algebra, such as the classic \cite{dummit2004}. We introduce the stabilizer formalism in Section~\ref{sec:stab-formalism} and construct the 5-qubit code within this context in Section~\ref{sec:5qubitcode}. We then translate the stabilizer code formalism to the language of linear algebra in Section~\ref{sec:paritycheck} before discussing difficulties that should be considered when using stabilizer codes in Section~\ref{sec:difficulties}.

%%%%%%%%%%%%%%%%%%%%%%%%%%%%%%%%%%%%%%%%%%%%%%%%%%%%%%%%%%%%%%%%%%%%%%%%%%%%%%%%%%%%%%%%%%%%%%%%%%%%%%%%%%%%%%%%%%%%%%%
\subsection{The Theory of Groups}\label{sec:groups}

Symmetry is pervasive in physics. Often it appears as a tool to simplify a problem, such as in the case of the spherically symmetric
hydrogen atom. We model the electron as sitting in a Coulomb potential which is inversely proportional to the radial distance from the nucleus but independent of angular coordinates. Thus, the Coulomb potential is symmetric with respect to rotations, and this property can be leveraged to simplify the problem of finding the state of the electron according to Schr\"odinger's equation. A more intuitive notion of symmetry is that of a regular polygon. The square, for example, is symmetric with respect to a flip across the vertical line passing through its center point as well as rotations in 90 degree increments. It also possesses symmetry with respect to flips along the horizontal line passing through its center, as well as the diagonal lines connecting the center point to each corner, but these symmetries can be decomposed into combinations of the vertical flip with the rotations and are therefore not independent.

In any case, we notice that whenever a symmetry occurs, the composition of two symmetries is again a symmetry. Indeed, two rotations of our electron in the spherical coordinate system are still a rotation and therefore constitute a symmetry of the system. Likewise, a rotation of the square followed by a flip across the vertical line is still a symmetry of the square. We also notice that there is a symmetry which does absolutely nothing. Indeed, it should come at no surprise that doing nothing might leave a system invariant. Additionally, we notice that for every symmetry, there is a corresponding symmetry which transforms the system back to the state it started in; that is, a symmetry which undoes the original symmetry transformation. In the square example, a rotation followed by a rotation in the opposite direction takes the square back to where it started, and so this combination of symmetries is equivalent to the symmetry which does nothing. In both of the cases we have examined, we also notice that the composition of symmetries is associative, so that the composition of three symmetries is well-defined without having to specify in what order we compose them. 

These properties define an abstract mathematical object known as a group. Thus, groups capture what it means to be a symmetry, and we would be well-advised to study their structure. As it turns out, the theory of groups is a central component of the stabilizer code formalism, and so we take a detour in this section to learn what group theory is needed.

\begin{definition}\label{def:group} A \textbf{group} $(G,b)$ is a set $G$ together with a binary 
    operation $b:G\times G\to G$ such that the following properties are satisfied:
\begin{enumerate}
    \item (Associativity) $b(b(g,h),k)=b(g,b(h,k))$ for all $g,h,k\in G$.
    \item (Existence of Identity) There exists an element $e\in G$ called the identity such that $b(e,g)=g=b(g,e)$ for all $g\in G$.
    \item (Invertibility) For every $g\in G$, there exists an $h\in G$ such that $b(g,h)=e=b(h,g)$. 
    Such an element $h$ is called the inverse of $g$.
\end{enumerate}
\end{definition}

Usually, the clunky $b(g,h)$ notation for the binary operation is dropped, and the notation 
$gh:=b(g,h)$ is instead used. That is, the binary operation is written as a juxtaposition of the 
arguments. When $b(g,h)=b(h,g)$ for all $g,h\in G$, the group is called \textbf{abelian}, and the 
binary operation is instead often written as $g+h:=b(g,h)$. When $G$ is abelian, and the $g+h$ notation is 
used, the inverse of $g$ is written as $-g$. Otherwise, the inverse of $g$ is written $g^{-1}$. 
With these conventions, we may drop the notation $(G,b)$ for a group, and label it instead by the 
set symbol $G$ alone with the understanding that this set comes equipped with a binary operation.

From Definition~\ref{def:group} it follows that the identity element $e$ and the inverse $g^{-1}$ 
of an element $g\in G$ are unique. Indeed, if there exists $e'$ such that $g e'=g=e' g$ for all 
$g\in G$, it follows that $ge'=g=ge$. Thus, by applying $g^{-1}$, we find that $e'=e$. Similarly, 
if there exists an $h\in G$ such that $gh=e=hg$, then $gh=e=gg^{-1}$, from which it follows that 
$h=g^{-1}$. The identity element of a non-abelian group is usually denoted `$1$', whereas the 
identity in an abelian group is written `$0$'.

\begin{definition}
    A \textbf{subgroup} of a group $G$ is a subset $H$ of $G$ which is itself a group under the 
    operation of $G$. That is, $H\subset G$ is a subgroup if $h_1h_2\in H$ for all $h_1,h_2\in H$, 
    where the operation of $G$ has been used.
\end{definition}

The notation $H< G$ is used to indicate that $H$ is a subgroup of $G$. The integers $\mathbb{Z}$ 
equipped with the operation of addition form an example of an abelian group, and the even integers 
are a subgroup of this group since the sum of even integers is an even integer. Given a subset 
$S\subset G$, there is a smallest subgroup $\langle S\rangle$ of $G$ containing $S$ called the 
\textbf{subgroup generated by} $S$. This follows from the fact that the intersection of subgroups 
is again a subgroup. When $S=\{g\}$ is a singleton, we write $\langle S\rangle:=\langle g\rangle$.

\begin{definition}
    The \textbf{order} $\lvert G\rvert$ of a group $G$ is the number of elements in the underlying 
    set. The order $o(g)$ of an element $g\in G$ is the smallest positive integer $n$ such that 
    $g^n=1$ (if it exists).
\end{definition}

The order of the subgroup generated by the singleton $g$ is $\lvert\langle g\rangle\rvert=o(g)$. 
Moreover, if $o(g)=n$, then $\langle g\rangle:=\{1,g,\ldots,g^{n-1}\}$. If $G=\langle g\rangle$ for 
some $g\in G$, then $G$ is called \textbf{cyclic} and we say that $G$ is generated by the element 
$g$. One such example is the set $\{0,1\}$ equipped with addition modulo 2, called $\mathbb{Z}_2$, 
which is generated by the element $1$.

To classify groups, we must introduce a mapping which preserves the properties inherent in this 
structure. Let $\varphi:G\to H$ be a map between groups. In order to preserve the structure of the 
binary operation as $G$ is mapped to $H$, we require that $\varphi(g_1g_2)=\varphi(g_1)\varphi(g_2)$ 
for all $g_1,g_2\in G$. Notice that the multiplication on the left hand side is given by the binary 
operation of $G$, while on the right hand side, the binary operation of $H$ is used. Also, a simple 
consequence of this condition is that $\varphi$ maps the identity in $G$ to the identity in $H$. 
Indeed, $\varphi(1)=\varphi(1^2)=(\varphi(1))^2$. Now multiplying by $(\varphi(1))^{-1}$ yields 
$\varphi(1)=1$. The remaining structure of a group is that of the underlying set. Thus, to fully 
preserve the structure of a group, we make the additional requirement that $\varphi$ is a bijection. 

\begin{definition}
    A map $\varphi:G\to H$ is called a \textbf{homomorphism} if 
    $\varphi(g_1g_2)=\varphi(g_1)\varphi(g_2)$ for all $g_1,g_2\in G$. 
    If, in addition, $\varphi$ is a bijection, then it is called an \textbf{isomorphism}. 
    In this case, $G$ and $H$ are said to be isomorphic and we write $G\cong H$.
\end{definition}

A homomorphism $\varphi:G\to G$ from a group to itself is called an \textbf{endomorphism}. 
Moreover, if $\varphi$ is a bijection, then it is called an \textbf{automorphism}. 
Note that the set Aut$(G)$ of automorphisms of a group $G$ forms a group under the operation of 
composition of functions. Indeed, the identity in this group is the identity automorphism 
$id:G\to G$ given by $id(g)=g$, and the inverse of an automorphism exists and is itself an 
automorphism since automorphisms are bijective and $\varphi(gh)=\varphi(g)\varphi(h)$ implies 
that $\varphi^{-1}(gh)=\varphi^{-1}(g)\varphi^{-1}(h)$. Associativity follows from the definition 
of the composition of functions.

It will be useful to be able to take quotients of groups by their subgroups, and a special type of 
subgroup will be necessary to accomplish this. Suppose $K$ is a subgroup of $G$. If $gkg^{-1}\in K$ 
for all $g\in G$, then $K$ is called a \textbf{normal} subgroup. The next proposition shows that 
the kernel $\ker(\varphi):=\{g\in G\ |\ \varphi(g)=1\}$ and the image 
$\varphi(G):=\{h\in H\ |\ h=\varphi(g)\ \text{for some }g\in G\}$ of a homomorphism 
$\varphi:G\to H$ are subgroups and that the kernel is in fact normal.

\begin{proposition} Let $\varphi:G\to H$ be a homomorphism. Then $\ker(\varphi)$ is a subgroup of 
    $G$ and $\varphi(G)$ is a subgroup of $H$. Moreover, $\ker(\varphi)$ is normal.
\end{proposition}
\begin{proof} It is clear that $\ker(\varphi)$ is a subgroup of $G$ since $g_1,g_2\in \ker(\varphi)$ 
    implies that $\varphi(g_1g_2)=\varphi(g_1)\varphi(g_2)=1$. Similarly, $h_1,h_2\in\varphi(G)$ 
    implies that there exists $g_1,g_2\in G$ such that $\varphi(g_1)=h_1$ and $\varphi(g_2)=h_2$. 
    Then $\varphi(g_1g_2)=\varphi(g_1)\varphi(g_2)=h_1h_2$, which establishes that $\varphi(G)$ is 
    a subgroup of $H$. To see that the kernel is in fact a normal subgroup, let $g\in G$ and 
    $k\in\ker(\varphi)$. Then 
\begin{align*}
    \varphi(gkg^{-1})=\varphi(g)\varphi(k)\varphi(g^{-1})=\varphi(g)\varphi(g^{-1})=\varphi(gg^{-1})=\varphi(1)=1,
\end{align*}
from which it follows that $gkg^{-1}$ is in $\ker(\varphi)$ for all $g\in G$, so that $\ker(\varphi)$ 
is a normal subgroup of $G$.
\end{proof}
The left \textbf{cosets} of $H$ in $G$ are the sets $gH:=\{gh\ \lvert\ h\in H\}$ for $g\in G$, and 
the right cosets are defined similarly. The \textbf{index} $[G:H]$ of $H$ in $G$ is the number of 
left cosets of $H$ in $G$. If $K$ is a normal subgroup of $G$, then the left and right cosets are 
equivalent. To see this, let $h\in gK$. Then there exists a $k\in K$ such that $h=gk$. Let 
$k'=gkg^{-1}$, which is an element of $K$ since $K$ is normal. Then $ k'g=gkg^{-1}g=gk=h$, which 
implies that $h\in Kg$. The converse is similar, showing that $gK=Kg$. Now define a relation on $G$ 
by $g\sim h$ if $g^{-1}h\in K$. Then $\sim$ is an equivalence relation. Indeed, $g^{-1}g=1$, so that 
$g\sim g$. Furthermore, if $g^{-1}h\in K$, then $h^{-1}g=(g^{-1}h)^{-1}\in K$ since $K$ is a group, 
which shows that $g\sim h$ implies $h\sim g$. Finally, suppose $g_1\sim g_2$ and $g_2\sim g_3$. 
Then $g_1^{-1}g_2\in K$ and $g_2^{-1}g_3\in K$. Thus, $g_1^{-1}g_3=g_1^{-1}g_2g_2^{-1}g_3\in K$, 
which shows that $g_1\sim g_3$; therefore, $\sim$ is an equivalence relation. Moreover, the 
equivalence classes of $\sim$ are the cosets of $K$ in $G$.

\begin{proposition}
    The set $G/K$ of cosets of $K$ in $G$ forms a group under multiplication of representatives. 
    That is, under the operation $gK\cdot hK=ghK$.
\end{proposition}
\begin{proof}
    To see that this operation is well-defined, suppose $gK=g'K$ for some $g'\ne g$. Then $g'\in gK$, 
    so that there exists $k\in K$ such that $g'=gk$. Thus, $g'K\cdot hK=gkK\cdot hK=gK\cdot hK$. 
    Similarly, the product is independent of the representative of $hK$ and is therefore well-defined. 
    Associativity follows from the associativity of the operation on $G$, the identity in $G/K$ is 
    $1K$, and the inverse of $gK$ is $g^{-1}K$. Hence, $G/K$ is a group.
\end{proof}

The importance of the quotient group construction cannot be understated. In the quotient $G/H$, the elements of $H$ are identified with the identity element of $G$. Later, we will let a group act on our quantum system in a well-defined way. In this context, the reduction of $H$ to the identity element forces the elements of $H$ to act trivially on the system. The next theorem is useful for calculating the order of the quotient group.

\begin{theorem}[Lagrange's Theorem]
    Let $H$ be a subgroup of $G$. Then $\lvert G\rvert=[G:H]\lvert H\rvert$. In particular, if $H$ 
    is a normal subgroup, then $\lvert G/H\rvert=\lvert G\rvert/\lvert H\rvert$.
\end{theorem}
\begin{proof}
    Suppose $gH\ne g'H$ for some $g,g'\in G$ and let $k\in gH\cap g'H$. Then there exists $h,h'\in H$ 
    such that $gh=k=g'h'$. Thus, $g=g'h'h^{-1}$ and so $g\in g'H$ and $gH\subset g'H$. Similarly, 
    $g'=gh(h')^{-1}$ and so $g'\in gH$ and $g'H\subset gH$. Then $gH=g'H$, a contradiction. Thus, 
    $gH\cap g'H=\emptyset$, so that any two cosets are disjoint. This establishes that the cosets of 
    $H$ in $G$ partition the set $G$. Now, the mapping $\theta_g:H\to gH$ given by $\theta_g(h)=gh$ 
    is surjective by construction and injective since $gh=gh'$ implies that $h=h'$. Thus, 
    $\lvert gH\rvert=\lvert H\rvert$ for all $g\in G$. It follows that the number of cosets of $H$ 
    in $G$ is $\lvert G\rvert/\lvert H\rvert$.
\end{proof}

The mapping $q:G\to G/K$ given by $q(g):=gK$ is a homomorphism called the quotient mapping. That 
this map is a homomorphism follows from $q(gh)=ghK=gK\cdot hK=q(g)q(h)$. Moreover, the kernel of 
this map is $K$ since $gK=q(g)=K$ is true if and only if $g\in K$. Thus, $q(G)=G/\ker(q)$ since $q$ 
is surjective. In fact, if $\varphi:G\to H$ is any homomorphism, then $G/\ker(\varphi)\cong \varphi(G)$,
a fact which is known as the first isomorphism theorem. 
In particular, if $\varphi$ is surjective, then $G/\ker(\varphi)\cong H$.

\begin{theorem}[First Isomorphism Theorem for Groups]\label{thm:first-iso-groups} If 
    $\varphi:G\to H$ is a homomorphism of groups, then $G/\ker(\varphi)\cong\varphi(G)$.
\end{theorem}
\begin{proof}
    Define a mapping $\Phi:G/\ker(\varphi)\to \varphi(G)$ by $\Phi(g\ker(\varphi))=\varphi(g)$. 
    To see that this mapping is well-defined, let $g'\ker(\varphi)=g\ker(\varphi)$. 
    Then $g'=gk$ for some $k\in\ker(\varphi)$ and 
    \begin{align*}
        \Phi(g'\ker(\varphi))&=\varphi(g')\\
        &=\varphi(gk)\\
        &=\varphi(g)\varphi(k)\\
        &=\varphi(g)\\
        &=\Phi(g\ker(\varphi)).
    \end{align*}
    Now, given $g,g'\in G$, one has
    \begin{align*}
        \Phi(g\ker(\varphi)\cdot g'\ker(\varphi))&=\Phi(gg'\ker(\varphi))\\
        &=\varphi(gg')\\
        &=\varphi(g)\varphi(g')\\
        &=\Phi(g\ker(\varphi))\Phi(g'\ker(\varphi)),
    \end{align*}
    which establishes that $\Phi$ is a homomorphism. For any $\varphi(g)\in\varphi(G)$, one finds 
    that $\Phi(g\ker(\varphi))=\varphi(g)$, so that $\Phi$ is surjective. Furthermore, if 
    $\Phi(g\ker(\varphi))=\Phi(g'\ker(\varphi))$, then $\varphi(g)=\varphi(g')$, which implies that 
    $\varphi(g^{-1}g')=1$, so that $g^{-1}g\in\ker(\varphi)$. Thus, $g^{-1}g'\ker(\varphi)=\ker(\varphi)$, 
    from which it follows that $g'\ker(\varphi)=g\ker(\varphi)$; that is, $\Phi$ is injective. 
    This establishes that $\Phi$ is an isomorphism. Hence, $G/\ker(\varphi)\cong\varphi(G)$.
\end{proof}

There are other types of groups and subgroups of interest. A group $G$ is called \textbf{simple} if 
it has no nontrivial, proper, normal subgroups; that is, the only normal subgroups are $G$ itself and 
the trivial group $\{1\}$. The \textbf{centralizer} of a subset $A\subset G$ is the subgroup
\begin{align*}
    C_G(A)=\{g\in G\ \lvert\ ga=ag\text{ for all }a\in A\}.
\end{align*}
To see that this is indeed a subgroup, let $g,g'\in C_G(A)$ and observe that $gg'a=gag'=agg'$. 
A special name is given to the centralizer of $G$ in itself. The set
\begin{align*}
    Z(G):=C_G(G)=\{g\in G\ \lvert\ gh=hg\text{ for all }h\in G\}
\end{align*}
is called the \textbf{center} of $G$, and it can be shown that this is a normal subgroup of $G$. 
Indeed, let $g\in G$ and $z\in Z(G)$ and observe that $gzg^{-1}=gg^{-1}z=z\in Z(G)$. 
The \textbf{direct product} $G\times H$ of groups $G$ and $H$ is the collection of elements of the 
form $(g,h)$ with $g\in G$ and $h\in H$ together with the binary operation defined by
\begin{align*}
    (g_1,h_1)(g_2,h_2)=(g_1g_2,h_1h_2),
\end{align*}
where it is understood that the operation in the first component is that of $G$, while the operation 
in the second component is that of $H$.

A subgroup of particular importance to us will be the \textbf{normalizer} of a subset $A\subset G$, 
which is defined as
\begin{equation}
    N_G(A)=\{g\in G\vert gAg^{-1}=A\}.
\end{equation}
To see that this is indeed a subgroup, let $g,g'\in N_G(A)$ and observe that
$gg'A(gg')^{-1}=gg'A(g')^{-1}g^{-1}=gAg^{-1}=A$.
Observe that if $A$ is a normal subgroup of $G$, then $gAg^{-1}=A$ for all $g\in G$.
Thus, the normalizer is $N_G(A)=G$. 
Observe also that if $g\in C_G(A)$, then $gAg^{-1}=gg^{-1}A=A$, and so $g\in N_G(A)$, and it follows 
that $C_G(A)<N_G(A)$.
%%%%%%%%%%%%%%%%%%%%%%%%%%%%%%%%%%%%%%%%%%%%%%%%%%%%%%%%%%%%%%%%%%%%%%%%%%%%%%%%%%%%%%%%%%%%%%%%%%%%%%%%%%%%%%%%%%%%%%%
\subsection{Representations and the Stabilizer Subgroup}\label{sec:groupactions}

We have given a general overview of the absolute basics of group theory, and our goal now is to give 
the theory some practical importance. It is clear that groups capture the notion of a symmetry.
We now need a formalism which allows a group to act on our system in the same way that we would expect a symmetry to.
To this end, let us define a representation of a group.

\begin{definition}
    Let $G$ be a group. A \textbf{representation} 
    of $G$ acting on a Hilbert space $\mathcal{H}$
    is a group homomorphism $\varphi:G\to GL(\mathcal{H})$.
\end{definition}

The homomorphism property of the representation allows us to transfer information
about the group structure on $G$ to its image $\varphi(G)$ in the general linear group $GL(\mathcal{H})$.
The general linear group is the group of all invertible linear transformations from $\mathcal{H}$ to itself.
When $\varphi$ is injective, its kernel is trivial and so by the first isomorphism theorem, 
we have that $G$ is isomorphic to $\varphi(G)$. 
For this reason, we call such a representation \textbf{faithful}; it embeds the entire structure of $G$ into $GL(\mathcal{H})$.
Since we are interested in letting our group act on quantum states, 
we will restrict our attention to representations that map group elements to
unitary operators so that the convention that quantum states are normalized is unaffected
by the action of such an element. Thus, we consider representations $W:G\to\mathcal{U}(\mathcal{H})\subset GL(\mathcal{H})$
which we call \textbf{unitary}.

\begin{definition}
    Let $W:G\to\mathcal{U}(\mathcal{H})$ be a unitary representation. 
    A quantum state $\ket{\psi}$ is called $G$-\textbf{invariant} with respect to the representation $W$ if
    \begin{equation}
        W(g)\ket{\psi}=\ket{\psi}
    \end{equation}
    for all $g\in G$. Likewise, a subspace $S$ of $\mathcal{H}$ is called $G$-invariant 
    if $W(g)\ket{\psi}=\ket{\psi}$ for all $\ket{\psi}\in S$.
\end{definition}

Elements of $\mathcal{H}$ are $G$-invariant when they are left unchanged by the action of a group $G$.
It is a common abuse of notation to drop the reference to the representation $W$ when it is clear from the context.
The collection of all such quantum states forms a subspace of the original 
Hilbert space called the $G$-\textbf{symmetric subspace} and denoted $\sym_{(W,G)}$ 
or $\sym_G$ when the representation is clear from the context. That this set is indeed
a subspace is made clear by the fact that if $\ket{\psi}$ and $\ket{\phi}$ are elements of $\sym_G$, 
then $W(g)(\ket{\psi}+\ket{\phi})=W(g)\ket{\psi}+W(g)\ket{\phi}=\ket{\psi}+\ket{\phi}$, 
and if $\lambda\in\mathbb{C}$, then $W(g)\lambda\ket{\psi}=\lambda W(g)\ket{\psi}=\lambda\ket{\psi}$. The following lemma is often helpful.

\begin{lemma}\label{lemma:gprojection}
    If $G$ is a finite group, then the projection $\Pi_G$ onto the $G$-symmetric subspace is given by
    \begin{equation}
        \Pi_G = \frac{1}{\lvert G\rvert}\sum_{g\in G} W(g).
    \end{equation}
\end{lemma}
\begin{proof}
    Let us first show that $\Pi$ is idempotent. We have
    \begin{align}
        \Pi_G^2 &= \frac{1}{\lvert G\rvert^2}\sum_{g\in G}\sum_{h\in G}W(g)W(h)\\
        &=\frac{1}{\lvert G\rvert^2}\sum_{g\in G}\sum_{h\in G}W(gh)\\
        &=\frac{1}{\lvert G\rvert^2}\sum_{g\in G}\sum_{h\in G}W(g)\\
        &=\frac{1}{\lvert G\rvert}\sum_{g\in G}W(g)=\Pi_G,
    \end{align}
    where in the third equality, we have used the fact that $Gh=G$; that is, the collection of all elements $gh$ with $g\in G$ is the same as the collection of all elements $g\in G$. This shows that $\Pi$ is a projection. To see that this projection is in fact onto the $G$-symmetric subspace, let $\ket{\psi}\in\imag(\Pi_G)$ and let $h\in G$. Then $\ket{\psi}=\Pi_G\ket{\phi}$ for some $\ket{\phi}$ and we have $W(h)\ket{\psi}=W(h)\Pi_G\ket{\phi}$. But notice that
    \begin{equation}
        W(h)\Pi_G=\frac{1}{\lvert G\rvert}W(h)\sum_{g\in G}W(g)=\frac{1}{\lvert G\rvert}\sum_{g\in G}W(hg)=\frac{1}{\lvert G\rvert}\sum_{g\in G}W(g)=\Pi_G.
    \end{equation}
    Thus, $W(h)\ket{\psi}=W(h)\Pi_G\ket{\phi}=\Pi_G\ket{\phi}=\ket{\psi}$, establishing that $\imag(\Pi_G)\subset\sym_G$. Now suppose $\ket{\psi}\in\sym_G$. Then $W(g)\ket{\psi}=\ket{\psi}$ for all $g\in G$. Thus, we have
    \begin{equation}
        \Pi_G\ket{\psi}=\frac{1}{\lvert G\rvert}\sum_{g\in G}W(g)\ket{\psi}=\frac{1}{\lvert G\rvert}\sum_{g\in G}\ket{\psi}=\ket{\psi},
    \end{equation}
    which establishes that $\sym_G\subset\imag(\Pi_G)$, thereby completing the proof.
\end{proof}

An important subgroup of $G$ is the \textbf{stabilizer} of an element of $\mathcal{H}$, the collection of all
$g$ which leave that element invariant. To see that this is indeed a subgroup of $G$, 
observe that if $g_1$ and $g_2$ leave $\ket{\psi}$ invariant, then $W(g_1g_2)\ket{\psi}=W(g_1)W(g_2)\ket{\psi}=\ket{\psi}$.
\begin{definition}
    Let $G$ be a group with unitary representation $W$ acting on $\mathcal{H}$ and let $\ket{\psi}\in\mathcal{H}$.
    The \textbf{stabilizer subgroup} associated to $\ket{\psi}$ is the subgroup of $G$ defined by
    \begin{equation}
        G_{\ket{\psi}}:=\{g\in G\, \vert\,  W(g)\ket{\psi}=\ket{\psi}\}.
    \end{equation}
\end{definition}
\noindent We will see this subgroup appear in the stabilizer codes introduced in the next section.

We would like an equivalence relation that tells us whether or not two representations should be considered the same.
In group theory, the structure of a group was preserved by the homomorphism property
and the remaining set structure was preserved by the bijection property.
Similarly, we can define a class of maps which preserve the structure of a representation.

\begin{definition}
    Let $W_1:G\to\mathcal{U}(\mathcal{H}_1)$ and $W_2:G\to\mathcal{U}(\mathcal{H}_2)$ be representations of some group $G$.
    A linear map $T:\mathcal{H}_1\to\mathcal{H}_2$ is called an \textbf{intertwining map} if it satisfies
    \begin{equation}\label{eq:equivariance}
        T\circ W_1(g)=W_2(g)\circ T.
    \end{equation}
\end{definition}

The linearity property ensures that the linear vector space structure is 
transferred from $\mathcal{H}_1$ to its image in $\mathcal{H}_2$. The \textbf{equivariance} 
property \eqref{eq:equivariance} ensures that the action of $G$ on $\mathcal{H}_1$ via $W_1$
is respected by the action of $G$ on $\mathcal{H}_2$ via $W_2$. If, in addition, we
ensure that $\mathcal{H}_1$ and $\mathcal{H}_2$ are equivalent at the set theoretic level
by demanding that $T$ is a bijection, then $T$ is called an \textbf{isomorphism} of representations,
and the representations $W_1$ and $W_2$ are called \textbf{equivalent}. The classification of representations of finite groups is easily tackled with
the theory of characters, but we will have no need to do so here. The interested
reader is instead referred to the excellent book by Steinberg for an elementary
introduction \cite{steinberg2011}.

%%%%%%%%%%%%%%%%%%%%%%%%%%%%%%%%%%%%%%%%%%%%%%%%%%%%%%%%%%%%%%%%%%%%%%%%%%%%%%%%%%%%%%%%%%%%%%%%%%%%%%%%%%%%%%%%%%%%%%%
\subsection{The Stabilizer Code Formalism} \label{sec:stab-formalism}
The stabilizer code proceeds in three main steps: encoding,
syndrome extraction, and decoding. 
We begin by creating redundancy by entangling the state $\ket{\psi}$ with which we are concerned with additional qubits,
creating a logical state $\ket{\psi}_L$ in a higher dimensional Hilbert space.
This step is known as \textbf{encoding}, and the higher dimensional Hilbert space $\mathcal{H}_{L}$ 
to which $\ket{\psi}_L$ belongs is known as the \textbf{code space}. Observe that the group $\mathcal{U}(\mathcal{H}_{L})$
of unitary operations acts on $\mathcal{H}_L$ by the obvious representation $W(U)\ket{\psi}_L=U\ket{\psi}_L$.
Thus, we can drop any mention of $W$ and understand $\mathcal{U}(\mathcal{H}_L)$-invariance as the property that $U\ket{\psi}_L=\ket{\psi}_L$ for all $U\in\mathcal{U}(\mathcal{H}_L)$. We will be concerned with unitary
operators which act trivially on the code space; that is, operators $U\in\mathcal{U}(\mathcal{H}_{L})$
such that $U\ket{\psi}_L=\ket{\psi}_L$ for all $\ket{\psi}_L\in\mathcal{H}_{L}$. In other words, 
operators which are contained in the stabilizer of every element of $\mathcal{H}_L$. Since the
intersection of arbitrarily many groups is a group, we see that this collection of operators, 
which we call $\mathcal{S}$, is itself a subgroup of $\mathcal{U}(\mathcal{H}_{L})$. It will be these
operations which we use in our syndrome extraction procedure to extract information about our logical
qubit state without destroying it.

Let us therefore move on to the extraction of a \textbf{syndrome},
by which we mean one or more bits of classical information obtained via a measurement
which yields information about whether or not an error has occurred after the encoding procedure.
For an error $E\in\mathcal{U}(\mathcal{H}_L)$, we select an operator $S\in\mathcal{S}$ which anticommutes
with $E$ and then apply the procedure in Figure~\ref{fig:syndromeextraction}. In the event that this
error occurs, the syndrome extraction procedure transforms the state of the system as follows. After
applying the first Hadamard gate to the ancillary qubit, the system is in the state
\begin{equation}
    E\ket{\psi}_L H\ket{0}_A=E\ket{\psi}_L\frac{1}{\sqrt{2}}(\ket{0}_A+\ket{1}_A).
\end{equation}
Now applying the controlled-$S$ gate produces the state
\begin{equation}
    \frac{1}{\sqrt{2}}(E\ket{\psi}_L\ket{0}_A+SE\ket{\psi}_L\ket{1}_A),
\end{equation}
and so after the remaining Hadamard gate on the ancillary register, the system is in the state
\begin{equation}\label{eq:stabilizer-syndrome-state}
    \frac{E+SE}{2}\ket{\psi}_L\ket{0}_A+\frac{E-SE}{2}\ket{\psi}_L\ket{1}_A.
\end{equation}
Since $E$ anticommutes with $S$ by assumption, we see that the first term vanishes, and the second
term becomes
\begin{equation}
    E\ket{\psi}_L\ket{1}_A,
\end{equation}
and now a measurement of the ancillary qubit produces the syndrome 1 with certainty and leaves the
state of the logical qubit unchanged. Notice also that if we take $E=\mathds{1}$ so that there is no error 
at all, then \eqref{eq:stabilizer-syndrome-state} instead reduces to 
\begin{equation}
    \ket{\psi}_L\ket{0}_A,
\end{equation}
and so a measurement of the ancillary qubit produces the syndrome 0 with certainty and leaves the
state of the logical qubit unchanged. We have therefore constructed a procedure for detecting the 
error $E$ and can now correct it by applying its inverse operation $E^\dagger$ when it occurs, 
a process which is known as \textbf{decoding}. Importantly, this procedure does not harm the state of 
the logical qubit.

\begin{figure}
    \centering
    \begin{quantikz}
        \lstick{$\ket{\psi}_L$} & & \gate[1]{S} & &\\
        \lstick{$\ket{0}_A$} & \gate[1]{H} & \ctrl{-1} & \gate[1]{H} & \meter{}\\
    \end{quantikz}
    \caption{
        Syndrome extraction procedure for a stabilizer code.
    }\label{fig:syndromeextraction}
\end{figure}
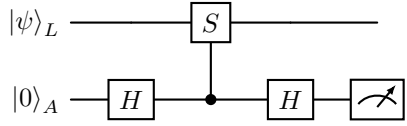

Now suppose we are concerned with two types of errors and label them $E_1$ and $E_2$. We find
operators $S_1,S_2\in\mathcal{S}$ such that $[S_1,S_2]=0$ and $S_i$ anticommutes with $E_j$ for $i=j$
and commutes with $E_j$ for $i\ne j$, and then we apply the procedure shown in Figure~\ref{fig:syndromeextraction2}. 
Suppose that $\ket{\psi}_L$ has been subject to an error $E\in\{\mathds{1},E_1,E_2\}$.
After the Hadamard gates on the ancillary register, we have the state
\begin{equation}
    \frac{1}{2}E\ket{\psi}_L(\ket{0}_{A_1}+\ket{1}_{A_1})(\ket{0}_{A_2}+\ket{1}_{A_2})
    =\frac{1}{2}(E\ket{\psi}_L\ket{00}_A+E\ket{\psi}_L\ket{01}_A+E\ket{\psi}_L\ket{10}_A+E\ket{\psi}_L\ket{11}_A),
\end{equation}
where we have introduced the notation $\ket{xy}_A=\ket{x}_{A_1}\ket{y}_{A_2}$ for brevity. Now
applying the controlled $S_1$ and $S_2$ operations produces the state
\begin{equation}
    \frac{1}{2}(E\ket{\psi}_L\ket{00}_A+S_2E\ket{\psi}_L\ket{01}_A+S_1E\ket{\psi}_L\ket{10}_A+S_2S_1E\ket{\psi}_L\ket{11}_A).
\end{equation}
Finally, applying Hadamard gates to the ancillary register again produces the state
\begin{align}
    \frac{1}{4}&((\mathds{1}+S_2+S_1+S_2S_1)E\ket{\psi}_L\ket{00}_A+(\mathds{1}-S_2+S_1-S_2S_1)E\ket{\psi}_L\ket{01}_A\\
    &+(\mathds{1}+S_2-S_1-S_2S_1)E\ket{\psi}_L\ket{10}_A+(\mathds{1}-S_2-S_1+S_2S_1)E\ket{\psi}_L\ket{11}_A).
\end{align}
Since $S_1$ and $S_2$ leave $\ket{\psi}_L$ invariant by assumption, we see that if $E=\mathds{1}$, then this 
state reduces to $\ket{\psi}_L\ket{00}_A$, and we are guaranteed that a measurement of the ancillary 
qubits will produce the syndrome 00 while leaving the state of the logical qubit unaffected. 
If $E=E_1$, then using the assumption that $E_1$ anticommutes with $S_1$, the state before measurement reduces to
\begin{align}
    \frac{1}{4}&((\mathds{1}+S_2-\mathds{1}-S_2)E_1\ket{\psi}_L\ket{00}_A+(\mathds{1}-S_2-\mathds{1}+S_2)E_1\ket{\psi}_L\ket{01}_A\\
    &+(\mathds{1}+S_2+\mathds{1}+S_2)E_1\ket{\psi}_L\ket{10}_A+(\mathds{1}-S_2+\mathds{1}-S_2)E_1\ket{\psi}_L\ket{11}_A),
\end{align}
which is equivalent to
\begin{equation}
    \frac{1}{2}((\mathds{1}+S_2)E_1\ket{\psi}_L\ket{10}_A+(\mathds{1}-S_2)E_1\ket{\psi}_L\ket{11}_A).
\end{equation}
Now using the assumption that $E_1$ commutes with $S_2$, we are left with the state $E_1\ket{\psi}_L\ket{10}_A$,
so that a measurement on the ancillary register produces the syndrome 10 with certainty and leaves
the state of the logical qubit invariant. Similarly, if $E=E_2$, then the state before measurement is
$E_2\ket{\psi}_L\ket{10}_A$, which will produce error syndrome 10 upon measurement and leave the
state of the logical qubit invariant. Finally, if $E=E_1E_2$, then the state before measurement is
$E_1E_2\ket{\psi}_L\ket{11}_A$, which will produce error syndrome 11 upon measurement and again
leave the state of the logical qubit invariant. Thus, our code can identify and correct $E_1,E_2,$ and
$E_1E_2$ errors.

\begin{figure}
    \centering
    \begin{quantikz}
        \lstick{$\ket{\psi}_L$}  &             & \gate[1]{S_1} & \gate[1]{S_2}&             &\\
        \lstick{$\ket{0}_{A_1}$} & \gate[1]{H} & \ctrl{-1}     &              & \gate[1]{H} &\meter{}\\
        \lstick{$\ket{0}_{A_2}$} & \gate[1]{H} &               & \ctrl{-2}    & \gate[1]{H} &\meter{}
    \end{quantikz}
    \caption{
        Syndrome extraction procedure for a stabilizer code concerned with two errors.
    }\label{fig:syndromeextraction2}
\end{figure}
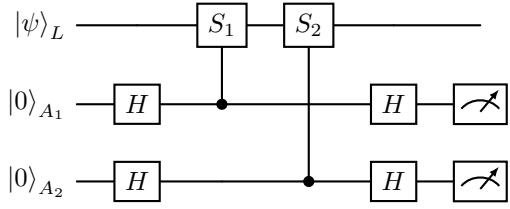

It is clear how to proceed inductively to incorporate an arbitrary number of operators $S_i$ capable
of detecting and correcting multiple errors. The codes we have outlined are called 
\textbf{stabilizer codes}. Let us formalize the procedure for constructing a stabilizer code. Because
the only errors we are concerned with are Pauli-type errors, we restrict our attention to the so-called
Pauli group $G_k$, which is the group generated by all possible $k$-fold tensor products of the Pauli
matrices up to a multiplicative factor $\lambda\in\{\pm1,\pm i\}$. Rather than defining the code space first, let us choose an abelian subgroup $\mathcal{S}$ of the Pauli 
group not containing the element $-1$ and call it the \textbf{stabilizer}. We then take the 
\textbf{code space} to be the $\mathcal{S}$-symmetric subspace $\textnormal{Sym}_{\mathcal{S}}$ of 
$\mathbb{C}^{2^k}$, thereby making sense of this terminology. The assumption that the stabilizer
does not contain $-1$ is made because if it did contain $-1$, then the $\mathcal{S}$-symmetric subspace
would be trivial.

This is an important point that is worth repeating; our stabilizer code is defined by a choice of 
stabilizer. The code space does not need to be defined in addition to the stabilizer, it is constructed
from the choice of stabilizer. We now link the dimension of the symmetric subspace to the number
of generators of $\mathcal{S}$.

\begin{proposition}\label{prop:sym-subspace-dim}
    Suppose the stabilizer $\mathcal{S}\le G_k$ has a minimal generating set of length $\ell$. Then the dimension of the
    $\mathcal{S}$-symmetric subspace is $\dim(\sym_{\mathcal{S}})=2^{k-\ell}$.
\end{proposition}
\begin{proof}
    Let $\{S_1,\ldots, S_\ell\}$ be a set of generators of $\mathcal{S}$. Observe that $P_i:=\frac{\mathds{1}+S_i}{2}$
    is an orthogonal projection onto the $S_i$-symmetric subspace (the symmetric subspace associated to the cyclic group
    generated by $S_i$). To see this, note that $S_i^2=\mathds{1}$ and that $S_i$ is self-adjoint since it 
    is an element of the Pauli group. Then $P_i^2=\left(\frac{\mathds{1}+S_i}{2}\right)^2=\frac{\mathds{1}+S_i}{2}=P_i$ 
    and $P_i$ is self-adjoint. Thus, $P_i$ is an orthogonal projection, and it remains for us to show
    that its image is the $S_i$-symmetric subspace. Certainly if $\ket{\psi}$ is in the image of $P_i$,
    then $\ket{\psi}=P_i\ket{\phi}$ for some $\ket{\phi}$. Then we have
    \begin{equation}
        S_i\ket{\psi}=S_i\frac{\mathds{1}+S_i}{2}\ket{\phi}=\frac{S_i+S_i^2}{2}\ket{\phi}=\frac{\mathds{1}+S_i}{2}\ket{\phi}=\ket{\psi},
    \end{equation}
    and so $\ket{\psi}$ is in the $S_i$-symmetric subspace. Conversely, if $\ket{\psi}$ is in the $S_i$-symmetric 
    subspace, then $S_i\ket{\psi}=\ket{\psi}$, and it follows that
    \begin{equation}
        P_i\ket{\psi}=\frac{\mathds{1}+S_i}{2}\ket{\psi}=\frac{1}{2}\ket{\psi}+\frac{1}{2}S_i\ket{\psi}=\ket{\psi},
    \end{equation}
    so that $\ket{\psi}$ is in the image of $P_i$. We have therefore shown that $P_i$ is an orthogonal
    projection onto the $S_i$-symmetric subspace. Moreover, the product of all $\ell$ of these projections
    is an orthogonal projection onto the intersection of the symmetric subspaces; that is, the $\mathcal{S}$-symmetric 
    subspace $\textnormal{Sym}_\mathcal{S}$. 
    
    To determine the dimension of $\textnormal{Sym}_\mathcal{S}$, let
    us analyze how the dimension is affected at each step of the intersection. First observe that
    $P_1$ projects our Hilbert space down to $\textnormal{Sym}_{S_1}$. Since $S_1$ is an element of the Pauli group,
    it is (up to a constant) a tensor product of Pauli matrices $S_1=A^{(1)}_1\otimes\cdots\otimes A^{(1)}_k$, 
    each of which has eigenvalues $\pm1$. Let the eigenvalue of $A^{(1)}_i$ be $\lambda_i$ with eigenvector $\ket{\lambda_i}$. 
    Then the eigenvalues of $S_1$ are $\lambda=\lambda_1\cdots\lambda_k=\pm1$ with eigenvectors 
    $\ket{\lambda}=\ket{\lambda_1}\cdots\ket{\lambda_k}$ given by the tensor product of the eigenvectors 
    for each of the Pauli matrices. Each such eigenvector $\ket{\lambda}$ of $S_1$ has eigenvalue $+1$ 
    if it is made up of an even number of Pauli eigenvectors $\ket{\lambda_i}$ with $\lambda_i=-1$.
    Similarly, each eigenvector of $S_1$ has eigenvalue $-1$ if it is made up of an odd number of Pauli
    eigenvectors with $\lambda_i=-1$. Then by symmetry, the $+1$ and $-1$ eigenspaces of $S_1$ have
    the same dimension. Since $\mathcal{H}$ is a direct sum of these spaces, their dimensions must be
    $2^{k-1}$ (that is, half of $\dim(\mathcal{H})$). Moreover, the rank of $P_1$ is the dimension of the $+1$ 
    eigenspace of $S_1$. Thus $\dim(\sym_{S_1})=2^{k-1}$. Now applying the same procedure with $\mathcal{H}$
    replaced by $\sym_{S_1}$ and $S_1$ replaced by $S_2$, we see that $\dim(\sym_{S_1\cap S_2})=2^{k-2}$.
    Continuing in this way until all $\ell$ generators have been used, we see that
    $\dim(\sym_\mathcal{S})=2^{k-\ell}$.
\end{proof}

Proposition~\ref{prop:sym-subspace-dim} tells us that for every generator, the dimension of the code space
is cut in half. Thus, if we embed a data qubit into a Hilbert space of dimension $2^k$, we find that
the code space describes exactly one logical qubit when the number of generators of the Stabilizer subgroup
is $\ell=k-1$. 

Suppose we care about the errors $\{E_1,\ldots,E_m\}$, each of which is an element of the Pauli group.
Since both the generators of the stabilizer and the possible errors are elements of the Pauli group,
they must either commute or anticommute. We call elements of a basis for the code space \textbf{code words}. Let
$\ket{\psi}_L$ and $\ket{\phi}_L$ be two such code words and let them pass through our communication
channel, picking up an error along the way so that the channel outputs the states $E_i\ket{\psi}_L$ and 
$E_j\ket{\phi}_L$, respectively for some $i,j$. Then
\begin{equation}\label{eq:error-rotates}
    \bra{\psi}S_kE_i^\dagger E_j\ket{\phi}=\bra{\psi}E_i^\dagger E_j\ket{\phi}=\bra{\psi}E_i^\dagger E_jS_k\ket{\phi}
\end{equation} 
for any generator $S_k$ of the stabilizer subgroup. Since $E_i^\dagger E_j$ and $S_k$ are elements
of the Pauli group, they must either commute or anticommute. If they anticommute, then from \eqref{eq:error-rotates} we see that
\begin{equation}
    \bra{\psi}E_i^\dagger E_j\ket{\phi}=-\bra{\psi}E_i^\dagger E_j\ket{\phi},
\end{equation}
which can only be true if $\bra{\psi}E_i^\dagger E_j\ket{\phi}=0$. The errors have rotated the code
words into mutually orthogonal subspaces! This property allows us to distinguish between our codewords
even after they have been subject to error. Moreover, when $\ket{\psi}=\ket{\phi}$ and the errors are
different, we find that $E_i\ket{\psi}$ and $E_j\ket{\psi}$ are in orthogonal subspaces, allowing us
to distinguish between $E_i$ and $E_j$ so that each error can be corrected. Let us formalize all that
we have just shown, rephrasing it in terms of the normalizer of the stabilizer subgroup.

\begin{theorem}\label{thm:stab-error}
    The code defined by the stabilizer subgroup $\mathcal{S}\le G_k$ corrects the errors $\{E_1,\ldots,E_m\}\subset G_k$
    if $E_iE_j\in G_k\setminus N(\mathcal{S})$ for all $i,j=1,\ldots,m$.
\end{theorem}
\begin{proof}
    We first show that $N(\mathcal{S})=C(\mathcal{S})$. To see this, let $A\in N(\mathcal{S})$. Then
    for any $S\in \mathcal{S}$, we have $ASA^\dagger\in \mathcal{S}$. But $A$ and $S$ either commute
    or anticommute. But if they anticommute, then $ASA^\dagger=-SAA^\dagger=-S$ is in $\mathcal{S}$.
    Since $\mathcal{S}$ is a group, it contains $S^{-1}$, and so $-1$ is in $\mathcal{S}$, which
    contradicts the assumption that $\mathcal{S}$ is a stabilizer subgroup. It follows that $A$ and $S$
    commute and so $ASA^\dagger=S$, so that $A\in C(\mathcal{S})$. Conversely, if $S\in C(\mathcal{S})$,
    then $ASA^\dagger=S\in\mathcal{S}$ and so $A\in N(\mathcal{S})$. Thus, $N(\mathcal{S})=C(\mathcal{S})$
    as we set out to show.

    We now examine the syndrome extraction procedure in Figure~\ref{fig:syndromeextraction-general} in the event of an error $E$. Since all of the generator $S_i$ commute, each Hadamard, controlled-stabilizer, Hadamard sequence on each ancillary qubit is independent from the others. The result of the application of the first stabilizer sequence is the state
    \begin{equation} \label{eq:first-stab}
        \frac{1}{2}(\mathds{1}+S_1)E\ket{\psi}_L\ket{0}_{A_1}+\frac{1}{2}(\mathds{1}-S_1)E\ket{\psi}_L\ket{1}_{A_1},
    \end{equation}
    and after the second stabilizer sequence, this becomes
    \begin{align}
        \frac{1}{4}&(\mathds{1}+S_2)(\mathds{1}+S_1)E\ket{\psi}_L\ket{00}_{A_1A_2}+\frac{1}{4}(\mathds{1}-S_2)(\mathds{1}+S_1)E\ket{\psi}_L\ket{01}_{A_1A_2}\\
        &+\frac{1}{4}(\mathds{1}+S_2)(\mathds{1}-S_1)E\ket{\psi}_L\ket{10}_{A_1A_2}+\frac{1}{4}(\mathds{1}-S_2)(\mathds{1}-S_1)E\ket{\psi}_L\ket{11}_{A_1A_2}.
    \end{align}
    Continuing inductively, we see that after the entire syndrome extraction procedure, the resulting state is
    \begin{equation}
        \frac{1}{2^{\ell}}\sum_{i_1,\ldots,i_\ell=0}^1\left(\prod_{j=1}^\ell(\mathds{1}+(-1)^{i_j}S_j)\right) E\ket{\psi}_L\ket{i_1\cdots i_\ell}_A.
    \end{equation}
    Now, each choice of $E$ either commutes or anticommutes with each generator $S_j$. Thus, the only term in the sum that survives is the one for which $i_j=0$ when $E$ commutes with $S_j$ and $i_j=1$ when $E$ anticommutes with $S_j$. Let $\alpha_1,\ldots,\alpha_\ell$ denote this choice of indices. Then the result of the syndrome extraction procedure is the state $\ket{\psi}_L\ket{\alpha_1\cdots\alpha_\ell}$, and a measurement of the ancillary qubits will produce the syndrome $\alpha_1\cdots\alpha_\ell$ with certainty.

    It remains to be shown that for each choice $E\in\{E_1,\ldots,E_m\}$, the syndrome $\alpha_1\cdots\alpha_\ell$ is unique. Take two arbitrary errors $E_\alpha$ and $E_\beta$ and let their corresponding syndromes be $\alpha_1\cdots\alpha_\ell$ and $\beta_1\cdots\beta_\ell$, respectively. By assumption, we have $E_\alpha E_\beta\in G_k\setminus N(\mathcal{S})$. Since $N(\mathcal{S})=C(\mathcal{S})$, this is equivalent to the statement that $E_\alpha E_\beta$ anticommutes with some element of the stabilizer. Thus, there is an $S\in\mathcal{S}$ such that $E_\alpha E_\beta S=-SE_\alpha E_\beta$. Now, $S$ is a product of the generators of $\mathcal{S}$, and so we may write $E_\alpha E_\beta S_1^{k_1}\cdots S_\ell^{k_\ell}=-S_1^{k_1}\cdots S_\ell^{k_\ell}E_\alpha E_\beta$ for some choice of $(k_1,\ldots,k_\ell)\in\{0,1\}^{\ell}$. But then we have
    \begin{equation}
        -S_1^{k_1}\cdots S_\ell^{k_\ell}E_\alpha E_\beta=E_\alpha E_\beta S_1^{k_1}\cdots S_\ell^{k_\ell}= (-1)^{k_1(\alpha_1+\beta_1)+\cdots+k_\ell(\alpha_\ell+\beta_\ell)}S_1^{k_1}\cdots S_\ell^{k_\ell}E_\alpha E_\beta,
    \end{equation}
    and so it follows that $\alpha_i\ne \beta_i$ for some $i$.
\end{proof}

\begin{figure}
    \centering
    \begin{quantikz}
        \lstick{$E\ket{\psi}_L$}    &             & \gate[1]{S_1} & \gate[1]{S_2}& \ \ldots \ & \gate[1]{S_\ell} &            &\\
        \lstick{$\ket{0}_{A_1}$}    & \gate[1]{H} & \ctrl{-1}     &              &            &                  &\gate[1]{H} &\meter{}\\
        \lstick{$\ket{0}_{A_2}$}    & \gate[1]{H} &               & \ctrl{-2}    &            &                  &\gate[1]{H} &\meter{}\\
        \setwiretype{n}             & \ \vdots \  &               &              & \ \ddots \ &                  &\vdots \\
        \lstick{$\ket{0}_{A_\ell}$} & \gate[1]{H} &               &              &            & \ctrl{-4}        &\gate[1]{H} &\meter{}\\
    \end{quantikz}
    \caption{
        Syndrome extraction procedure for a general stabilizer code.
    }\label{fig:syndromeextraction-general}
\end{figure}
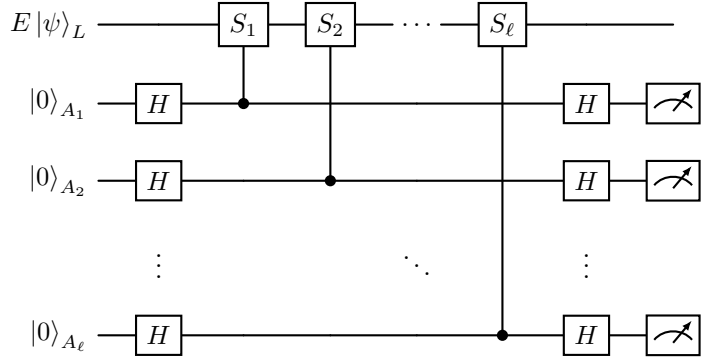

Sometimes the syndrome extraction procedure is referred to as ``measuring the syndrome''. This terminology can be justified by analyzing the effect of the syndrome extraction procedure on the errored logical state $E\ket{\psi}_L$. As can be seen in \eqref{eq:first-stab}, the syndrome extraction forces the logical register into either a projection onto the symmetric subspace of the stabilizer generator $S$ or a projection onto the asymmetric subspace (the complement of the symmetric subspace). The expected value of $Z$ when measuring the ancillary qubit is therefore
\begin{equation}
    \bra{\psi}E^\dagger\frac12(\mathds{1}+S)E\ket{\psi}_L-\bra{\psi}E^\dagger\frac12(\mathds{1}-S)E\ket{\psi}_L=\bra{\psi}E^\dagger SE\ket{\psi}_L,
\end{equation}
which is exactly the expected value of $S$ when measuring the logical register. The key difference between taking these two measurements is that the measurement on the ancillary register obtains information without disturbing the state of the logical register (when $S$ is a stabilizer and $E$ is an error). Thus, if we want to continuously correct errors, we should apply the syndrome extraction procedure instead of measuring the state directly, but the information content of the measurement is the same and so the terminology is justified.

The stabilizer, its normalizer, and the Pauli group behave quite nicely together, allowing us to take quotients consisting of cosets which themselves form a group. In particular, we have the following lemma.

\begin{lemma}\label{lemma:normal-subgroups}
    Let $\mathcal{S}$ be a stabilizer subgroup of $G_k$ for some $k$ and let $N(\mathcal{S})$ denote the normalizer of this stabilizer. Then $\mathcal{S}$ is a normal subgroup of $N(\mathcal{S})$ and $N(\mathcal{S})$ is a normal subgroup of $G_k$.
\end{lemma}
\begin{proof}
    Certainly $\mathcal{S}$ is a subgroup of $N(\mathcal{S})=C(\mathcal{S})$ since $\mathcal{S}$ is abelian. Let $S\in\mathcal{S}$ and let $N\in N(\mathcal{S})$. Then $NSN^{-1}\in\mathcal{S}$ by definition, and so $\mathcal{S}$ is a normal subgroup of $N(\mathcal{S})$. Now let $P\in G_k$ and note that $N(\mathcal{S})$ is trivially a subgroup of $G_k$. Since all elements of the Pauli group either commute or anticommute, we have $PNP^{-1}=\pm N$. If $PNP^{-1}=N$, then $PNP^{-1}\in N(\mathcal{S})$ and we are done. If $PNP^{-1}=-N$, observe that $-NS(-N)^{-1}=NSN^{-1}\in\mathcal{S}$, and so $PNP^{-1}\in N(\mathcal{S})$. Thus, $N(\mathcal{S})$ is a normal subgroup of $G_k$.
\end{proof}

Let $\ket{\psi}_L$ be an element of the code space. By definition, the stabilizer fixes $\ket{\psi}_L$. Now let us take a look at the remaining elements of the normalizer $N(\mathcal{S})$. We are guaranteed only that $NSN^{-1}\in\mathcal{S}$ for all $S\in\mathcal{S}$. Thus, $NSN^{-1}\ket{\psi}_L=\ket{\psi}_L$ for every $S\in\mathcal{S}$, or equivalently, $N\ket{\psi}_L=S^{-1}N\ket{\psi}_L$ for every $S\in\mathcal{S}$. Therefore, $N\ket{\psi}_L$ is in our code space, but isn't necessarily equivalent to $\ket{\psi}_L$. But this is a feature, not a bug! The elements of $N(\mathcal{S})\setminus\mathcal{S}$ transform elements of the code space to another element of the code space and can therefore be considered \textbf{logical operators} which act on our logical qubits. Moreover, as the next proposition shows, logical operators which belong to the same coset of $N(\mathcal{S})/\mathcal{S}$ have the same effect on elements of the code space and are therefore different implementations of the same logical operation.

\begin{proposition}
    Let $N(\mathcal{S})$ be the normalizer of the stabilizer subgroup $\mathcal{S}$. Then
    \begin{equation}
        N\ket{\psi}_L=M\ket{\psi}_L
    \end{equation}
    for every $\ket{\psi}_L\in\mathcal{H}_L$ whenever $N$ and $M$ belong to the same coset in $N(\mathcal{S})/\mathcal{S}$.
\end{proposition}
\begin{proof}
    If $N$ and $M$ belong to the same coset, then there exists $S,S'\in\mathcal{S}$ such that $NS=MS'$, and it follows that
    \begin{equation}
        N\ket{\psi}_L=NS\ket{\psi}_L=MS'\ket{\psi}_L=M\ket{\psi}_L
    \end{equation}
    for every $\ket{\psi}_L\in\mathcal{H}_L$.
\end{proof}

This proposition allows us to think about cosets of the stabilizer in its normalizer as logical operations in a well-defined way. By Lagrange's theorem, there are $\lvert N(\mathcal{S})\rvert/\lvert\mathcal{S}\rvert$ such cosets. We define the logical operators to be these cosets, and we consider any representative of such a coset to be one of a number of implementations of this logical operation. Since the distance of a code is defined to be the minimum weight of a logical operation that transforms one code word to another, we have the following useful result.

\begin{proposition}\label{prop:distance}
    The distance of a stabilizer code with stabilizer $\mathcal{S}$ is the minimum weight over all non-trivial elements in $N(\mathcal{S})/\mathcal{S}$. That is,
    \begin{equation}
        d=\min\{w(N):\mathds{1}\ne N\in N(\mathcal{S})/\mathcal
        S\},
    \end{equation}
    where $w(N)$ denotes the number of non-identity Pauli operators in $N$.
\end{proposition}

The other quotient group that we are able to form using Lemma~\ref{lemma:normal-subgroups} is $G_k/N(\mathcal{S})$. The elements of $G_k\setminus N(\mathcal{S})$ are the possible errors that can be detected. Suppose the errors $E$ and $F$ belong to the same coset of $N(\mathcal{S})$ in $G_k$. Then the stabilizer code produces the same syndrome for $E$ and $F$.

\begin{proposition}\label{prop:errors}
    Let $E$ and $F$ be errors belonging to the same coset of $N(\mathcal{S})$ in $G_k$. Then the stabilizer code produces the same syndrome for $E$ and $F$.
\end{proposition}
\begin{proof}
    Let $E$ and $F$ have syndromes $\alpha_1\cdots\alpha_\ell$ and $\beta_1\cdots\beta_\ell$, respectively. There are $N,M\in N(\mathcal{S})$ such that $EN=FM$. Letting $S_i\in\mathcal{S}$ and recalling that $N(\mathcal{S})=C(\mathcal{S})$, we see that $ENS_i=ES_iN=(-1)^{\alpha_i}S_iEN=(-1)^{\alpha_i}S_iFM=(-1)^{\alpha_i+\beta_i}FMS_i=(-1)^{\alpha_i+\beta_i}ENS_i$. It follows that $\alpha_i=\beta_i$ for all $i=1,\ldots,\ell$.
\end{proof}

This proposition organizes the possible errors into equivalence classes according to their syndrome. It is tempting to say that the stabilizer code then corrects $\lvert G_k\rvert/\lvert N(\mathcal{S})\rvert$ errors by Lagrange's theorem. However, sometimes errors in the same class can be corrected by the same operation. Indeed, it is often the case that a correction of some specific error has the effect of correcting another error in the same class, just as the operators $I\otimes Z$ and $Z\otimes I$ both map the state $-\ket{11}$ to the state $\ket{11}$. When this is the case (as in the Shor code), we call the code \textbf{degenerate}.

To summarize the findings of this section, the stabilizer code is defined by a choice of stabilizer $\mathcal{S}$, an abelian subgroup of the Pauli group not containing the element $-1$. This choice completely defines the $\mathcal{S}$-symmetric subspace, which is taken to be the code space $\mathcal{H}_L$. Thus, our encoder will be chosen to transform an arbitrary data state to an element of the $\mathcal{S}$-symmetric subspace, and our syndrome extraction procedure is given in Figure~\ref{fig:syndromeextraction-general}. This code detects and corrects a set of errors which are in the complement of the normalizer of $\mathcal{S}$ and belong to distinct cosets of the normalizer in the Pauli group. Meanwhile, elements of the normalizer act as logical operations on the logical state, allowing us to treat the logical state as the state of a single error-resistant logical qubit and the logical operations as effectively single qubit operations on the logical qubit. We will illustrate this point in the next section with another example.

%%%%%%%%%%%%%%%%%%%%%%%%%%%%%%%%%%%%%%%%%%%%%%%%%%%%%%%%%%%%%%%%%%%%%%%%%%%%%%%%%%%%%%%%%%%%%%%%%%%%%%%%%%%%%%%%%%%%%%%
\subsection{Another Example - The 5 Qubit Code}\label{sec:5qubitcode}

Although the Shor code was the first to protect against all single qubit unitary errors, it is not the smallest such code. That title belongs to the following 5-qubit code. We define the stabilizer subgroup to be the group generated by $X_1Z_2Z_3X_4$, $X_2Z_3Z_4X_5$, $X_1X_3Z_4Z_5$, and $Z_1X_2X_4Z_5$. It is sometimes helpful to use instead the notation $XZZXI:=X_1Z_2Z_3X_4$ with the understanding that the first operation acts on the first qubit, and so on\footnote{We have dropped the symbol $\mathds{1}$ in favor of $I$ because this is what commonly appears in the literature}. Thus, the stabilizer subgroup is $\mathcal{S}=\langle XZZXI, IXZZX, XIXZZ, ZXIXZ\rangle$. Recalling that the stabilizer defines the codespace (it is the $\mathcal{S}$-symmetric subspace), we compute the possible logical states by demanding invariance under the action of the generators of the stabilizer. To this end, let
\begin{equation}
    \ket{\psi}=\sum_{i,j,k,l,m=0}^1 c_{ijklm}\ket{ijklm}.
\end{equation}
Letting the generators act on $\ket{\psi}$ produces
\begin{align}
    \sum_{i,j,k,l,m=0}^1 c_{ijklm}\ket{ijklm}&=\sum_{i,j,k,l,m=0}^1 (-1)^{j+k}c_{i\oplus1,j,k,l\oplus1,m}\ket{ijklm}\\
    \sum_{i,j,k,l,m=0}^1 c_{ijklm}\ket{ijklm}&=\sum_{i,j,k,l,m=0}^1 (-1)^{k+l}c_{i,j\oplus1,k,l,m\oplus1}\ket{ijklm}\\
    \sum_{i,j,k,l,m=0}^1 c_{ijklm}\ket{ijklm}&=\sum_{i,j,k,l,m=0}^1 (-1)^{l+m}c_{i\oplus1,j,k\oplus1,l,m}\ket{ijklm}\\
    \sum_{i,j,k,l,m=0}^1 c_{ijklm}\ket{ijklm}&=\sum_{i,j,k,l,m=0}^1 (-1)^{m+i}c_{i,j\oplus1,k,l\oplus1,m}\ket{ijklm}.
\end{align}
There are two independent solutions to the linear system of equations that arise from identifying coefficients of basis vectors. The first, which we take as the logical zero state, is
\begin{align}
    \ket{0}_L=\frac{1}{4}(&\ket{00000}+\ket{10010}+\ket{01001}+\ket{10100}+\ket{01010}+\ket{00101}-\ket{00110}-\ket{11000}\\
    &-\ket{11101}-\ket{00011}-\ket{11110}-\ket{01111}-\ket{01100}-\ket{10111}-\ket{11011}-\ket{10001})\notag
\end{align}
while the second state, which we take as the logical one state, is
\begin{align}
    \ket{1}_L=\frac{1}{4}(&\ket{11111}+\ket{01101}+\ket{10110}+\ket{01011}+\ket{10101}+\ket{11010}-\ket{11001}-\ket{00111}\\
    &-\ket{00010}-\ket{11100}-\ket{00001}-\ket{10000}-\ket{01110}-\ket{10011}-\ket{01000}-\ket{00100}.\notag
\end{align}
Using these logical basis states, we encode a state $\ket{\psi}=\alpha\ket{0}+\beta\ket{1}$ as $\ket{\psi}_L=\alpha\ket{0}_L+\beta\ket{1}_L$.

Recall that the normalizer of the stabilizer subgroup is in fact the centralizer of the stabilizer subgroup. Thus, the logical operations (up to a factor of $\lambda\in\{\pm1,\pm i\}$) provided by this code are the operations $\sigma_1\sigma_2\sigma_3\sigma_4\sigma_5$ with $\sigma_i\in\{I,X,Y,Z\}$ satisfying
\begin{align}
    \sigma_1\sigma_2\sigma_3\sigma_4\sigma_5X_1Z_2Z_3X_4&=X_1Z_2Z_3X_4\sigma_1\sigma_2\sigma_3\sigma_4\sigma_5\\
    \sigma_1\sigma_2\sigma_3\sigma_4\sigma_5X_2Z_3Z_4X_5&=X_2Z_3Z_4X_5\sigma_1\sigma_2\sigma_3\sigma_4\sigma_5\\
    \sigma_1\sigma_2\sigma_3\sigma_4\sigma_5X_1X_3Z_4Z_5&=X_1X_3Z_4Z_5\sigma_1\sigma_2\sigma_3\sigma_4\sigma_5\\
    \sigma_1\sigma_2\sigma_3\sigma_4\sigma_5Z_1X_2X_4Z_5&=Z_1X_2X_4Z_5\sigma_1\sigma_2\sigma_3\sigma_4\sigma_5.
\end{align}
But the commutation of each $\sigma_i$ past each $X_i$ or $Z_i$ can introduce a sign change. Let us write $\sigma_iX_i=(-1)^{\alpha_i}X_i\sigma_i$ and $\sigma_iZ_i=(-1)^{\beta_i}Z_i\sigma_i$ with $\alpha_i,\beta_i\in\{0,1\}$. Then we have
\begin{align}
    \alpha_1+\beta_2+\beta_3+\alpha_4&=0,2,\textnormal{ or }4\\
    \alpha_2+\beta_3+\beta_4+\alpha_5&=0,2,\textnormal{ or }4\\
    \alpha_1+\alpha_3+\beta_4+\beta_5&=0,2,\textnormal{ or }4\\
    \beta_1+\alpha_2+\alpha_4+\beta_5&=0,2,\textnormal{ or }4,
\end{align}
and any solution to this system of equations defines a logical operation. The logical operation transforming $\ket{0}_L$ to $\ket{1}_L$ and vice versa is called the logical $X$ gate and takes the form
\begin{equation}
    \overline{X}=XXXXX.
\end{equation}
Similarly, the logical $Z$ gate takes $\ket{0}_L$ to $\ket{0}_L$ and $\ket{1}_L$ to $-\ket{1}_L$ and is defined by
\begin{equation}
    \overline{Z}=ZZZZZ.
\end{equation}

To complete the construction of the 5 qubit code, we need a method for encoding the logical state $\ket{\psi}_L$. Finding a simple circuit to accomplish this task can be difficult for a general stabilizer code. The reader can verify that the simple circuit\footnote{This encoder circuit is credited to Stack Exchange user tsgeorgios, who devised it as an answer to \hyperlink{https://quantumcomputing.stackexchange.com/questions/14264/nielsenchuang-5-qubit-quantum-error-correction-encoding-gate}{this post}.} just before the noise channel in Figure~\ref{fig:5qubita} accomplishes this task. The reader can check that $ZXZII$ is a logical operation and that no error of weight two can be a logical operation, so that the distance of this code is $d=3$. Thus, the 5-qubit code is a $[[5,1,3]]$ code with rate $R=1/5$.

To figure out what error corresponds to what syndrome, we check which generators of the stabilizer each error anticommutes with. If $E$ anticommutes with stabilizer $S_i$, then the $i$-th bit in the syndrome is 1. For example, the error $Z_1$ anticommutes with $XZZXI$ and $XIXZZ$, which are the first and third generators in the syndrome extraction portion of Figure~\ref{fig:5qubita}. Thus, the syndrome associated to $Z_1$ is 1010. We then correct the error by controlling for this syndrome and applying the inverse of the error. Since there are four ancillary qubits (because there are four generators), there are $2^4=16$ possible syndromes. We list all possible syndromes and the corresponding single qubit error in Table~\ref{tab:5qubitsyndromes}. We see that all single qubit errors are accounted for, and this table can be used to verify that the errors are corrected properly in Figure~\ref{fig:5qubitb}.

\begin{table}[]
    \centering
    \begin{tabular}{cc} \toprule
        Syndrome & Error \\ \midrule
        0000 & $I$\\
        0001 & $X_1$\\
        0010 & $Z_3$\\
        0011 & $X_5$\\
        0100 & $Z_5$\\
        0101 & $Z_2$\\
        0110 & $X_4$\\
        0111 & $Y_5=iX_5Z_5$\\ \bottomrule
    \end{tabular}
    \hspace{2em}
    \begin{tabular}{cc} \toprule
        Syndrome & Error \\ \midrule
        1000 & $X_2$\\
        1001 & $Z_4$\\
        1010 & $Z_1$\\
        1011 & $Y_1=iX_1Z_1$\\
        1100 & $X_3$\\
        1101 & $Y_2=iX_2Z_2$\\
        1110 & $Y_3=iX_3Z_3$\\
        1111 & $Y_4=iX_4Z_4$\\ \bottomrule
    \end{tabular}
    \caption{Syndrome look-up table for 5-qubit code.}
    \label{tab:5qubitsyndromes}
\end{table}

\begin{figure}
    \centering

    \begin{subfigure}{\textwidth}
        \centering
        \includegraphics[height=1.9in]{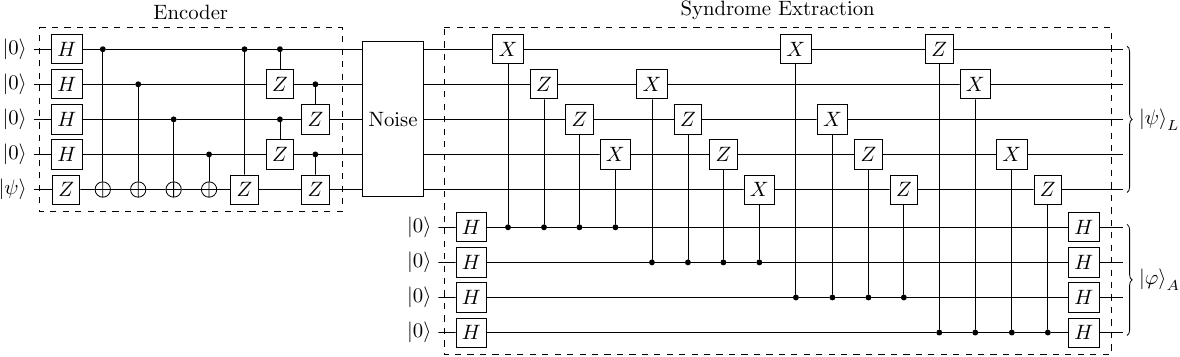}
        \caption{Circuit diagram (first half).}
        \label{fig:5qubita}
    \end{subfigure}

    \vspace{1em}

    \begin{subfigure}{\textwidth}
        \centering
        \includegraphics[height=1.9in]{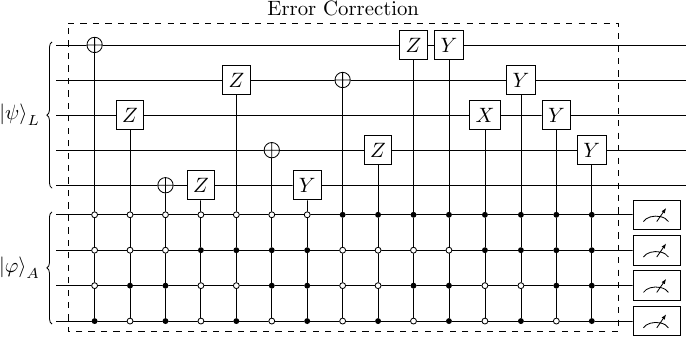}
        \caption{Circuit diagram (second half).}
        \label{fig:5qubitb}
    \end{subfigure}

    \caption{The 5-qubit error correcting code.}
    \label{fig:5qubit}

\end{figure}

\begin{table}
    \centering
    \begin{tabular}{cccccl}
        \toprule
        Noise & $\ket{\varphi}_A$ & Count & $P_{obs}$ & $P_{true}$ \\[0.1cm] 
                                                                                       \midrule \\[-0.4cm]
        % =======================================================================================================
        $\mathbb{I}$ & 0000 & 254 & $3.10 \cdot 10^{-2}$\hspace{0.7em} & $5.99 \cdot 10^{-1}$ \\[0.15cm]
        % =======================================================================================================
        $X_1$ & 0001 & 63 &
            \multirow{10}{*}{$\left.\begin{array}{l}
                7.69 \cdot 10^{-3} \\
                1.93 \cdot 10^{-2} \\
                9.77 \cdot 10^{-4} \\
                0.00 \\
                1.59 \cdot 10^{-3} \\
                2.44 \cdot 10^{-4} \\
                1.22 \cdot 10^{-3} \\
                1.22 \cdot 10^{-4} \\
                6.10 \cdot 10^{-4} \\
                2.44 \cdot 10^{-4}
            \end{array}\right\rbrace$} & \multirow{10}*{$3.15 \cdot 10^{-2}$} \\
        $X_2$ & 1000 & 158 & & & \\
        $X_3$ & 1100 & 8 & & & \\
        $X_4$ & 0110 & 0 & & & \\
        $X_5$ & 0011 & 13 & & & \\
        $Z_1$ & 1010 & 2 & & & \\
        $Z_2$ & 0101 & 10 & & & \\
        $Z_3$ & 0010 & 1 & & & \\
        $Z_4$ & 1001 & 5 & & & \\
        $Z_5$ & 0100 & 2 & & & \\[0.15cm]
        % =======================================================================================================
        $Y_1$ & 1011 & 1 &
            \multirow{5}{*}{$\left.\begin{array}{l}
                1.22 \cdot 10^{-4} \\
                1.10 \cdot 10^{-3} \\
                0.00 \\
                0.00 \\
                0.00
            \end{array}\right\rbrace$} & \multirow{5}*{$1.66 \cdot 10^{-3}$} \\
        $Y_2$ & 1101 & 9 & & & \\
        $Y_3$ & 1110 & 0 & & & \\
        $Y_4$ & 1111 & 0 & & & \\
        $Y_5$ & 0111 & 0 & & & \\[0.15cm]
        % =======================================================================================================
        \multicolumn{2}{r}{Uncorrectable:} &  600 & $7.32 \cdot 10^{-2}$\hspace{0.7em} & $7.78 \cdot 10^{-2}$ \\
        \multicolumn{2}{r}{Hardware:}      & 7066 & $8.63 \cdot 10^{-1}$\hspace{0.7em} & 0.00\hspace{2.8em}   \\ \bottomrule
    \end{tabular}
    
    \caption{
        Empirical probabilities of Pauli $X$, $Y$, and $Z$ errors on logical qubits in the 5-qubit code.
        Errors are introduced independently with probability $0.05$.
        Values were obtained by running the circuit in Figure~\ref{fig:5qubit} on IBM Torino for 8192 shots. The discrepancy between the observed and true values can be explained by noisy encoding, syndrome extraction, and measurement operations which have not been accounted for. In our construction, we assume perfect such operations with the knowledge that if the error rate is small enough, the threshold theorem guarantees that the logical rate will be lowered. We also simulated this circuit in Qiskit Aer to verify that it functions as expected.
    }
    \label{tab:5qubit}
\end{table}

Finally, let us calculate the threshold for this code. As usual, suppose the physical qubits independently experience a single qubit Pauli error with probability $p$. The probability that the code is successful is 
\begin{equation}
    p_{success}=(1-p)^5+5p(1-p)^4,
\end{equation}
and so the probability that the code fails is
\begin{equation}
    p_{fail}=1-p_{success}=1-(1-p)^5-5p(1-p)^4.
\end{equation}
Then the (pseudo) threshold of the code can be solved for numerically, producing $p\approx0.1311$. Thus, the 5-qubit code lowers the error rate when the error rate of the physical qubits is below $\approx13.11\%$ (again, this assumes depolarizing type noise, perfect measurements, etc.).
%%%%%%%%%%%%%%%%%%%%%%%%%%%%%%%%%%%%%%%%%%%%%%%%%%%%%%%%%%%%%%%%%%%%%%%%%%%%%%%%%%%%%%%%%%%%%%%%%%%%%%%%%%%%%%%%%%%%%%%

\subsection{The Parity Check Matrix Perspective}\label{sec:paritycheck}

It turns out that there is a way to convert the stabilizer code formalism to a linear algebra perspective over the field $\mathbb{Z}_2$ which is very convenient for computational purposes, and it is this perspective that we will now endeavor to understand. Suppose we are given a stabilizer subgroup $\mathcal{S}$ of the Pauli group $G_n$ that is generated by $\{S_1,\ldots,S_k\}$ with $k<n$. Every element of $G_n$ has the form $\lambda X_1^{p_1}\cdots X_n^{p_n} Z_1^{q_1}\cdots Z_n^{q_n}$, where $\lambda\in\{\pm1,\pm i\}$ and $p_i,q_j\in\mathbb{Z}_2$. Since every Pauli operator squares to the identity, there is a homomorphism $\varphi:G_k\to\mathbb{Z}_2^{2n}$ given by $\varphi(\lambda X_1^{p_1}\cdots X_n^{p_n} Z_1^{q_1}\cdots Z_n^{q_n})=(p_1,\ldots,p_n,q_1,\ldots,q_n)$. In what follows, we will define $\Vec{p}=(p_1,\ldots,p_n)$ and $\Vec{q}=(q_1,\ldots,q_n)$ for notational convenience. The map $\varphi$ is not an isomorphism since changing $\lambda$ does not affect the output and the mapping therefore cannot be injective. In fact, the kernel of $\varphi$ is the subgroup $\{\pm1,\pm i\}$. However, since the stabilizer subgroup $\mathcal{S}$ does not contain $-1$ (which implies it does not contain $\pm i$), the restriction of $\varphi$ to $\mathcal{S}$ has trivial kernel and is therefore injective. Since the map is clearly also a surjection, it follows that $\mathcal{S}$ is isomorphic to $\mathbb{Z}^{2n}$. This isomorphism allows us to translate many properties of stabilizer codes into a linear algebraic setting. Indeed, as we show below, independence of the generators of the stabilizer subgroup is equivalent to linear independence in $\mathbb{Z}_2^{2n}$ and we can equip $\mathbb{Z}_2^{2n}$ with a symplectic product that determines when elements of the Pauli group commute. Let us prove the first of these results. We will need to develop some mathematical machinery to attack the second. 
\begin{proposition}\label{prop:independence}
    Let $g_1,\ldots,g_m\in G_n$ and assume that the group generated by these elements does not contain $-1$ (we do not assume commutativity, so the generated group may not be a stabilizer subgroup). Then $g_1,\ldots,g_m$ are independent generators if and only if $\varphi(g_1),\ldots,\varphi(g_m)$ are linearly independent.
\end{proposition}
\begin{proof}
    Suppose $g_1,\ldots,g_m$ are dependent but $\varphi(g_1),\ldots,\varphi(g_m)$ are linearly independent. Then there is a $g_j$ such that
    \begin{equation}
        g_j = \prod_{i\ne j}g_i^{r_i}
    \end{equation}
    for some $r_i\in\mathbb{Z}_2$. Since $\varphi$ is a homomorphism, we then have
    \begin{equation}
        \varphi(g_j)=\sum_{i\ne j}r_i\varphi(g_i),
    \end{equation}
    but this contradicts the linear independence of the $\varphi(g_i)$. Thus, linear independence in $\mathbb{Z}_2^{2n}$ implies independence in $G_n$.

    Conversely, suppose $g_1,\ldots,g_m$ are independent but $\varphi(g_1),\ldots,\varphi(g_m)$ are linearly dependent. Then there are $r_i\in\mathbb{Z}_2$ not all zero such that
    \begin{equation}
        r_1\varphi(g_1)+\cdots+r_n\varphi(g_n)=0.
    \end{equation}
    It therefore follows that
    \begin{equation}
        \varphi(g_1^{r_1}\cdots g_n^{r_n})=0,
    \end{equation}
    but we know $\ker(\varphi)=\{\pm1,\pm i\}$, and so we have
    \begin{equation}
        g_1^{r_1}\cdots g_n^{r_n}=\pm i^p
    \end{equation}
    for some $p\in\mathbb{Z}_2$. Since we assume that not all $r_i$ are zero, this produces a dependence between the generators, contradicting the assumption that they are independent. Thus, independence in the generators implies linear independence in $\mathbb{Z}_2^{2n}$.
\end{proof}

To proceed, it is helpful to first discuss the theory of symplectic vector spaces. We introduce only the tools we need here, and the interested reader is referred to Roman's excellent text \cite{roman2007} for a more detailed treatment. Let $V$ be a vector space over a field $F$. A map $\omega:V\times V\to F$ is called \textbf{bilinear} if it is linear in each component; that is, $\omega(\lambda x+y,z)=\lambda\omega(x,z)+\omega(y,z)$ and $\omega(x, \lambda y+z)=\lambda\omega(x,y)+\omega(x,z)$ for all $x,y,z\in V$ and $\lambda\in F$. Such maps are sometimes called \textbf{bilinear forms}. They are additionally called \textbf{alternating} if $\omega(v,v)=0$ for all $v\in V$ and \textbf{non-degenerate} if $\omega(u,v)=0$ for all $v\in V$ implies $u=0$. A non-degenerate, alternating, bilinear form is called a \textbf{symplectic form}, and a vector space equipped with a symplectic form is called a \textbf{symplectic vector space}. As noted by Roman, the term ``symplectic'' comes from the Greek word for ``intertwined''. However, the original usage of the term in this context can be attributed to Weyl, who used it as a substitute for the word ``complex'' in his book \textit{The Classical Groups} \cite{weyl1946}.

\begin{lemma}\label{lemma:skew-sym}
    Let $\omega$ be a symplectic form. Then $\omega$ is \textbf{skew-symmetric}: $\omega(u,v)=-\omega(v,u)$ for all $u,v\in V$.
\end{lemma}
\begin{proof}
    Let $u,v\in V$. Then $\omega(u+v,u+v)=0$ since $\omega$ is alternating. But $\omega$ is also bilinear, so that we have
    \begin{equation}
        0=\omega(u,u)+\omega(v,v)+\omega(u,v)+\omega(v,u)=\omega(u,v)+\omega(v,u).
    \end{equation}
    Thus, $\omega(u,v)=-\omega(v,u)$.
\end{proof}

When $V$ is finite dimensional, there is an isomorphism $\Phi:V\to F^n$, where $n=\dim(V)$. This allows us to give any symplectic form $\omega$ a matrix representation $\Omega$ by writing
\begin{equation}
    \omega(u,v)=\Phi(u)^T\Omega\Phi(v).
\end{equation}
The coefficients of $\Omega$ are $\Omega_{ij}=\omega(v_i,v_j)$, where $\{v_1,\ldots,v_n\}$ is a basis for $V$, and so skew-symmetry of $\omega$ is equivalent to skew-symmetry of the matrix representation $\Omega$, justifying the terminology. The next proposition characterizes the form of $\Omega$ and tells us that symplectic vector spaces are even dimensional. Actually, an easier proof using the determinant is available when the characteristic of the base field is not 2, but we will be working in $\mathbb{Z}_2$ in what follows.

\begin{proposition}\label{prop:standard-form}
    Let $V$ be a symplectic vector space. Then $\dim(V)=2n$ for some $n\in\mathbb{N}$ and there is a basis $\{e_1,\ldots,e_{n},f_1,\ldots,f_n\}$ for $V$ such that $\omega(e_i,f_j) = \delta_{ij}$ and $\omega(e_i,e_j)=\omega(f_i,f_j)=0$. Equivalently,
    \begin{equation}\label{eq:standard-form}
        \Omega = \begin{pmatrix}
            0&\mathds{1}\\-\mathds{1}&0
        \end{pmatrix}.
    \end{equation}
\end{proposition}
\begin{proof}
    Let $e_1\ne0$ be any nonzero vector in $V$. Since $\omega$ is non-degenerate, there is a $v\in V$ such that $\omega(e_1,v)\ne0$. Define $f_1:=v/\omega(e_1,v)$ and $W:=\textnormal{span}\{e_1,f_1\}$. Since $\omega$ is alternating and bilinear, we have $\omega(e_1,\lambda e_1)=\lambda\omega(e_1,e_1)=0$ for all $\lambda\in F$, and so the fact that $\omega(e_1,v)\ne0$ guarantees that $e_1$ and $f_1$ are linearly independent and therefore form a basis for $W$. Then $V=W\oplus W^\bot$, where $W^\bot$ is the orthogonal complement of $W$ with respect to $\omega$ defined by
    \begin{equation}
        W^\bot:=\{v\in V:\omega(w,v)=0\textnormal{ for all }w\in W\}.
    \end{equation}
    Moreover, $\omega(e_1,f_1)=\omega(e_1,v)/\omega(e_1,v)=1$ and by skew-symmetry (Lemma~\ref{lemma:skew-sym}), $\omega(f_1,e_1)=-1$. Continuing inductively completes the proof.
\end{proof}

Let $A = \lambda_A X_1^{p_1}\cdots X_n^{p_n}Z_1^{q_1}\cdots Z_n^{q_n}$ and $B = \lambda_B X_1^{r_1}\cdots X_n^{r_n}Z_1^{s_1}\cdots Z_n^{s_n}$ be two elements of the Pauli group. In $\mathbb{Z}_2^{2n}$, these elements are represented by $(\Vec{p},\Vec{q})$ and $(\Vec{r},\Vec{s})$, respectively. In order to commute $B$ past $A$, we must count how many times an $X$ from $B$ gets applied to the same qubit as a $Z$ from $A$ and how many times a $Z$ from $B$ gets applied to the same qubits as an $X$ from $A$. The former is seen to be $\Vec{q}\cdot\Vec{r}$ and the latter is $\Vec{p}\cdot\Vec{s}$. Thus, we have that $A$ and $B$ commute if and only if $\Vec{p}\cdot\Vec{s}+\Vec{q}\cdot\Vec{r}=0$, and this quantity can be written as a symplectic form.

\begin{lemma}\label{lemma:pauli-commutes}
    Two elements $g,h$ of the Pauli group $G_n$ commute if and only if 
    \begin{equation}
        \varphi(g)\Omega \varphi(h)^T=0,
    \end{equation}
    where $\Omega$ is the matrix representation of a symplectic form $\omega$.
\end{lemma}
\begin{proof}
    Let $\varphi(g)=(\Vec{p},\Vec{q})$ and $\varphi(h)=(\Vec{r},\Vec{s})$. Let $\omega$ be the symplectic form defined in the standard basis by \eqref{eq:standard-form}. Then
    \begin{align}
        \varphi(g)\Omega\varphi(h)^T&=\begin{pmatrix}\Vec{p}&\Vec{q}
        \end{pmatrix}
        \begin{pmatrix}
            0&\mathds{1}\\-\mathds{1}&0
        \end{pmatrix}
        \begin{pmatrix}
            \Vec{r}\\\Vec{s}
        \end{pmatrix}\\
        &=\begin{pmatrix}\Vec{p}&\Vec{q}
        \end{pmatrix}
        \begin{pmatrix}
            \Vec{s}\\-\Vec{r}
        \end{pmatrix}\\
        &=\Vec{p}\cdot\Vec{s}-\Vec{q}\cdot\Vec{r}.
    \end{align}
    Since we are working in $\mathbb{Z}_2$, we have $-1\equiv+1$ and so 
    \begin{equation}
        \varphi(g)\Omega\varphi(h)^T=\Vec{p}\cdot\Vec{s}+\Vec{q}\cdot\Vec{r}.
    \end{equation}
    Thus, $g$ and $h$ commute if and only if $\varphi(g)\Omega\varphi(h)^T=0$.
\end{proof}

This lemma allows us to write the condition that $g$ is in the normalizer of the stabilizer subgroup in terms of linear algebra in $\mathbb{Z}_2^{2n}$. Indeed, let us define the \textbf{parity check matrix} to be the matrix of all images of the generators of the stabilizer under the isomorphism $\varphi$. Explicitly, we define
\begin{equation}
    C = \begin{pmatrix}
        \varphi(S_1)\\
        \varphi(S_2)\\
        \vdots\\
        \varphi(S_n)
    \end{pmatrix}.
\end{equation}
Letting $\varphi(S_i)=(\Vec{x}_i,\Vec{y}_i)$ with $\Vec{x}_i=(x_{i,1},\ldots,x_{i,n})$ and $\Vec{y}_i=(y_{i,1},\ldots,y_{i,n})$, we can write
\begin{equation}
    C = \begin{pmatrix}
        X\vert Y
    \end{pmatrix},
\end{equation}
where the coefficients of $X$ are $X_{ij}=x_{i,j}$ and the coefficients of $Y$ are $Y_{ij}=y_{i,j}$. By Proposition~\ref{prop:independence}, the parity check matrix has full rank. Recalling that an element of the Pauli group is in the normalizer of the stabilizer subgroup if and only if it commutes with all generators of the stabilizer, we have the following result.

\begin{lemma}\label{lemma:in-normalizer}
    Let $g\in G_n$ be an element of the Pauli group and let $\mathcal{S}$ be the stabilizer subgroup with generators $S_1,\ldots,S_k$. Then $g\in N(\mathcal{S})$ if and only if 
    \begin{equation}
        C\Omega\varphi(g)^T=0.
    \end{equation}
\end{lemma}
\begin{proof}
    Recall that $g\in N(\mathcal{S})$ if and only if $[g,S_i]=0$ for all $i=1,\ldots,k$. By Lemma~\ref{lemma:pauli-commutes}, this is equivalent to the condition
    \begin{equation}
        \varphi(S_i)\Omega\varphi(g)^T=0
    \end{equation}
    for all $i=1,\ldots,k$. But the rows of $C$ are $\varphi(S_i)$, and so this system of equations is equivalent to $C\Omega\varphi(g)^T=0$.
\end{proof}

As linear algebra is one of the most widely understood fields of mathematics, this parity check version of the stabilizer formalism is quite convenient, and the reader may wonder if the original description can be abandoned altogether. While there are advantages to each description, we can indeed show that the stabilizer subgroup is completely defined by a choice of parity check matrix.

\begin{theorem}
    Let $C$ be a $k\times 2n$ matrix with linearly independent rows $\Vec{r}_1,\ldots,\Vec{r}_k$ having vanishing symplectic product
    \begin{equation}
        \Vec{r}_i\Omega\Vec{r}_j^T=0
    \end{equation}
    for all $i,j=1,\ldots,k$. Then there exists independent $S_1,\ldots,S_k\in G_n$ such that $[S_i,S_j]=0$ and $-1\notin\langle S_1,\ldots,S_k\rangle$.
\end{theorem}
\begin{proof}
    We can always choose $S_i\in G_n$ so that $\varphi(S_i)=\Vec{r}_i$. Since we assume the symplectic product of the rows vanishes, Lemma~\ref{lemma:pauli-commutes} guarantees that $[S_i,S_j]=0$ for all $i,j$. The linear independence of the rows guarantees that the $S_i$ are independent by Proposition~\ref{prop:independence}. We must show that $-1\notin\langle S_1,\ldots,S_k\rangle$. 
    
    Suppose the contrary, that $-1\in\mathcal{S}:=\langle S_1,\ldots,S_k\rangle$, and define $\Pi:\mathbb{Z}_2^k\to\mathcal{S}$ by $\Pi(v_1,\ldots,v_k)=\prod_{i=1}^kS_i^{v_i}$. Then $\Pi(\Vec{u}+\Vec{v})=\prod_{i=1}^k S_i^{u_i+v_i}=\left(\prod_{i=1}^k S_i^{u_i}\right)\left(\prod_{i=1}^kS_i^{v_i}\right)=\Pi(\Vec{u})\Pi(\Vec{v})$. Thus, $\Pi$ is a (surjective) group homomorphism. Since $\mathcal{S}$ contains $-1$, there is a $\Vec{v}\in\mathbb{Z}_2^{k}$ such that $\Pi(\Vec{v})=-1$. But then the homomorphism $\varphi$ can be applied, producing
    \begin{equation}
        0=\varphi(-1)=\varphi(\Pi(\Vec{v}))=\sum_{i=1}^kv_i\varphi(S_i).
    \end{equation}
    By the assumption that the $\varphi(S_i)$ are linearly independent, we then have that $v_i=0$ for all $i=1,\ldots, k$. But then $\Pi(\Vec{v})=1$, which contradicts the assumption that $\Pi(\Vec{v})=-1$. Thus, $-1\notin\langle S_1,\ldots,S_k\rangle$.
\end{proof}

\subsection{Difficulties}\label{sec:difficulties}

It may appear that stabilizer codes solve all of our problems and that there is nothing more to do than implement them on real devices to perform error correction. However, consider an $[[n,k,d]]$ stabilizer code with $k=1$ and $n$ large. In this code, one logical qubit is encoded in $n$ physical qubits. There are $n-1$ stabilizers and thus $2^{n-1}$ possible syndromes. It is therefore impractical to expect that a table can be created which matches each syndrome with an operator that properly corrects the error corresponding to the syndrome. In fact, the problem of optimally decoding a stabilizer code is known to be \#P-complete \cite{iyer2015}. We can attempt to alleviate this problem by using a decoder which only approximates the error correcting operation associated to a syndrome or makes an informed guess as to which operation should be performed. This difficulty has turned finding good decoders for particular code designs into an active area of research. One promising approach is the training of neural networks to approximate the best decoding procedure for a given syndrome \cite{baireuther2018,krastanov2017,torlai2017,varsamopoulos2017}, and we will have more to say about this in Section~\ref{sec:decoderML}.

In addition to the issue of decoding, there is the problem of determining whether a given code actually lowers the error rate that we started with. Luckily, there has been a number of results showing that there is a threshold such that if the physical qubits making up a code experience an error less than this threshold, then the stabilizer code will indeed lower the logical error rate \cite{aharonov1997,gottesman1998,kitaev1997,knill1998,knill1998a,preskill1998}. The exact statement of each such result depends on the architecture of the code, but each such theorem is referred to as a \textbf{threshold theorem}. A related problem that we have not yet mentioned is our assumption that the error happens between the encoding and the syndrome extraction procedure. Of course, nature will not choose to be so kind to us on every evaluation of the error correction routine. It may happen that one or more of the gates in the encoder and/or syndrome extractor fails, producing an erroneous encoded logical qubit or an erroneous syndrome which then triggers the wrong decoding routine. With this in mind, let us take a step back to reconsider things. We have limited our discussion thus far to a subclass of easily accessible errors, and it is therefore worth identifying the types of errors which typically occur in a quantum system in order to ensure our codes properly protect the system from all feasible errors. Of course, the physical manifestation of these errors will vary between architectures; what is easy with a photonic device may be quite difficult with a trapped-ion device, and so on. However, there are errors of sufficiently general form which can be identified and characterized across all architectures.

The first such error is due to imperfections in the device. Suppose, for example, that we wish to implement the single qubit identity operation $\mathds{1}$. To do so, we create a mechanism $U_\mathds{1}$ that after many tests exhibits the behavior of an identity operation within some tolerance. Given that we have indeed managed to make a unitary operation, the true result of applying $U_\mathds{1}$ is a rotation in some direction. Indeed, unitary transformations preserve the norm of a state and therefore map the unit sphere to itself. Given an initial state $\ket{\psi}=\alpha\ket{0}+\beta\ket{1}$, the normalization condition implies $\lvert\alpha\rvert^2+\lvert\beta\rvert^2=1$, where we recall that $\alpha,\beta\in\mathbb{C}$. Thus, we have $(\real(\alpha))^2+(\imag(\alpha))^2+(\real(\beta))^2+(\imag(\beta))^2=1$, which defines a unit sphere embedded in $\mathbb{R}^4$. However, we recall that global phases do not affect physically observable phenomenon, and so we are free to multiply by a global phase to eliminate the imaginary part of $\alpha$. This leaves us with
$(\real(\alpha))^2+(\real(\beta))^2+(\imag(\beta))^2=1$, which defines a unit sphere embedded in $\mathbb{R}^3$ called the \textbf{Bloch sphere}. Thus, every pure quantum state is defined by a coordinate on the Bloch sphere, and so the $U_\mathds{1}$ operation is given by a rotation on the Bloch sphere.

In order to further discuss errors due to imperfect gates, we will need to understand rotations on the Bloch sphere. Since the Bloch sphere is a 2-dimensional surface, it can be parameterized with two coordinates using the usual spherical coordinate system with the radius set to one. Recall that the conversion from cartesian coordinates to spherical coordinates is accomplished by
\begin{align}
    x&=r\sin(\theta)\cos(\varphi)\\
    y&=r\sin(\theta)\sin(\varphi)\\
    z&=r\cos(\theta),
\end{align}
where $0\le\theta\le\pi$ and $0\le\varphi<2\pi$.
Identifying $\real(\alpha)$ with the $z$ coordinate, $\real(\beta)$ with the $x$ coordinate, and $\imag(\beta)$ with the $y$ coordinate produces
\begin{equation}
    \ket{\psi}=\alpha\ket{0}+\beta\ket{1}=\cos(\theta)\ket{0}+\sin(\theta)e^{i\varphi}\ket{1}.
\end{equation}
Observe that the eigenstates $\ket{0},\ket{1}$ of the Pauli-$Z$ operator are defined by the conditions $\theta=0$ and $\theta=\pi/2$, respectively, and that if we choose $\theta=\pi$, we obtain the state $-\ket{1}$, which is physically equivalent to $\ket{1}$. In order to avoid obtaining the same state twice, we reparameterize our coordinates by $\theta\mapsto\theta/2$, so that the resulting representation is
\begin{equation}
    \ket{\psi}=\alpha\ket{0}+\beta\ket{1}=\cos\left(\frac{\theta}{2}\right)\ket{0}+\sin\left(\frac{\theta}{2}\right)e^{i\varphi}\ket{1}.
\end{equation}
Even though the coordinates of $\ket{0}$ and $\ket{1}$ are still not unique (since $\varphi$ is free), we have managed to adjust the parameterization so that each unique point on the Bloch sphere produces a unique state. 

In this new parameterization, $\ket{0}$ is given by $\theta=0$ and $\ket{1}$ is given by $\theta=\pi$, so that these orthogonal states are antipodal to each other on the Bloch sphere. Similarly, the eigenstates $\ket{+}=\frac{1}{\sqrt{2}}(\ket{0}+\ket{1})$ and $\ket{-}=\frac{1}{\sqrt{2}}(\ket{0}-\ket{1})$ of the Pauli-$X$ operator are defined by the antipodal coordinates $\theta=\pi/2,\varphi=0$ and $\theta=\pi/2,\varphi=\pi$, respectively, and the eigenstates $\frac{1}{\sqrt{2}}(\ket{0}+i\ket{1})$ and $\frac{1}{\sqrt{2}}(\ket{0}-i\ket{1})$ of the Pauli-$Y$ operator are defined by the antipodal coordinates $\theta=\pi/2,\varphi=\pi/2$ and $\theta=\pi/2,\varphi=3\pi/2$, respectively. Thus, in the Bloch sphere, the eigenstates of a given Pauli operator are antiparallel and orthogonal to the eigenstates of the remaining Pauli operators. A rotation about the $Z$ axis is achieved by multiplying the coefficient of $\ket{1}$ by a relative phase of $e^{i\varphi}$, which can be accomplished with the matrix
\begin{equation}
    R_Z(\varphi)=e^{-i\varphi Z/2}=\begin{pmatrix}
        e^{-i\varphi/2}&0\\0&e^{i\varphi/2}
    \end{pmatrix},
\end{equation}
which leaves our state a global phase of $e^{i\varphi/2}$ away from the target state. Similarly, it can be verified that the rotations in the $X$ and $Y$ directions are accomplished by the exponentiation of the respective Pauli matrices, producing
\begin{equation}
    R_Y(\theta)=e^{-i\theta Y/2}=\begin{pmatrix}
        \cos(\theta/2)&-\sin(\theta/2)\\\sin(\theta/2)&\cos(\theta/2)
    \end{pmatrix}
\end{equation}
and
\begin{equation}
    R_X(\theta)=e^{-i\theta X/2}=\begin{pmatrix}
        \cos(\theta/2)&-i\sin(\theta/2)\\-i\sin(\theta/2)&\cos(\theta/2)
    \end{pmatrix}.
\end{equation}
There is a good reason that the rotations about these axes are given by the exponentiation of our Pauli matrices. Fix a unit vector $\hat{n}=(n_x,n_y,n_z)$ in the Bloch sphere and note that the rotations about this axis are parameterized by a single coordinate, call it $\theta$. All such rotations satisfy $R_{\hat{n}}(0)=\mathds{1}$ and the homomorphism property $R_{\hat{n}}(\theta_1+\theta_2)=R_{\hat{n}}(\theta_1)R_{\hat{n}}(\theta_2)$, thereby forming a one parameter unitary group. Moreover, these rotations satisfy $\lim_{\theta\to \theta_0}R_{\hat{n}}(\theta)\ket{\psi}=R_{\hat{n}}(\theta_0)\ket{\psi}$ for all $\ket{\psi}$ in our Hilbert space and are therefore strongly continuous. Then by Stone's theorem, there is a self-adjoint operator $H$ such that
\begin{equation}
    R_{\hat{n}}(\theta)=e^{i\theta H}.
\end{equation}
Since the matrices $\mathds{1},X,Y,Z$ form a basis for the single qubit self-adjoint operators, we have $H=a\mathds{1}+bX+cY+dZ$ for some $a,b,c,d\in\mathbb{C}$. We can set $a=0$ since the identity operation commutes with everything and the Baker-Campbell-Hausdorff formula therefore allows us to pull this term out, producing a global phase. We recover $b=-n_x/2$, $c=-n_y/2$, and $d=-n_z/2$ by comparing with the above cases, and it follows that the rotation about an arbitrary axis $\hat{n}=(n_x,n_y,n_z)$ is given by
\begin{equation}
    R_{\hat{n}}(\theta)=e^{-i\theta(n_xX+n_yY+n_zZ)/2}=\cos(\theta/2)\mathds{1}-i\sin(\theta/2)(n_xX+n_yY+n_zZ).
\end{equation}

Returning now to the problem of modeling imperfections in gates, we see that since $U_\mathds{1}$ is a rotation in the Bloch sphere, it has the form
\begin{equation}
    U_\mathds{1}=R_{\hat{n}}(\epsilon)
\end{equation}
for some unit vector $\hat{n}$ and some small angle $\epsilon$. Applying this gate $N$ times will produce the operator
\begin{equation}
    U_\mathds{1}^N=R_{\hat{n}}(N\epsilon)=\cos(\epsilon N/2)\mathds{1}-i\sin(\epsilon N/2)(n_xX+n_yY+n_zZ).
\end{equation}
As long as $\epsilon N$ is small, we can use the small angle approximation to write
\begin{equation}
    U_\mathds{1}^N\approx \left(1-\frac{\epsilon^2 N^2}{8}\right)\mathds{1}-i\frac{\epsilon N}{2}(n_xX+n_yY+n_zZ)
\end{equation}
so that the difference between the imperfect gate and the identity operation is
\begin{equation}
    \|\mathds{1}-U_\mathds{1}^N\|\approx \left\|\frac{\epsilon^2 N^2}{8}\mathds{1}+i\frac{\epsilon N}{2}(n_xX+n_yY+n_zZ)\right\|=\frac{\epsilon N}{2}\left\|\frac{\epsilon N}{4}\mathds{1}+i(n_xX+n_yY+n_zZ)\right\|,
\end{equation}
which is small whenever $\epsilon N$ is small. More importantly, given a state $\ket{\psi}=\alpha\ket{0}+\beta\ket{1}$, the probability that the state is $\ket{\psi}$ after $N$ applications of $U_\mathds{1}$ is
\begin{equation}
    P(\ket{\psi})=\lvert\bra{\psi}U_\mathds{1}^N\ket{\psi}\rvert^2\approx\left\lvert\bra{\psi}\left(\left(1-\frac{\epsilon^2 N^2}{8}\right)\mathds{1}-i\frac{\epsilon N}{2}(n_xX+n_yY+n_zZ)\right)\ket{\psi}\right\rvert^2,
\end{equation}
and so the gate fails with probability
\begin{equation}
    P_{fail}\approx1-\left\lvert\bra{\psi}\left(\left(1-\frac{\epsilon^2 N^2}{8}\right)\mathds{1}-i\frac{\epsilon N}{2}(n_xX+n_yY+n_zZ)\right)\ket{\psi}\right\rvert^2.
\end{equation}
Observe that $n_xX+n_yY+n_zZ$ is self-adjoint and therefore has real eigenvalues. It follows that $\bra{\psi}(n_xX+n_yY+n_zZ)\ket{\psi}$ is real, so that $P_{fail}$ becomes
\begin{equation}
    P_{fail}\approx1-\left(1-\frac{\epsilon^2N^2}{8}\right)^2-\frac{\epsilon^2N^2}{4}\lvert\bra{\psi}(n_xX+n_yY+n_zZ)\ket{\psi}\rvert^2=O(\epsilon^2N^2).
\end{equation}
Thus, for small enough $\epsilon N$, this source of error is negligible.

Another form of error that can arise is through an interaction with the environment, by which we mean an interaction in a larger Hilbert space which we had not taken into account. In other words, these errors are those that arise from having an imperfectly closed system. We tackled an example of this type in Section~\ref{sec:3qubitcode} to illustrate the discretization of errors. It is worth noting that the discretization of errors is a property that arises due to the nature of a stabilizer code. Let us work out a slightly more general example to illustrate this point. After the logical qubit interacts with the environment, it becomes a pure quantum state in the larger Hilbert space, but on the level of the logical Hilbert space, the state of the system becomes an ensemble of pure states $\mathcal{E}=\{(p_1,\ket{\psi_1}_L),\ldots,(p_k,\ket{\psi_k}_L)\}$ which represents a system which is in the state $\ket{\psi_i}_L$ with probability $p_i$. If we pass this system through the syndrome extraction procedure, we obtain the state
\begin{equation}\label{eq:stabilizer-difficulties}
    \frac{1}{2^\ell}\sum_{i_1,\ldots,i_\ell=0}^1\left(\prod_{j=1}^\ell(\mathds{1}+(-1)^{i_j}S_j)\right)\ket{\psi_k}_L\ket{i_1\cdots i_\ell}
\end{equation}
with probability $p_k$, where $\ell$ is the number of stabilizers $S_j$. If $\ket{\psi_k}_L$ is in the code space, then the term with $i_j=0$ for all $j$ is the only term that survives, and so a measurement will produce the all zero syndrome with certainty given that the system is in this state. Let us make the additional assumption that $\ket{\psi_k}_L=E_k\ket{\psi}_L$ for some Pauli operator $E_k$. Each such $E_k$ either commutes or anticommutes with each stabilizer, and so the only term in \eqref{eq:stabilizer-difficulties} which does not vanish is the one for which $i_j=0$ when $S_j$ and $E_k$ commute and $i_j=1$ when they anticommute. Upon measurement, the state therefore collapses to $E_k\ket{\psi}_L$ after syndrome extraction, allowing us to correct the error by applying an $E_k^\dagger$ operation. Thus, no matter which state in the ensemble is the true state of the system, the syndrome extraction procedure will force the state into one of a discrete set of error states. That is, the stabilizer code forces the system to choose between the states in the ensemble so that an error can be identified by the syndrome extraction procedure and potentially corrected. Thus, we are in principle able to correct errors due to an interaction with the environment of this form.

\section{A Brief Look at Topological Codes}\label{sec:topological-codes}
The theory of arbitrary stabilizer codes does not take into account the geometry of a particular quantum computing architecture. The individual qubits in such a device can be arranged in a number of ways with potentially limited connectivity between qubits. Therefore, to design the best possible code for a particular device, we should study codes which take this geometry into account. The so-called topological codes accomplish this by intersecting the theory of stabilizer codes with the homology theory of topological spaces as it typically appears in algebraic topology texts \cite{hatcher2002}. We will therefore begin our study of topological codes by reviewing basic homology theory in Section~\ref{sec:homology}. We then study the topological code formalism in Section~\ref{sec:top-code-form} before working out the canonical examples of the toric and planar codes in Sections~\ref{sec:toric-codes} and \ref{sec:planar-codes}. Much of the content of this section follows an approach similar to that of Fujii \cite{fujii2015}.

\subsection{Basic Homology Theory}\label{sec:homology}

A \textbf{topology} on a set $X$ is a collection $\tau$ of subsets of $X$ satisfying
\begin{enumerate}
    \item $\emptyset\in\tau$ and $X\in\tau$.
    \item If $\{U_\alpha\}_{\alpha}\subset\tau$, then $\cup_{\alpha}U_\alpha\in\tau$.
    \item If $\{U_1,\ldots,U_n\}\subset\tau$, then $\cap_{i=1}^nU_i\in\tau$.
\end{enumerate}
In words, a topology on $X$ is a collection of subsets of $X$, called \textbf{open sets}, that contains the empty set and the total set $X$, is closed under arbitrary unions of its elements, and is closed under finite intersections of its elements.
\begin{example}
    Let $X=\{1,2,3\}$ and define $\tau=\{\emptyset,\{1\},\{2,3\},X\}$. Then $\tau$ is a topology on $X$ since $\{1\}\cup\{2,3\}=X\in\tau$ and $\{1\}\cap\{2,3\}=\emptyset\in\tau$. On the other hand, if we instead define $\tau = \{\emptyset,\{1\},\{2,3\},\{3\},X\}$, then $\tau$ is not a topology on $X$ because $\{1\}\cup\{3\}=\{1,3\}$ is not an element of $\tau$.
\end{example}
A \textbf{topological space} is a set $X$ equipped with a topology $\tau$. Whenever the topology is clear from the context, we refer to the topological space by its underlying set $X$. Topological spaces are important for a variety of reasons, not the least of which is the ability to define continuity in a precise way. Indeed, we call a map $f:X\to Y$ between topological spaces \textbf{continuous} if the preimage of every open set in $Y$ is again open in $X$. While the notion of continuity is on its own an important topic in mathematical analysis, it also serves as the clear structure preserving map between topological spaces. 

\begin{definition}
    A \textbf{homeomorphism} is a continuous bijection $f:X\to Y$ whose inverse is also continuous. We call $X$ and $Y$ homeomorphic and write $X\cong Y$ when such a map exists.
\end{definition}

The study of topology is often termed ``rubber sheet geometry'' because it is usually the case that our topological spaces can be viewed as some geometric structure made of a material that can be stretched and compressed but not torn. This perspective has its limitations but is more often a convenient picture to keep in our minds as we attempt to distinguish between topological spaces. As is so often the case for an interesting mathematical structure, there aren't any great necessary and sufficient conditions to guarantee that two spaces are homeomorphic, but there are many necessary conditions that can be defined. Such conditions are often called \textbf{topological invariants}. Here we will discuss one such topological invariant known as \textbf{homology}. Actually, there are many types of homology, but we will endeavor to describe only one here. The reader is referred to \cite{hatcher2002} for further details.

Consider a graph $X=(\mathcal{V},\mathcal{E})$ consisting of a set of vertices $\mathcal{V}$ and a set of edges $\mathcal{E}$. Let us denote the vertices by $v_i$ for $i=1,\ldots,\lvert\mathcal{V}\rvert$. The edges can be regarded as pairs of vertices which we will denote $e_{ij}=(v_i,v_j)$. When the ordering of the vertices of an edge matters, we call the graph \textbf{directed} or a \textbf{digraph}. As an example, consider the digraph in Figure~\ref{fig:homology1}. It has vertices $v_1,\ldots, v_4$ and edges $e_{12}, e_{23}, e_{34}, e_{41}, e_{24}$. Let us include also the faces $f_{124}$ and $f_{234}$ enclosed by the edges $e_{41}, e_{12}, e_{24}$ and $e_{23}, e_{34}, -e_{42}$, respectively. Here we have included a factor of $-1$ in front of $e_{42}$ to denote that we are traversing the edge in the reverse direction. With the additional structure of the faces, we have defined an object called a \textbf{cell complex}. We think of each vertex as a 0-dimensional open ball also called a \textbf{0-cell}, each edge as a 1-dimensional ball called a \textbf{1-cell}, each face as a 2-dimensional ball called a \textbf{2-cell}, and so on. If we now identify the vertical 1-cells $e_{41}$ and $-e_{23}$ (Pacman style), we see that the cell complex becomes a cylinder. The two horizontal 1-cells form the end circles of the cylinder and the diagonal 1-cell spirals down the surface. If we further identify $e_{12}$ with $-e_{43}$, the two end circles of the cylinder coincide, so that an ant traveling towards the top of the cylinder reappears at the bottom once it reaches the end. The cylinder has therefore become a torus (a hollow donut shape), which we denote by the symbol $\mathbb{T}^2$. Another way to see this is to connect the circles (making sure that the identified 1-cells line up with the correct orientation). With these identifications, the cell complex in Figure~\ref{fig:homology1} now represents the two-dimensional torus, and so one might correctly guess that a study of the cell complex will yield information about the topology of the torus. To this end, let us give a more algebraic description of what is happening in this cell complex.

\begin{figure}
    \centering
    \includegraphics[width=2.25in]{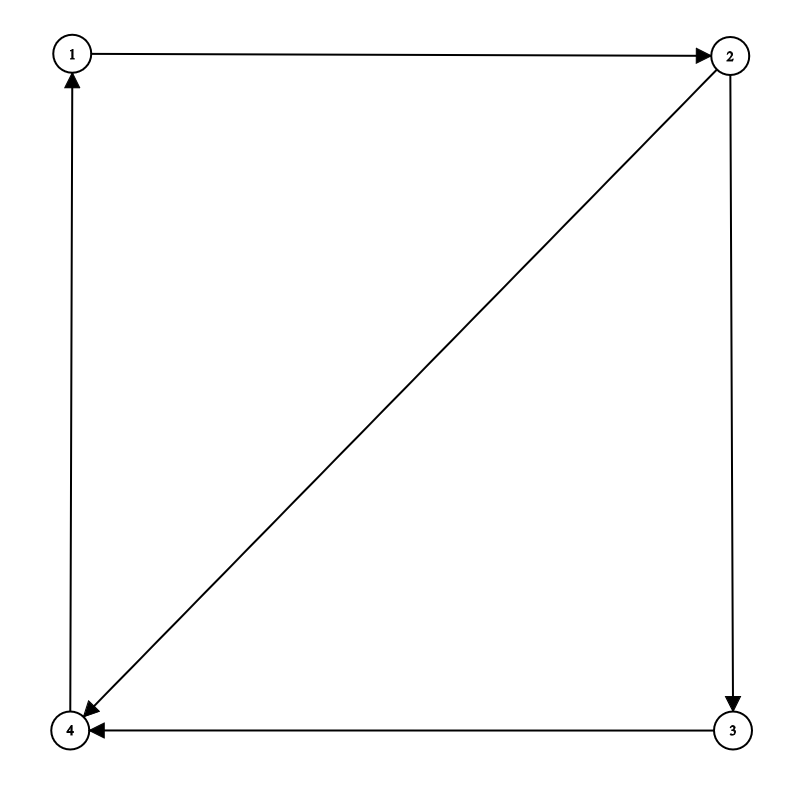}
    \caption{
        Example of a digraph with vertices $v_1, v_2, v_3, v_4$ and edges $e_{12}, e_{23}, e_{34}, e_{41},e_{24}$.
    }
    \label{fig:homology1}
\end{figure}

Linear combinations $\sum_ic_iv_i$ of the 0-cells with $c_i\in\mathbb{Z}$ are called \textbf{0-chains}, and the collection of all 0-chains is an abelian group under componentwise addition. We will let $C_0(X,\mathbb{Z})$ denote this group. Whenever the space $X$ and set of coefficients $\mathbb{Z}$ are clear from the context, we omit them and write $C_0$. Similarly, the collection of all linear combinations of 1-cells over $\mathbb{Z}$ forms the abelian group $C_1$ of \textbf{1-chains} and the linear combinations of 2-cells over $\mathbb{Z}$ form the abelian group $C_2$ of all \textbf{2-chains}. It is hopefully clear from the definition of each group that they are isomorphic to a direct product of several copies of $\mathbb{Z}$ (to be precise, $C_0\cong\mathbb{Z}^{\lvert\mathcal{V}\rvert}$, $C_1\cong\mathbb{Z}^{\lvert\mathcal{E}\rvert}$, and so on). Let us now define a map $\partial_1:C_1\to C_0$ by $\partial_1(e_{ij})=v_j-v_i$ and extending linearly. Similarly, define a map $\partial_2:C_2\to C_1$ by sending a 2-cell to the sum of its boundary 1-cells traversed in the clockwise direction (traversing a 1-cell pointed in the opposite direction produces a minus sign). For example, $\partial_2(f_{234})=e_{23}+e_{34}-e_{24}$. These maps, which can naturally be extended to higher dimensions, are called $\textbf{boundary}$ maps for obvious reasons and by convention we will take $\partial_0=0$. Since there are no 3-cells in our cell complex, we have that $C_3$ is trivial and $\partial_3=0$. By composing these maps, we create the sequence
\begin{equation}
    0\stackrel{0}{\to} C_2\stackrel{\partial_2}{\to} C_1\stackrel{\partial_1}{\to} C_0\stackrel{0}{\to}0.
\end{equation}

\begin{proposition}\label{prop:boundary-comp}
    The composition of boundary maps is trivial: $\partial_1\circ\partial_2=0$.
\end{proposition}
\begin{proof}
    Let $f$ be the 2-cell defined by the sequence of 1-cells $e_1,\ldots,e_n$, and let 1-cell $e_i$ be defined by the pair of 0-cells $v_{i,1},v_{i,2}$. Then $\partial_2(f)=e_1+\cdots+e_n$ and so $\partial_1(\partial_2(f))=\sum_{i=1}^n(v_{i,2}-v_{i,1})$. But $v_{i,2}=v_{i+1,1}$ for $i=1,\ldots,n-1$ and $v_{n,2}=v_{1,1}$. Thus, the sum telescopes and we have $\partial_1\partial_2(f)=0$.
\end{proof}

Proposition~\ref{prop:boundary-comp} tells us that the image of $\partial_2$ is contained in the kernel of $\partial_1$. In fact, the image of $\partial_2$ is trivially a normal subgroup of $\ker(\partial_1)$. Indeed, if $k\in\ker(\partial_1)$ and $m\in\imag(\partial_2)$, then $k+m-k=m\in\imag(\partial_2)$. This allows us to take the quotient $H_1(X)=\ker(\partial_1)/\imag(\partial_2)$, and it is this quotient group that is called the \textbf{first homology group} of the topological space $X$ defined by the cell complex. The result of this proposition extends to the boundary maps between any two consecutive dimensions. For the cell complex (without the torus 1-cell identifications), we have $\partial_2(f_{124})=e_{12}+e_{24}+e_{41}$ and $\partial_2(f_{234})=e_{23}+e_{34}-e_{24}$, and these 1-chains are clearly independent. Now computing the kernel of $\partial_1$, we must solve
\begin{equation}
    \partial_1(c_{12}e_{12}+c_{23}e_{23}+c_{34}e_{34}+c_{41}e_{41}+c_{24}e_{24})=0,
\end{equation}
which is equivalent to
\begin{equation}
    (c_{41}-c_{12})v_1+(c_{12}-c_{23}-c_{24})v_2+(c_{23}-c_{34})v_3+(c_{24}+c_{34}-c_{41})v_4=0.
\end{equation}
Solving this linear system of equations shows that $\ker(\partial_1)$ is generated by the 1-chains $e_{12}+e_{23}+e_{34}+e_{41}$ and $e_{24}-e_{23}-e_{34}$. Thus, in the quotient $H_1(X)=\ker(\partial_1)/\imag(\partial_2)$, we make the identifications $e_{24}=e_{23}+e_{34}$ and $e_{41}=-e_{12}-e_{23}-e_{34}$, which render the generators of the image trivial. It then follows that $H_1(X)=\{0\}$ is the trivial group. For the sake of clarity, it is worth being overly pedantic here. The group $\imag(\partial_2)$ has two generators over $\mathbb{Z}$ and is therefore isomorphic to $\mathbb{Z}\times\mathbb{Z}$. Similarly, the group $\ker(\partial_1)$ has two generators and is therefore isomorphic to $\mathbb{Z}\times\mathbb{Z}$. The elements of the quotient group $\ker(\partial_1)/\imag(\partial_2)$ are the cosets $e_{12}+e_{24}+e_{41}+\imag(\partial_2)=\imag(\partial_2)$ and $e_{23}+e_{34}-e_{24}+\imag(\partial_2)=\imag(\partial_2)$. Thus, the quotient group contains only the identity element and the first homology group is therefore trivial.

Let us now compute the zeroth homology group for the cell complex. Recalling that $\partial_0=0$, we must have that $\ker(\partial_0)=C_0$, which is generated by all four 0-cells. On the other hand, we have $\partial_1(e_{12})=v_2-v_1$, $\partial_1(e_{23})=v_3-v_2$, $\partial_1(e_{34})=v_4-v_3$, $\partial_1(e_{41})=v_1-v_4$, and $\partial_1(e_{24})=v_4-v_2$. Since $v_1-v_4$ is the sum of three other 1-cell images, it is not independent, and we can discard it. Similarly, we can discard $v_4-v_2$, and this leaves us with the three independent generators $v_2-v_1$, $v_3-v_2$, and $v_4-v_3$. Thus, $\imag(\partial_1)\cong\mathbb{Z}\times\mathbb{Z}\times\mathbb{Z}$. Taking the quotient amounts to making the identification $v_1=v_2=v_3=v_4$, which leaves a single generator. Thus, the zeroth homology group of the cell complex $X$ is $H_0(X)=\mathbb{Z}$.

Suppose now that we make the torus 1-cell identifications $e_{12}\equiv -e_{34}$ and $e_{41}\equiv -e_{23}$. Then all the 0-cells are also identified as a result: $v_1\equiv v_2\equiv v_3\equiv v_4$. Let us now compute the first homology group of the torus. The image of $\partial_2$ is generated by $e_{12}+e_{24}+e_{41}$ (the image of the other 2-cell is just the negation of this generator due to the torus 1-cell identifications). Meanwhile, all 0-cells have been identified, and so the boundary of any 1-cell is zero (for example, $\partial_1(e_{12})=v_2-v_1=v_1-v_1=0$). Thus, the kernel of $\partial_1$ is generated by $e_{12}, e_{24}, e_{41}$. Now taking the quotient produces $H_1(\mathbb{T}^2)=\mathbb{Z}\times\mathbb{Z}$. The zeroth homology group is much easier to compute. Since the boundary of each 1-cell vanishes, the image of $\partial_1$ is zero, and the kernel of the zero map is $C_0$, which is generated by the singular 0-cell $v_1$. Thus, $H_0(\mathbb{T}^2)\cong\mathbb{Z}$.

Comparing our results for $X$ and $\mathbb{T}^2$ tells us that these topological spaces are not homeomorphic because they have different homology groups and homology groups are preserved under a homeomorphism. Let us now see what happens when we remove one of the 2-cells of the cell complex $X$ (no edge identifications). In particular, let us remove $f_{234}$. The image of $\partial_2$ is now generated by $\partial_2(f_{124})=e_{12}+e_{24}+e_{41}$ alone, but the kernel of $\partial_1$ is still generated by $e_{12}+e_{23}+e_{34}+e_{41}$ and $e_{24}-e_{23}-e_{34}$. Thus, taking the quotient, the first homology group becomes $H_1(X\setminus\{f_{234}\})=\mathbb{Z}$. Meanwhile, the zeroth homology group, which is independent of the 2-chains (sums of 2-cells), remains unchanged. This change in the first homology group is a reflection of the fact that a hole has been created by the absence of the 2-cell $f_{234}$. Indeed, the first homology group is strongly linked to the number of independent holes with a one dimensional boundary. Similarly, the second homology group is linked to the number of independent holes with a two dimensional boundary (a sphere).

The group $C_k$ of $k$-chains actually has more than a group structure. It is in fact a generalization of a vector space known as a \textbf{module}, which drops the assumption that the underlying scalars form a field and instead allows for the more general algebraic structure known as a \textbf{ring}. In the above discussion, $C_k$ was a module over the ring of integers $\mathbb{Z}$, but this is not the only choice one can make. Take, for example, the field $\mathbb{Z}_2$ consisting of the integers under addition modulo 2. The coefficients of the generators of the $k$-chains are then 0 or 1 (or more precisely, the equivalence classes of 0 and 1). Consider the cell complex $X$ again and let us compute its zeroth homology group with respect to $\mathbb{Z}_2$. A set of generators for the image of $\partial_1$ is $v_1+v_2, v_2+v_3, v_3+v_4$ and the kernel of $\partial_0$ is $C_0$, which is generated by $v_1,v_2,v_3,v_4$. Now taking the quotient amounts to identifying $v_1=v_2=v_3=v_4$ so that $H_0(X,\mathbb{Z}_2)\cong\mathbb{Z}_2$. Let us also compute the first homology with respect to $\mathbb{Z}_2$. The image of $\partial_2$ is now generated by $e_{12}+e_{24}+e_{41}$ and $e_{23}+e_{34}+e_{24}$, while the kernel of $\partial_1$ is generated by $e_{12}+e_{23}+e_{34}+e_{41}$ and $e_{12}+e_{41}+e_{24}$. Now passing to the quotient amounts to identifying $e_{12}=e_{24}+e_{41}$ and $e_{23}=e_{34}+e_{24}$, which makes both generators of the kernel trivial. Thus, $H_1(X,\mathbb{Z}_2)=\{0\}$ is the trivial group. From now on, when it is unclear, we will specify the underlying ring $R$ by writing $C_k(X,R)$ for the module of $k$-chains in $X$ with coefficients in $R$.

\subsection{The Topological Code Formalism}\label{sec:top-code-form}

Topological codes form a beautiful intersection between the homology theory developed in the late 19th and early 20th centuries and the much more recent theory of quantum error correcting codes. Our goal will be to encode the information defining our code using homology theory. Here we will work with codes defined by a cell decomposition with $n$-cells no greater than $n=2$, and so we are really working with so-called \textbf{surface codes}, but much of this formalism is extendible to higher $n$. Note that we will drop the subscript on the boundary operator $\partial$ because the degree of the chain it acts on will be clear from the context. We begin by drawing a cell complex $\mathcal{X}\footnote{We are using $\mathcal{X}$ in place of $X$ so that the cell complex is not confused with the Pauli-$X$ operator.}=(\mathcal{V},\mathcal{E},\mathcal{F})$ (here, $\mathcal{F}$ is the collection of all 2-cells) and identifying each 1-cell with a physical qubit. We will identify the 2-cells $f$ of the cell complex with the operator $Z_f:=\prod_{e\in\partial_2(f)}Z_e$, where $Z_{e}$ is the Pauli-$Z$ operator acting on the qubit identified with 1-cell $e$ and the product is over all 1-cells in the boundary of $f$. Similarly, we identify all 0-cells $v$ with the operator $X_v:=\prod_{e:v\in\partial_1(e)}X_e$, where the product is over all 1-cells with boundary containing $v$. We then take the stabilizer of the code to be
\begin{equation}
    \mathcal{S}:=\langle\{X_v:v\in\mathcal{V}\}\cup\{Z_f:f\in\mathcal{F}\}\rangle,
\end{equation}
that is, we define the stabilizer to be the abelian subgroup of the Pauli group generated by the set of all $X_v,Z_f$. As we learned in Section~\ref{sec:stabilizer}, the stabilizer completely defines the $\mathcal{S}$-symmetric subspace, which we take to be the code space.

With these identifications, the module $C_0(\mathcal{X},\mathbb{Z}_2)$ (of all 0-chains) is generated by all Pauli-$X$ words in the stabilizer, $C_1(\mathcal{X},\mathbb{Z}_2)$ (all 1-chains) is generated by the qubits associated to the 1-cells, and $C_2(\mathcal{X},\mathbb{Z}_2)$ (all 2-chains) is generated by all Pauli-$Z$ words in the stabilizer. Explicitly, there is an injective group homomorphism $X:C_0(\mathcal{X},\mathbb{Z}_2)\to\mathcal{S}$ defined by $X(v)=X_v$ and extending linearly, as well as an injective group homomorphism $Z:C_2(\mathcal{X},\mathbb{Z}_2)\to\mathcal{S}$ defined by $Z(f)=Z_f$. On the level of the module of 1-chains, there is a bijection between $\varphi:\mathbb{Z}_2^{\lvert\mathcal{E}\rvert}\to C_1(\mathcal{X},\mathbb{Z}_2)$ defined by $\varphi(n_1,\ldots,n_{\lvert\mathcal{E}\rvert})=n_1e_1+\cdots+n_{\lvert\mathcal{E}\rvert}e_{\lvert\mathcal{E}\rvert}$, where we have labeled the 1-cells $e_i$, $i=1,\ldots,\lvert\mathcal{E}\rvert$. The practical effect of all of these identifications is that a Pauli-$X$ word can be thought of as a 0-chain, a Pauli-$Z$ word can be thought of as a 2-chain, and every collection of qubits can be thought of as a 1-chain.

\begin{figure}
    \centering
    \includegraphics[width=3in]{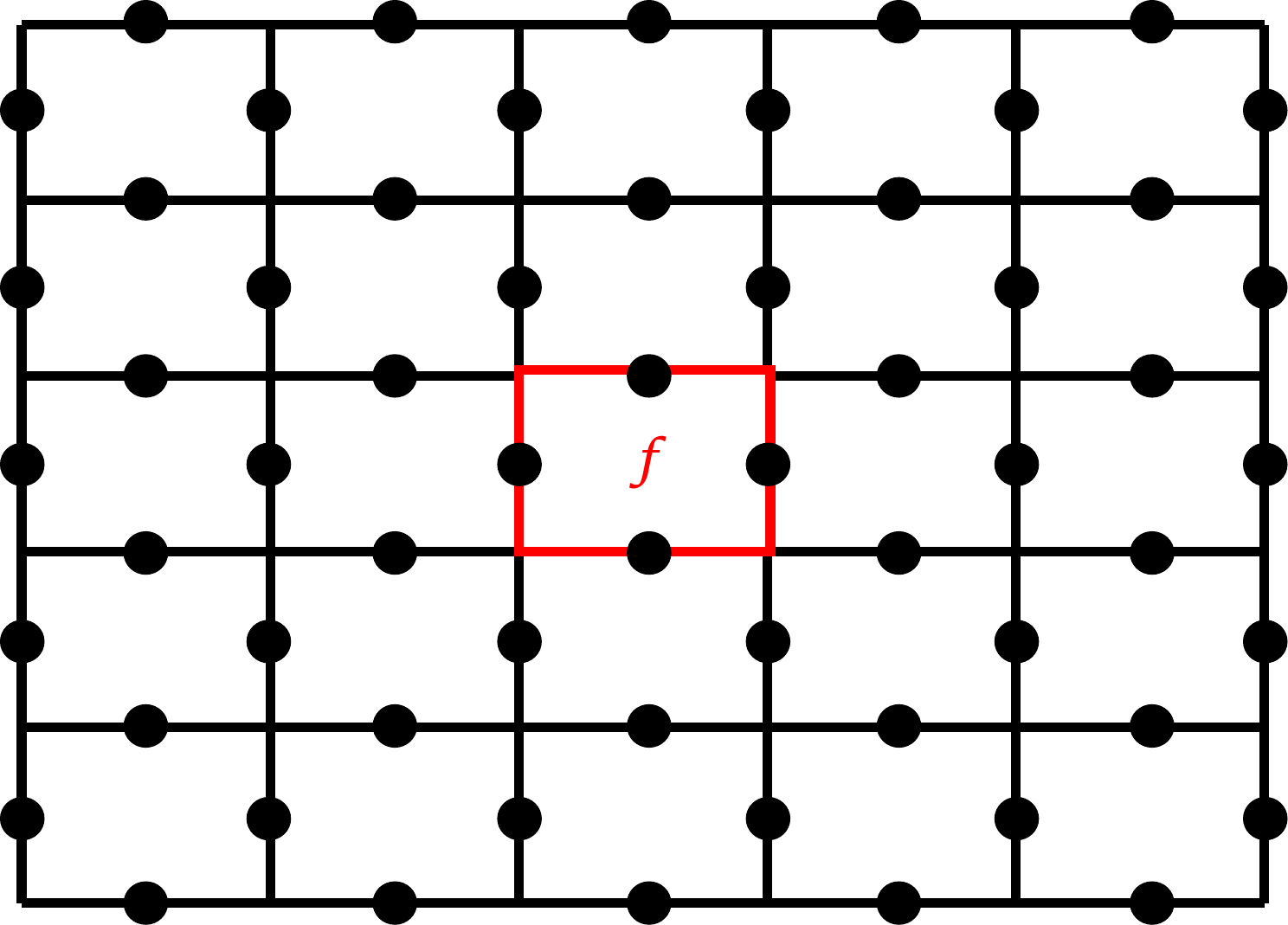}
    \caption{
        Lattice diagram for topological code. Black circles represent qubits and the red square corresponds to the operator $Z_f$.
    }
    \label{fig:grid-face}
\end{figure}
\begin{figure}
    \centering
    \includegraphics[width=3in]{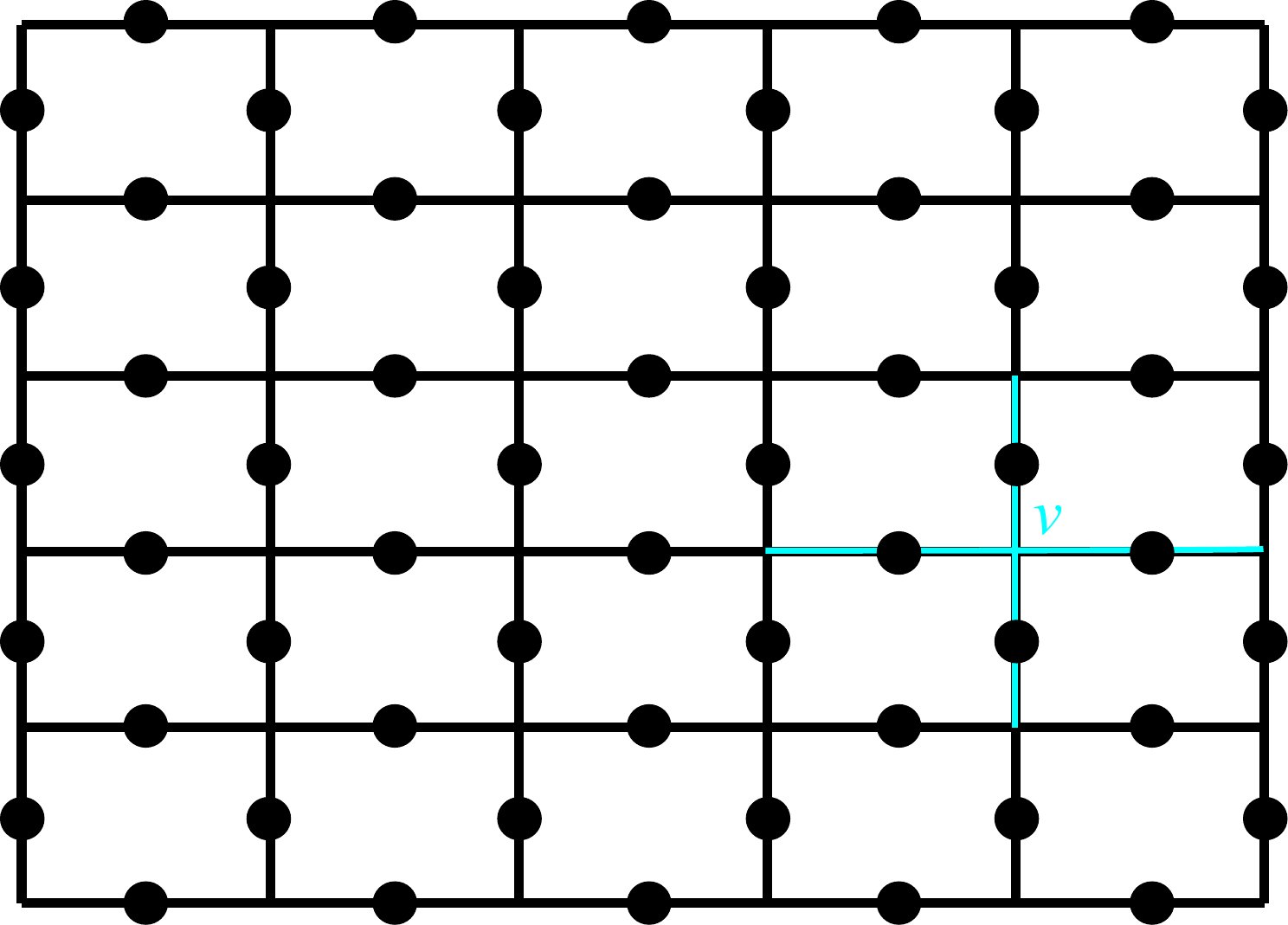}
    \caption{
        Lattice diagram for topological code. Black circles represent qubits and the blue cross corresponds to the operator $X_v$.
    }
    \label{fig:grid-x}
\end{figure}

Let us restrict our attention from a general cell complex to a lattice $\mathcal{X}=L$ as in Figure~\ref{fig:grid-face}. Each of the black dots lying on a 1-cell represents a qubit. The 2-cell enclosed by the red square corresponds to the Pauli-$Z$ word $Z_f$ associated to $f$. Likewise, in Figure~\ref{fig:grid-x}, the blue cross represents the Pauli-$X$ word $X_v$ associated to $v$. The nice feature of a lattice is that it has a dual structure which is easy to visualize. Indeed, for every lattice $L$, there is a dual lattice $L^*$ given by drawing a 1-cell between centers of adjacent 2-cells so that the 0-cells in $L$ become 2-cells in $L^*$ and 2-cells in $L$ become 0-cells in $L^*$. More precisely, there are isomorphisms $C_0(L,\mathbb{Z}_2)\cong C_2(L^*,\mathbb{Z}_2)$, $C_1(L,\mathbb{Z}_2)\cong C_1(L^*,\mathbb{Z}_2)$, and $C_2(L,\mathbb{Z}_2)\cong C_0(L^*,\mathbb{Z}_2)$, and this allows us to define a \textbf{co-boundary} $\partial^*_1:C_0(L,\mathbb{Z}_2)\to C_1(L,\mathbb{Z}_2)$ on $L$ by
\begin{equation}
    C_0(L,\mathbb{Z}_2)\stackrel{\cong}{\to}C_2(L^*,\mathbb{Z}_2)\stackrel{\partial_2}{\to}C_1(L^*,\mathbb{Z}_2)\stackrel{\cong}{\to}C_1(L,\mathbb{Z}_2).
\end{equation}
That is, the co-boundary applied to a 0-cell in $L$ produces a 1-cell in $L$ by first mapping the 0-cell in $L$ to the corresponding 2-cell in $L^*$, then applying the boundary map $\partial_2$ in $L^*$, and finally mapping the resulting 1-cell in $L^*$ to the corresponding 1-cell in $L$. Observe that $\partial^*(v)$ is the 1-chain consisting of all edges touching $v$. This allows us to write
\begin{equation}
    X_v = \prod_{e\in\partial^*(v)}X_e,
\end{equation}
where the product is now over all 1-cells in the co-boundary of $v$. Whereas in the original definitions of $Z_f$ and $X_v$, the 0 and 2-cells appeared to be treated differently, this co-boundary version of the $X$ homomorphism reveals the nice dual structure at play. In fact, we can further emphasize this duality by redefining the homomorphisms $X$ and $Z$ by $Z:C_1(L,\mathbb{Z}_2)\to G_n$ and $X:C_1(L^*,\mathbb{Z}_2)\to G_n$ so that $Z$ takes a 1-chain in the lattice $L$ to the corresponding element of the Pauli group and $X$ takes a 1-chain in the dual lattice $L^*$ to the Pauli group. The stabilizer is then generated by all elements of the form $X_{\partial^*(v)}$ and $Z_{\partial(f)}$.

Let $c\in C_1(L,\mathbb{Z}_2)$ and let $d\in C_1(L^*,\mathbb{Z}_2)$. Say $c=\sum_ic_ie_i$ for some $c_i\in\mathbb{Z}_2$ and some 1-cells $e_i\in L$ and $d=\sum_id_i\Tilde{e}_i$ for some $d_i\in\mathbb{Z}_2$ and some 1-cells $\Tilde{e}_i\in L^*$. The product $c\cdot d=\sum_ic_id_i$ is called the \textbf{intersection number} between the 1-chains, and takes the value 1 if $c$ and $d$ intersect at an odd number of qubits and the value 0 otherwise (because we are working in $\mathbb{Z}_2$). Recall that Pauli-$X$ and Pauli-$Z$ operations on the same qubit anticommute, while the same operations on different qubits obviously commute. Thus, in total, $Z_c$ commutes with $X_d$ if and only if $c$ and $d$ intersect at an even number qubits. Equivalently, $Z_c$ and $X_d$ commute if and only if their intersection number is 0. Thus, we have the following useful lemma.

\begin{lemma}
    Let $c$ be a 1-chain in $L$(not necessarily a boundary or co-boundary), $d$ a 1-chain in $L^*$, $v$ a 0-cell, and $f$ a 2-cell. Then $[Z_c,X_{\partial(v^*)}]=0$ if and only if $\partial(v^*)\cdot c=0$, where $v^*$ denotes the 2-cell in $L^*$ dual to the 0-cell $v$ in $L$. Similarly, $[Z_{\partial(f)},X_d]=0$ if and only if $\partial(f)\cdot d=0$.
\end{lemma}

It is worth stepping back for a moment to summarize what we have learned thus far and provide some further context. Our goal is to encode information into the topological degrees of freedom of a lattice. We have designed our lattice so that every 1-cell is associated to a qubit and the homomorphism $Z:C_1(L,\mathbb{Z}_2)\to G_n$ associates to every 1-chain an operation that acts on every qubit appearing in the 1-chain with a Pauli-$Z$ gate. Meanwhile, the dual lattice comes equipped with a homomorphism $X:C_1(L^*,\mathbb{Z}_2)\to G_n$ which associates to every 1-chain in $L^*$ an operation that acts on every qubit appearing in the 1-chain with a Pauli-$X$ gate. We define the stabilizer to be the abelian subgroup of the Pauli group generated by all Pauli operations of the form $Z_{\partial(f)}$ and $X_{\partial(v^*)}$, where $f$ is a 2-cell and $v$ is a 0-cell. These operations therefore correspond to the squares in $L$ and $L^*$, respectively. Notice that $Z_{\partial(f)}$ and $X_{\partial(v^*)}$ always intersect at an even number of points because the squares in the lattice intersect with a square in the dual lattice at either zero or two points. This confirms that the generators of the stabilizer do indeed commute. So, what does all of this have to do with error correction? The stabilizer group defines a stabilizer code completely. The detectable errors are those that anticommute with elements of the stabilizer group, and the logical operations are those that commute with the stabilizer. This means that every 1-chain (i.e. every collection of edges on the lattice or dual lattice) that has an odd intersection number with a generator of the stabilizer (i.e. a square in the dual lattice or lattice) corresponds to a detectable error. Meanwhile, every 1-chain with an even intersection number with all generators of the stabilizer is a logical operation that can be applied to the logical qubits encoded in the lattice.

\begin{figure}
    \centering
    \includegraphics[width=3in]{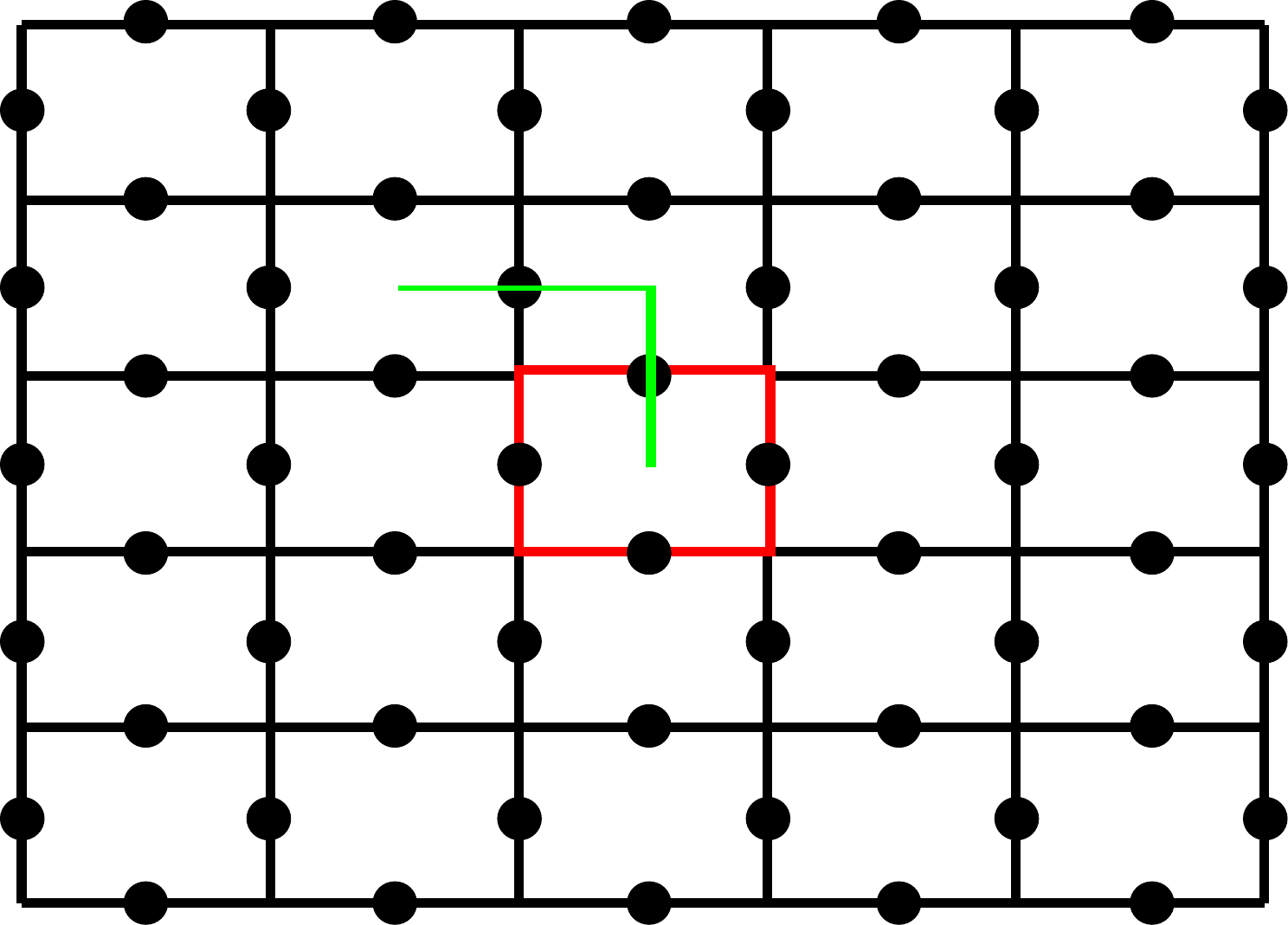}
    \caption{
        Example of an error (green) which doesn't commute with a generator of the stabilizer (red).
    }
    \label{fig:error-anticommute}
\end{figure}
\begin{figure}
    \centering
    \includegraphics[width=3in]{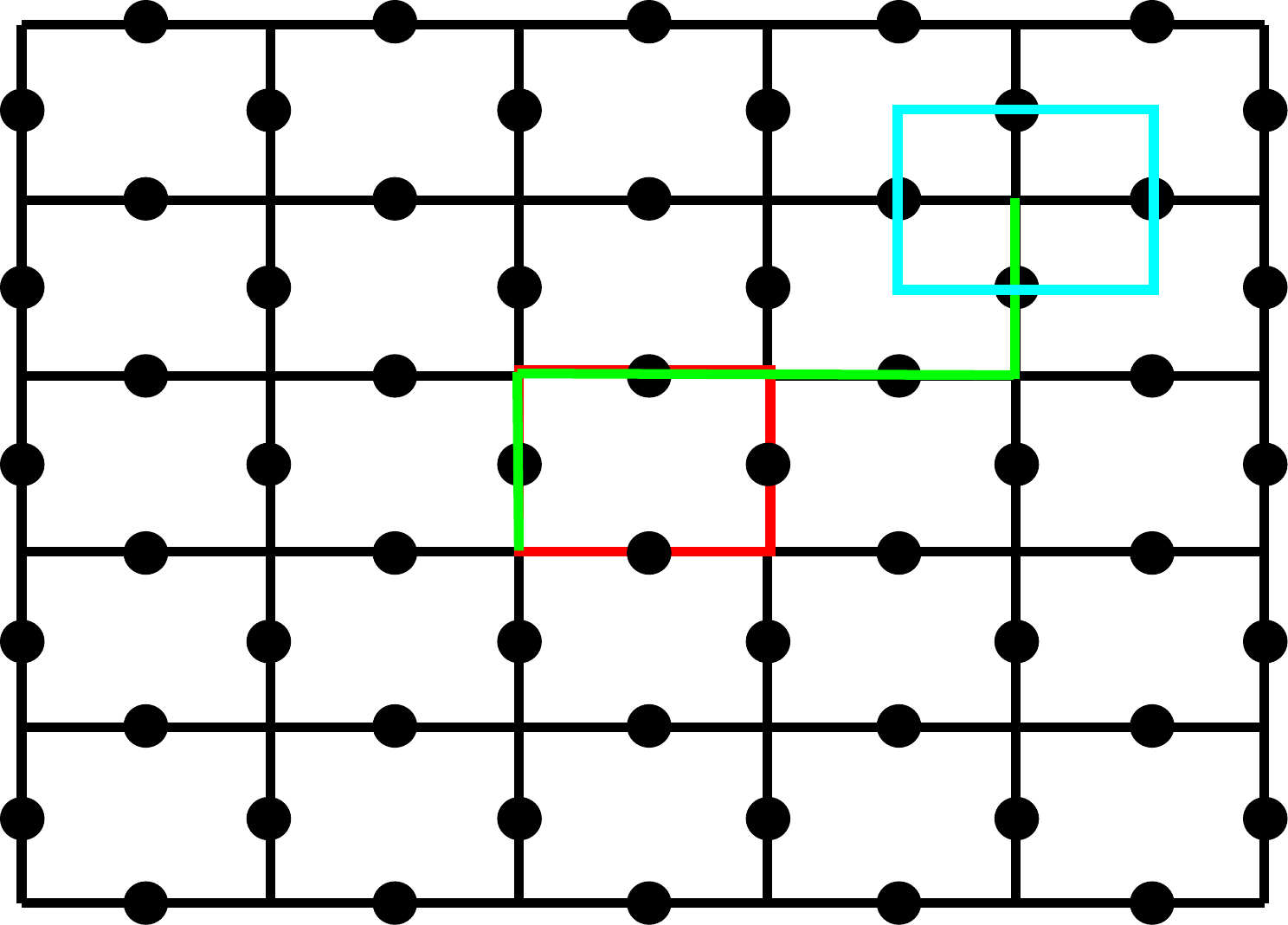}
    \caption{
        Example of an error (green) which commutes with a generator of the stabilizer in the lattice (red) but not a generator of the stabilizer in the dual lattice (blue).
    }
    \label{fig:error-both}
\end{figure}

Some examples are in order. In Figure~\ref{fig:error-anticommute}, an error is shown in green which intersects with the generator of the stabilizer in red at only one qubit. The error therefore anticommutes with this generator. In Figure~\ref{fig:error-both}, the error in green commutes with the generator of the stabilizer in red but anticommutes with the one in blue since they intersect at only one qubit. A \textbf{cycle} is a 1-chain in the lattice (or dual lattice) with a vanishing boundary, i.e. an element of the kernel of $\partial_1$. Observe that a cycle which enters into a square representing a generator of the stabilizer must also exit the square (otherwise it cannot be closed). It therefore follows that the Pauli words $X_d$ and $Z_c$ commute with the stabilizer subgroup whenever the 1-chains $c$ and $d$ are cycles. Since the normalizer of the stabilizer subgroup is equivalent to its centralizer, it follows that $X_d$ and $Z_c$ are each logical operations. An example of this is shown in Figure~\ref{fig:cycle-example}. Moreover, every cycle which is the boundary of a 2-cell is itself an element of the stabilizer subgroup, so only those cycles which are not the boundary of a 2-cell yield non-trivial logical operations. That is, the image of the boundary map $\partial_2$ on $L$ or $L^*$ acts as the identity operation on the logical qubits. Meanwhile, the kernel of $\partial_1$ is generated by the cycles in the lattice (or its dual). Putting these two thoughts together, we have the following beautiful result.

\begin{figure}
    \centering
    \includegraphics[width=3in]{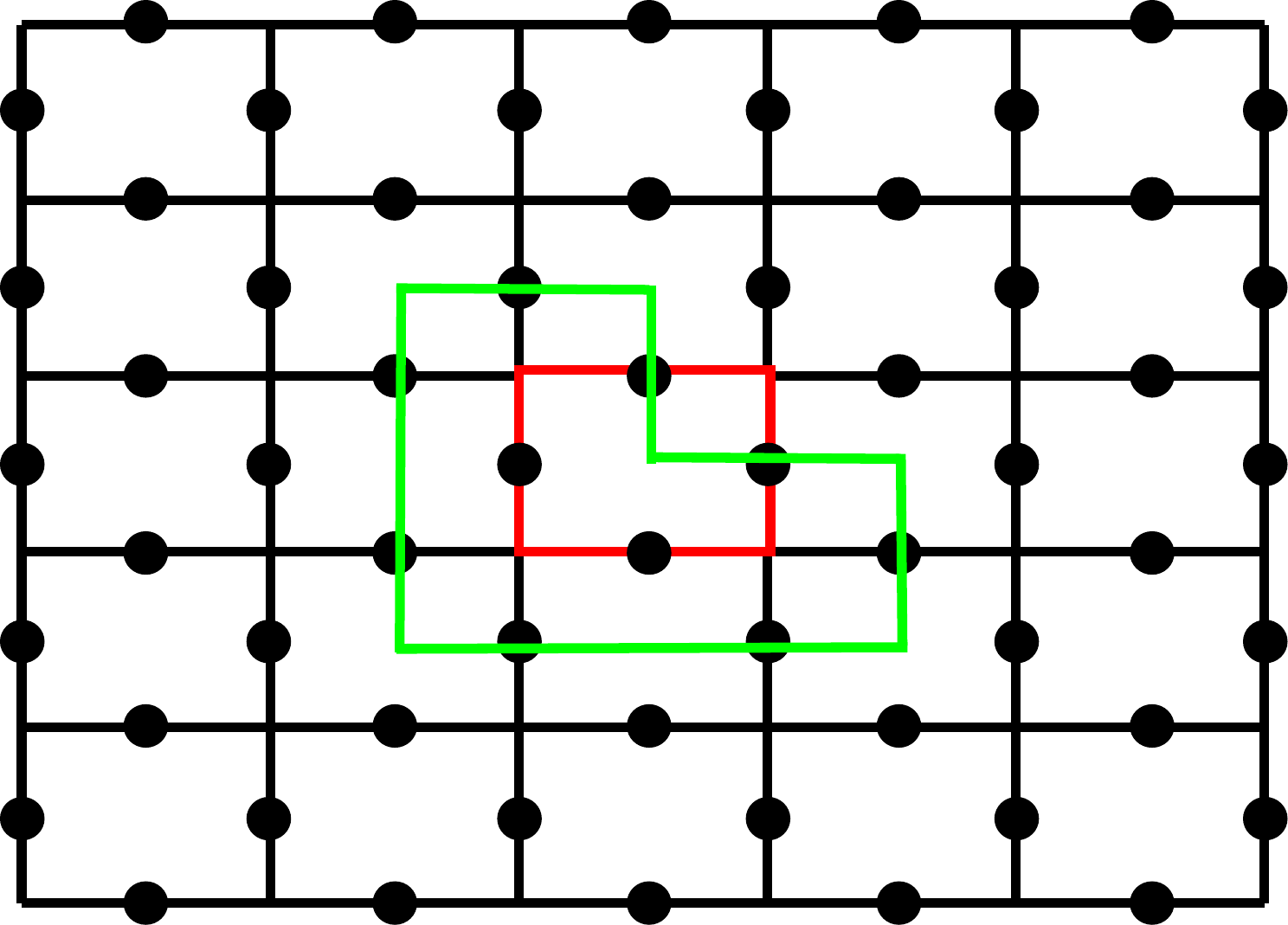}
    \caption{
        Example of a cycle (green) intersecting with a stabilizer generator (red).
    }
    \label{fig:cycle-example}
\end{figure}

\begin{proposition}
    The logical operations in a topological code are the Pauli words corresponding to the elements of $H_1(L,\mathbb{Z})$ and $H_1(L^*,\mathbb{Z})$. 
\end{proposition}
\begin{proof}
    The image of $\partial_2$ is precisely the collection of all boundaries of 2-cells. Meanwhile, the kernel of $\partial_1$ is given by all cycles in the lattice. Thus, the first homology group $H_1=\ker(\partial_2)/\imag(\partial_1)$ consists of all cosets of the set of boundaries of 2-cells in the group of all cycles. That is, every element of the homology has the form $c+\imag(\partial_1)$, where $c$ is a cycle and $\imag(\partial_1)$ is the set of all boundaries. The Pauli word associated to $c$ intersects with all stabilizer generators at an even number of points and is therefore a logical operation as it commutes with every generator of the stabilizer subgroup.
\end{proof}

We have shown that the normalizer of the stabilizer subgroup corresponds to the kernel of $\partial_1$ and the stabilizer subgroup corresponds to the image of $\partial_2$. Thus, the group $N(\mathcal{S})/\mathcal{S}$ of all logical operations corresponds to the first homology groups of the lattice and its dual. Elements in the same coset of the homology group therefore correspond to the same logical operation, and those elements in the image of $\partial_2$ correspond to elements of the stabilizer which act trivially on our logical qubits. Thus, in the homology picture, we can rephrase Proposition~\ref{prop:distance} in the following nice way:

\begin{proposition}
    The distance of a topological code is the minimum length of the non-trivial elements of $H_1(L,\mathbb{Z}_2)$ and $H_1(L^*,\mathbb{Z}_2)$. That is,
    \begin{equation}
        d=\min\{\ell(c):\mathds{1}\ne c\in H_1(L,\mathbb{Z}_2)\cup H_1(L^*,\mathbb{Z}_2)\},
    \end{equation}
    where $\ell(c)$ denotes the number of 1-cells in $c$.
\end{proposition}

\subsection{Toric Codes}\label{sec:toric-codes}

Let us continue our discussion of topological codes but in the context of the \textbf{toric code}. As before, we work on a lattice as in Figure~\ref{fig:grid-face}, but this time we identify the left and right boundary edges as well as the top and bottom boundary edges. We therefore have a surface homeomorphic to a torus, hence the name. From the figure, we see that there are 50 physical qubits in the lattice. Indeed, the diagram shows 60 black dots, but 10 of these have been identified with other dots by the periodic boundary conditions. There are 25 generators of the stabilizer subgroup coming from the 2-cells in the lattice and an additional 25 generators of the stabilizer subgroup coming from the 2-cells in the dual lattice. However, these generators are not independent. The product of the 25 generators in the lattice gives the identity operation, which allows us to remove a generator from this list. Likewise, we can remove a generator from the dual lattice set, leaving a total of 48 independent generators of the stabilizer subgroup. Then by Proposition~\ref{prop:sym-subspace-dim}, the dimension of the code space is $2^{50-48}=4$, and it follows that the toric code describes two logical qubits. Notice that this argument carries over to a lattice of arbitrary size. If the lattice is $m\times n$, then there are $n(m+1)+m(n+1)$ black dots in the diagram, but we need to remove $m+n$ due to the periodic boundary conditions, leaving $2mn$ physical qubits. The number of independent generators of the stabilizer is $2mn-2$, and so the dimension of the code space is $2^{2mn-2mn+2}=4$, independent of $m$ and $n$.

In addition to the cycles on the inside of the lattice diagram (as in Figure~\ref{fig:cycle-example}), the periodic boundary conditions of the torus allow us to produce cycles which travel through the boundary of the diagram and reappear at the other side. Four examples are shown in Figure~\ref{fig:cycle-example-torus}. The horizontal lines in the lattice diagram become circles that enclose the volume of the torus. The vertical lines become circles on the torus with constant radius from the center point. In continuing with our convention that chains in the lattice be labeled with a $c$ and chains in the dual lattice be labeled with a $d$, let us write $c_v$ for the vertical cycle in the lattice and $c_h$ for the horizontal cycle in the lattice. Similarly, we define $d_v$ and $d_h$ for the blue lines in the dual lattice. Observe that $[Z_{c_v},Z_{c_h}]=[X_{d_v},X_{d_h}]=[Z_{c_v},X_{d_v}]=[Z_{c_h},X_{d_h}]=0$, while $Z_{c_v}$ anticommutes with $X_{d_h}$ and $Z_{c_h}$ anticommutes with $X_{d_v}$ (since they intersect at one physical qubit). The two red cycles span the nontrivial homological degrees of freedom of the homology group $H_1(L,\mathbb{Z}_2)$ in the sense that they are representatives of the two non-trivial homology classes, and similarly the two blue cycles span the nontrivial homological degrees of freedom in the dual lattice. Thus, these cycles form all the possible non-trivial logical operations on our two logical qubits. This realization, along with the commutation relations above, motivate us to define these operations as our logical Pauli $X$ and $Z$ gates. Thus, we define
\begin{align}
\begin{aligned}
    \Bar{X}_0 &= X_{d_v}\\
    \Bar{X}_1 &= X_{d_h}\\
    \Bar{Z}_0 &= Z_{d_v}\\
    \Bar{Z}_1 &= Z_{d_h},
\end{aligned}
\end{align}
where the overbar indicates that the gate is a logical operation.

\begin{figure}
    \centering
    \includegraphics[width=3in]{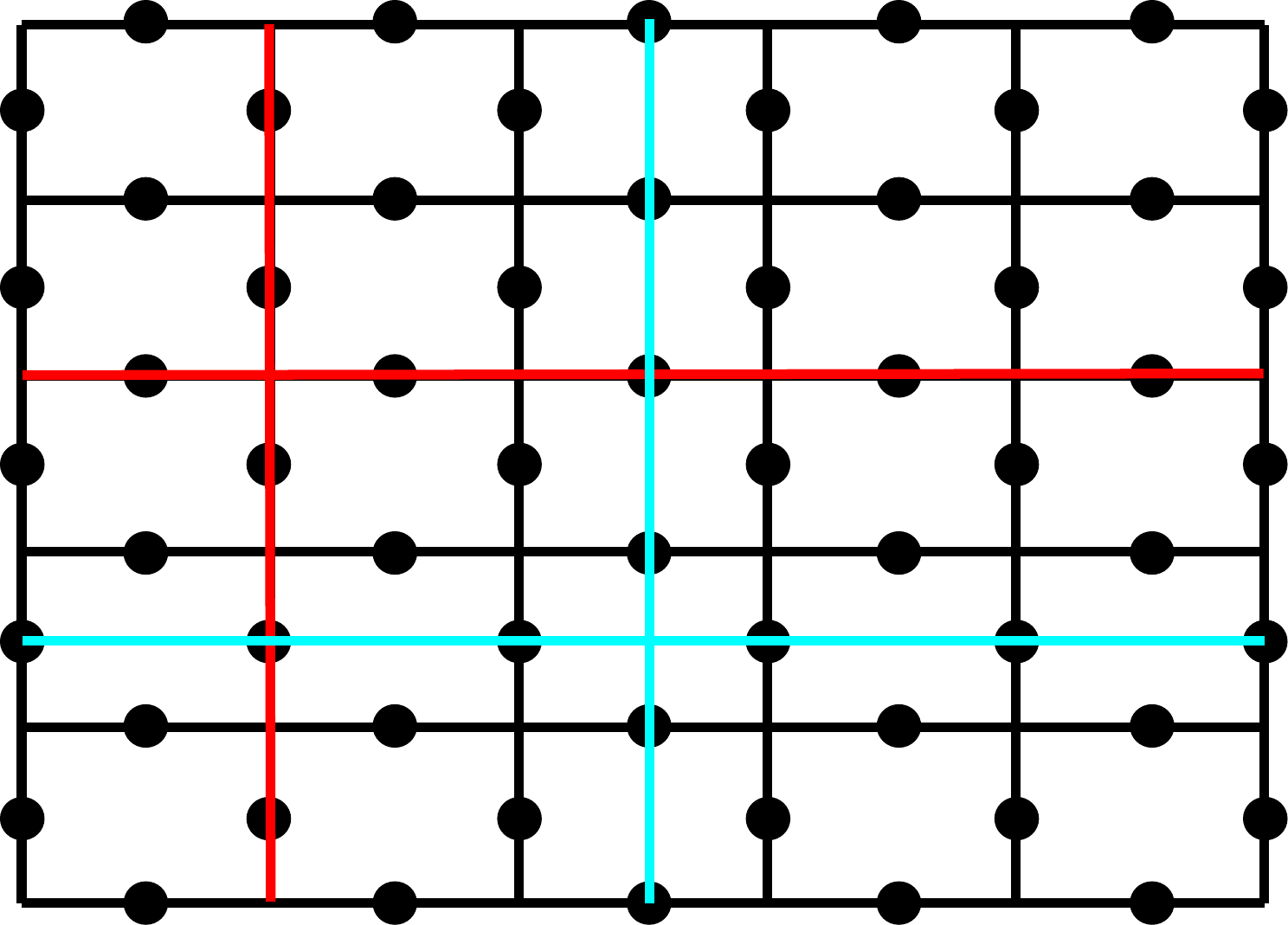}
    \caption{
        Four examples of cycles that make use of the periodic boundary conditions of the torus.
    }
    \label{fig:cycle-example-torus}
\end{figure}

For the sake of clarity, we should give a method for explicitly constructing elements of the code space. Recall that we take the $\mathcal{S}$-symmetric subspace as our code space. From Lemma~\ref{lemma:gprojection}, we know that the projection onto the $\mathcal{S}$-symmetric subspace is given by
\begin{equation}
    \Pi_\mathcal{S}=\frac{1}{\lvert \mathcal{S}\rvert}\sum_{S\in\mathcal{S}}S,
\end{equation}
but instead of writing out every element of $\mathcal{S}$, let us give a simpler approach. First start with a state $\ket{\psi}$ that is invariant under the action of every Pauli-$Z$ generator $Z_{\partial(f)}$. This could, for example, be the state $\ket{0\cdots0}$. Notice that for any cycle $d$ in the dual lattice, we have $Z_{\partial(f)}X_d\ket{\psi}=X_dZ_{\partial(f)}\ket{\psi}=X_d\ket{\psi}$, which establishes that $X_d\ket{\psi}$ is also invariant with respect to $Z_{\partial(f)}$. Now take $d=d_v+b$ and sum over all boundaries $b$ of 2-chains in the dual lattice to produce the state
\begin{equation}
    \ket{\Psi_v}=\frac{1}{N}\sum_{b\in \partial C_2(L^*,\mathbb{Z}_2)} X_{d_v+b}\ket{\psi},
\end{equation}
where $N$ is a normalization constant. By our previous observation, this state will be invariant under the action of a Pauli-$Z$ generator. Since the collection of all boundaries in the dual lattice forms a group, we also see that this state is invariant under the action of a Pauli-$X$ generator. Thus, $\ket{\Psi_v}$ is an element of the code space (as it's in the $\mathcal{S}$-symmetric subspace). Replacing $v\to h$ produces the state
\begin{equation}
    \ket{\Psi_h}=\frac{1}{N}\sum_{b\in \partial C_2(L^*,\mathbb{Z}_2)} X_{d_h+b}\ket{\psi},
\end{equation}
which is also in the code space. A similar construction with the $X$'s and $Z$'s swapped produces the states
\begin{equation}
    \ket{\Phi_v}=\frac{1}{N}\sum_{b\in \partial C_2(L,\mathbb{Z}_2)}Z_{c_v+b}\ket{\phi}
\end{equation}
and
\begin{equation}
    \ket{\Phi_h}=\frac{1}{N}\sum_{b\in \partial C_2(L,\mathbb{Z}_2)}Z_{c_h+b}\ket{\phi}.
\end{equation}
We will now build a logical basis using one of these choices. Let us choose $\ket{00}_L:=\ket{\Psi_v}$ with $\ket{\psi}=\ket{0\cdots0}$ for simplicity. Observe how each of the logical Pauli-$X$ operations acts on it. We have
\begin{equation}
    \ket{10}_L:=\Bar{X}_0\ket{00}_L=\frac{1}{N}\sum_{b\in \partial C_2(L^*,\mathbb{Z}_2)}X_{d_v}X_{d_v+b}\ket{0\cdots0}=\frac{1}{N}\sum_{b\in \partial C_2(L^*,\mathbb{Z}_2)}X_{b}\ket{0\cdots0},
\end{equation}
\begin{equation}
    \ket{01}_L:=\Bar{X}_1\ket{00}_L=\frac{1}{N}\sum_{b\in \partial C_2(L^*,\mathbb{Z}_2)}X_{d_h}X_{d_v+b}\ket{0\cdots0}=\frac{1}{N}\sum_{b\in \partial C_2(L^*,\mathbb{Z}_2)}X_{d_h+d_v+b}\ket{0\cdots0},
\end{equation}
and
\begin{equation}
    \ket{11}_L:=\Bar{X}_0\Bar{X}_1\ket{00}_L=\frac{1}{N}\sum_{b\in \partial C_2(L^*,\mathbb{Z}_2)}X_{d_v}X_{d_h}X_{d_v+b}\ket{0\cdots0}=\frac{1}{N}\sum_{b\in \partial C_2(L^*,\mathbb{Z}_2)}X_{d_h+b}\ket{0\cdots0},
\end{equation}
which we note is $\ket{\Psi_{d_h}}$. The logical Pauli-$Z$ operations act as one would expect on the these logical states thanks to the commutation relations between the $\Bar{X}_i$ and $\Bar{Z}_j$. By taking inner products between each of these states, it is easily seen that they are orthogonal and therefore form a linearly independent set in the code space, hence a basis. As an example, consider the inner product $\braket{00\vert10}_L$. We have
\begin{align}
    \begin{aligned}\braket{00\vert10}_L&=\frac{1}{N^2}\sum_{b\in \partial C_2(L^*,\mathbb{Z}_2)}\sum_{b'\in \partial C_2(L^*,\mathbb{Z}_2)}\bra{0\cdots0}X_{d_v+b}X_{b'}\ket{0\cdots0}\\
    &=\frac{1}{N^2}\sum_{b\in \partial C_2(L^*,\mathbb{Z}_2)}\sum_{b'\in \partial C_2(L^*,\mathbb{Z}_2)}\bra{0\cdots0}X_{d_v+b+b'}\ket{0\cdots0},
    \end{aligned}
\end{align}
and since $\ket{0\cdots0}$ is orthogonal to $X_{d_v+b+b'}\ket{0\cdots0}$ when $X_{d_v+b+b'}\ne\mathds{1}$, all of the remaining inner products vanish, establishing the orthogonality between $\ket{00}_L$ and $\ket{10}_L$. Thus, our encoder will take an arbitrary 2-qubit state $\alpha\ket{00}+\beta\ket{01}+\gamma\ket{10}+\delta\ket{11}$ to the logical state $\alpha\ket{00}_L+\beta\ket{01}_L+\gamma\ket{10}_L+\delta\ket{11}_L$. The key insight here is that if we start with an element of the code space, we can construct a basis for it by applying the logical $X$ operations. The logical $Z$ operations will then act on a logical qubit in the same way a Pauli-$Z$ acts on a physical qubit thanks to the commutation relations between the logical $X$ and $Z$ operators.

We have constructed a basis for the code space in order to understand the encoder portion of the toric code. We have also constructed the stabilizer subgroup, which we can use to do syndrome extraction as discussed in Section~\ref{sec:stabilizer}. The errors we can detect with this stabilizer code are those that anticommute with one of the generators of the stabilizer. Theorem~\ref{thm:stab-error} then tells us what errors we can correct, and decoding is done by applying the adjoint of whatever error operation is identified, which is easier said than done for large codes. According to Proposition~\ref{prop:errors}, all errors that belong to the same coset of the normalizer in the Pauli group give rise to the same syndrome. Thus, the 1-chains that are equivalent up to the addition of a cycle will produce the same syndrome.

As a final remark, notice that the distance for a toric code is $\min(m,n)$ for an $m\times n$ lattice because the logical Pauli operations consist of Pauli words of length $m$ and $n$. Thus, as the lattice grows larger, the distance of the code grows, allowing us to correct more and more errors. Unfortunately, this comes with an overhead in the number of physical qubits needed to encode our two logical qubits. Indeed, as we have already pointed out, an $m\times n$ lattice with toroidal boundary conditions contains $2mn$ physical qubits. The toric code is therefore a $[[2mn,2,\min(m,n)]]$ code with rate $R=\frac{1}{mn}$.

\subsection{Planar Codes}\label{sec:planar-codes}
It is possible to define a quantum error correction code on a lattice without the periodic boundary conditions of the torus. If we attempt to do so in the naive way, observe that the number of physical qubits on the lattice will be $2mn+m+n$, while the number of independent generators of the stabilizer will be given by counting all faces and vertices in the lattice and subtracting out any dependent generators. Due to the lack of boundary conditions, the Pauli-$Z$ words coming from the faces each act on at least one edge that no other word acts on, and so they are all independent. Since every edge has two vertices, the product of the Pauli-$X$ words generated by the vertices will be the identity, allowing us to remove one such generator. This gives us a total of $mn+(m+1)(n+1)-1=2mn+m+n$ independent generators of the stabilizer subgroup. Thus, the dimension of the code space would be $2^0=1$, which is incapable of encoding a logical qubit. We must therefore modify our construction. Instead of encoding our logical operators in the homology of this lattice, we will make use of a notion of relative homology \cite{bravyi1998}.

\begin{figure}
    \centering
    \includegraphics[width=3in]{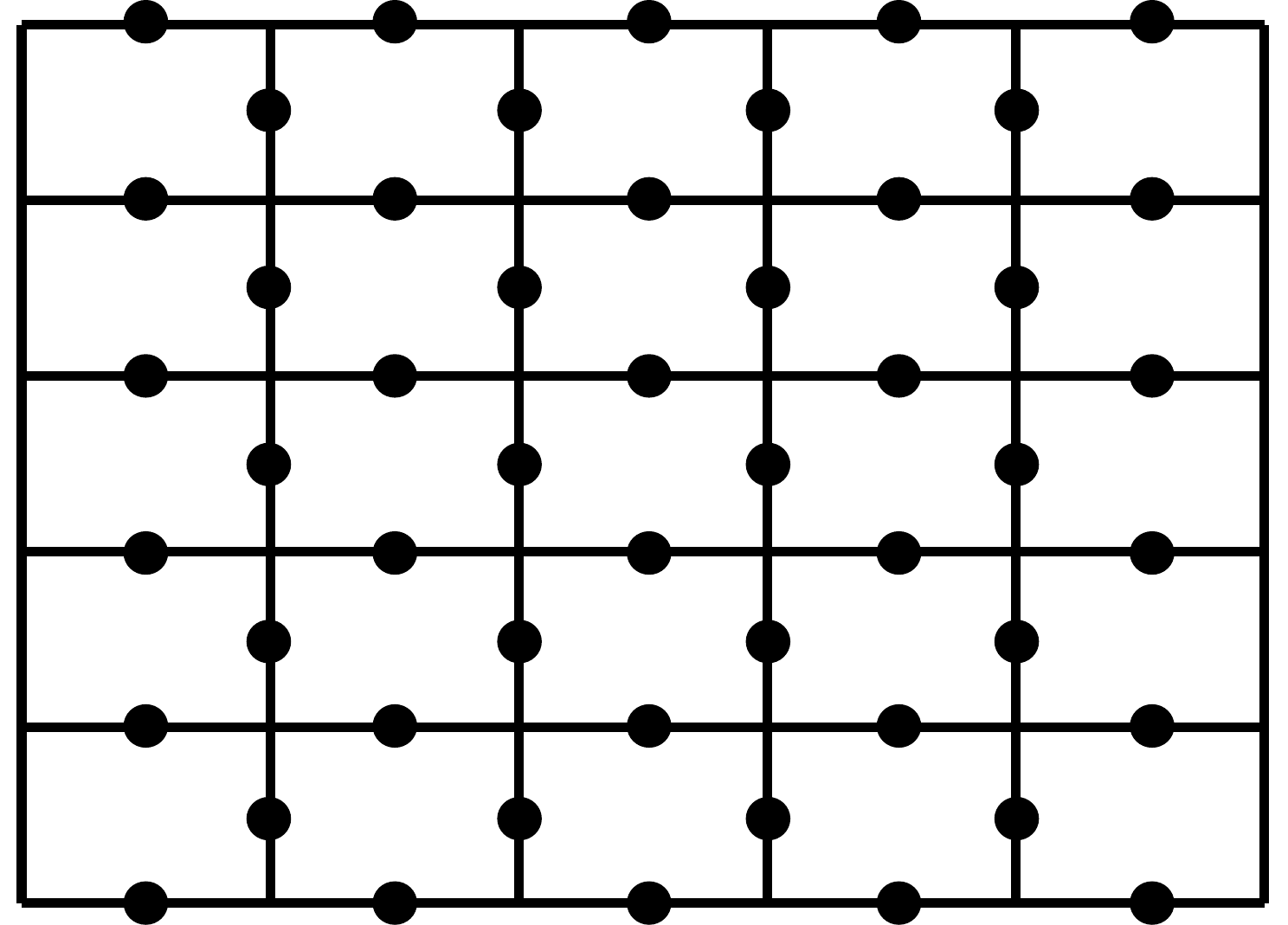}
    \caption{
        Example lattice for the planar code.
    }
    \label{fig:planar-lattice}
\end{figure}

We start by once again drawing a lattice, but this time we place qubits on all edges except for the vertical boundary edges as in Figure~\ref{fig:planar-lattice}. The upper and lower portions of the boundary of this lattice are sometimes called the \textbf{smooth boundary}, and the remaining vertical edges are often referred to as the \textbf{rough boundary}. We still denote the group of $k$-chains on the lattice by $C_k(L,\mathbb{Z}_2)$, but we also take note of the group of $k$-chains on the rough boundary $C_k(B_r,\mathbb{Z}_2)$, where we have used $B_r$ to indicate the rough boundary. Since the chain groups are abelian, $C_k(B_r,\mathbb{Z}_2)$ is a normal subgroup of $C_k(L,\mathbb{Z}_2)$. Thus, we can form the quotient group $C_k(L,B_r;\mathbb{Z}_2)=C_k(L,\mathbb{Z}_2)/C_k(B_r,\mathbb{Z}_2)$ in which all $k$-chains in $C_k(B_r,\mathbb{Z}_2)$ are identified with the trivial element. We call this the group of \textbf{relative $k$-chains}.

\begin{figure}
    \centering
    \includegraphics[width=3in]{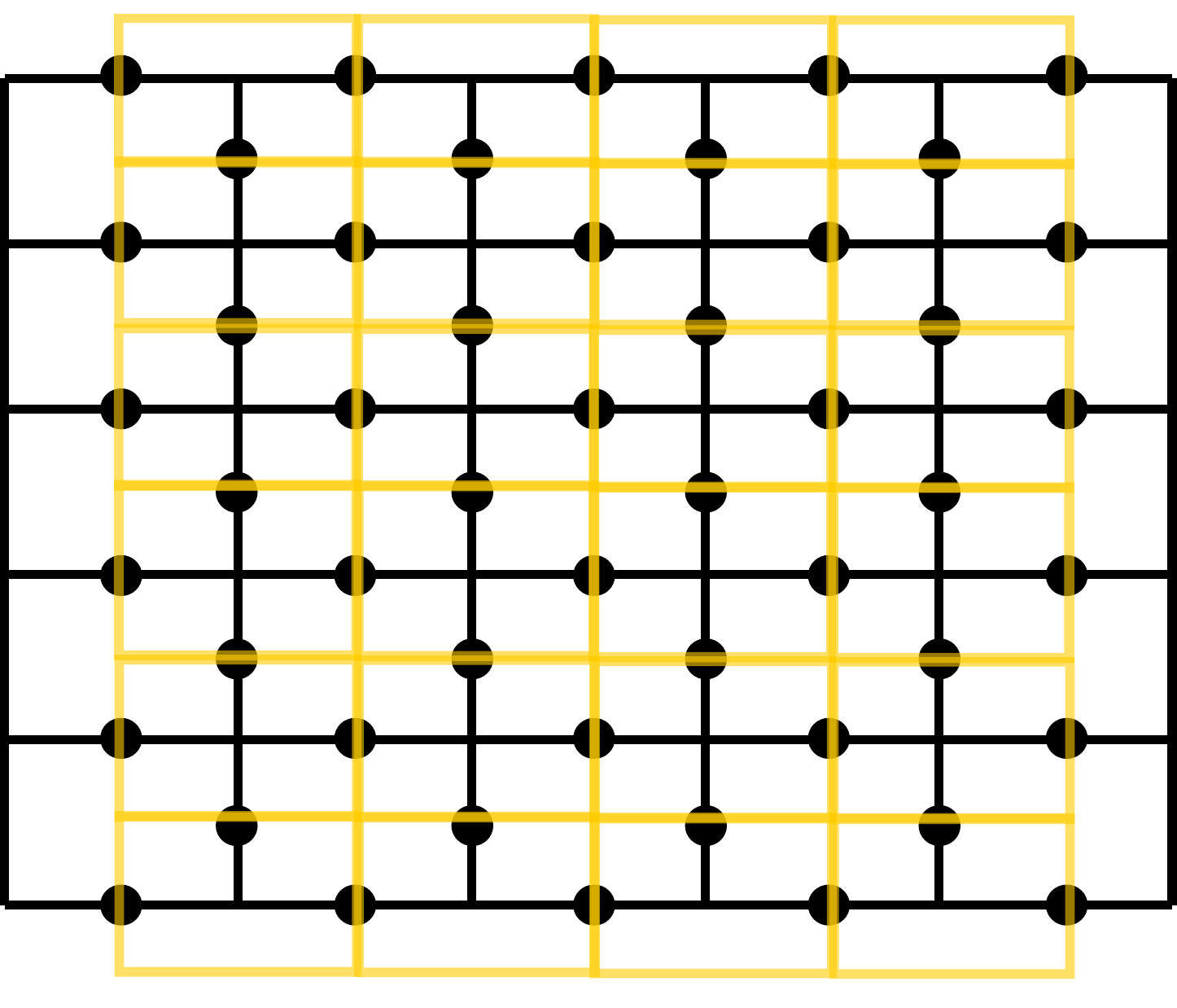}
    \caption{
        Dual lattice (gold) on top of lattice for the planar code.
    }
    \label{fig:planar-dual-lattice}
\end{figure}

Just as in the toric code, we can form the dual of the lattice by transforming all 0-cells to 2-cells and all 2-cells to 0-cells. We do so with the same procedure; in the center of every 2-cell, we place a 0-cell, and we connect the 0-cells in adjacent 2-cells with a 1-cell (we connect our new vertices in adjacent faces with an edge). This time, however, we add the additional step of closing the lattice at the top and bottom with a rough boundary. Thus, the dual operation transforms a rough boundary to a smooth boundary and a smooth boundary to a rough boundary. An example is shown in Figure~\ref{fig:planar-dual-lattice}. 

Observe that we can also form the quotient chain group $C_k(L^*,B_r^*;\mathbb{Z}_2)=C_k(L^*,\mathbb{Z}_2)/C_k(B_r^*,\mathbb{Z}_2)$ for the dual lattice. Just as in the toric code, we can therefore define homomorphisms that take chains to Pauli words, but this time we respect the structure of the quotient group. To this end, we define $Z:C_k(L,B_r;\mathbb{Z}_2)\to G_n$ by 
\begin{equation}
    Z_c=\prod_{\substack{e\in c\\e\notin B_r}}Z_e,
\end{equation}
where the product is taken over all edges in the 1-chain which are not in the rough boundary. Similarly, we define $X:C_k(L^*,B_r^*;\mathbb{Z}_2)\to G_n$ by
\begin{equation}
    X_d=\prod_{\substack{e\in d\\e\notin B_r^*}}Z_e,
\end{equation}
where the product is over all edges in the 1-chain $d$ which are not in the rough boundary of the dual lattice. We then take the generators of the stabilizer subgroup to be the collection of all independent elements of the form $Z_{\partial(f)}$ and $X_{\partial(v^*)}$, where $f$ is a 2-cell in $L$ and $v^*$ is the 2-cell dual to a 0-cell $v$. We are now in the position to determine the dimension of the code space. Suppose our lattice is $m\times n$. Then there are $mn$ Pauli-$Z$ generators and $(m+1)(n-1)$ Pauli-$X$ generators (this is best seen from Figure~\ref{fig:planar-dual-lattice} rather than counting vertices in Figure~\ref{fig:planar-lattice} because the vertices in the rough boundary do not contribute).  Since these generators are all independent and there are $n(m+1)+m(n-1)$ qubits, this leaves us with a code space dimension of $2^{2nm+n-m-(2mn-m+n-1)}=2$. Thus, the planar code we have just constructed is for a single logical qubit.

With this in mind, we should construct our logical $X$ and $Z$ operations. In the lattice, we take a horizontal chain $c$ and in the dual lattice, we take a vertical chain $d$. Since we are working in the relative chain group, where the rough boundary has been made trivial, each of these chains is in fact a cycle. Moreover, these cycles intersect with each stabilizer generator an even number of times, so they commute with the stabilizer subgroup and can therefore be considered logical operators. Neither of them is the logical identity operator because they do not belong to the stabilizer itself (they are not the boundary of a 2-cell). Topologically, we can compute the first (relative) homology group of this lattice and its dual, and we will find that they each have a single generator. A generator of the relative homology of the lattice is $c$ and a generator of the relative homology of the dual lattice is $d$, again justifying that these are non-trivial logical operations. Finally, since $c$ and $d$ intersect at a qubit exactly once, the corresponding Pauli words $Z_c$ and $X_d$ anticommute. Thus, we define our logical $X$ operation as $\Bar{X}=X_d$ and our logical $Z$ operation as $\Bar{Z}=Z_c$.

To build the code space, we proceed as we did with the toric code. We observe that $\Bar{Z}\ket{0\cdots0}=\ket{0\cdots0}$ and then define the logical zero state to be
\begin{equation}\label{eq:planar-logical0}
    \ket{0}_L := \frac{1}{N}\sum_{b\in \partial C_2(L^*,\mathbb{Z}_2)}X_{d+b}\ket{0\cdots0},
\end{equation}
where $N$ is again a normalization factor. Letting $\Bar{X}$ act on this state produces the logical one state
\begin{equation}\label{eq:planar-logical1}
    \ket{1}_L := \frac{1}{N}\sum_{b\in \partial C_2(L^*,\mathbb{Z}_2)}X_{b}\ket{0\cdots0}.
\end{equation}
It is easily verified that $\Bar{X}\ket{1}_L=\ket{0}_L$, $\Bar{Z}\ket{0}_L=\ket{0}_L$, and $\Bar{Z}\ket{1}_L=-\ket{1}_L$, as expected.

\begin{example}
    Consider the $m=1,n=2$ planar code. The number of physical qubits is $2nm+n-m=5$. There are 2 generators of the stabilizer produced by the lattice and an additional 2 stabilizers produced by the dual lattice. If we label the 5 qubits from left to right and from top to bottom, the upper horizontal line contains qubits 0 and 1, the middle vertical line contains qubit 2, and the lower horizontal line contains qubits 3 and 4. Thus, the stabilizers are $Z_0Z_2Z_3$ and $Z_1Z_2Z_4$ from the lattice and $X_0X_1X_2$ and $X_2X_3X_4$ from the dual lattice. For the logical $X$ gate, we choose the left most vertical line in the dual lattice, which corresponds to $\Bar{X}=X_d=X_0X_3$ and for the logical $Z$ gate, we choose the upper horizontal line in the lattice, which corresponds to $\Bar{Z}=Z_c=Z_0Z_1$.

    Let us now generate the logical basis states $\ket{0}_L$ and $\ket{1}_L$ by following \eqref{eq:planar-logical0} and \eqref{eq:planar-logical1}. The boundary terms $X_b$ are given by the boundaries of all possible chains of faces in the dual lattice, including the zero 2-chain consisting of no faces. Thus, the possible boundary terms are $\mathds{1}$, $X_0X_1X_2$, $X_2X_3X_4$, and $X_0X_1X_3X_4$. It follows from \eqref{eq:planar-logical0} that the logical zero state is
    \begin{equation}
        \ket{0}_L=\frac12(\ket{10010}+\ket{01110}+\ket{10101}+\ket{01001})
    \end{equation}
    and it follows from \eqref{eq:planar-logical1} that the logical one state is
    \begin{equation}
        \ket{1}_L=\frac12(\ket{00000}+\ket{11101}+\ket{00111}+\ket{11011}).
    \end{equation}
    The reader can now verify that the logical $X$ and $Z$ gates act as intended.

    Since there are four syndrome bits, there are $2^4=16$ possible syndromes, and the all zero syndrome corresponds to the non-error. Thus, this code is capable of correcting 15 possible errors. A lookup table is given in Table~\ref{tab:m1n2syndromes}. The first syndrome bit corresponds to the generator $Z_0Z_2Z_3$, followed by $Z_1Z_2Z_4$, $X_0X_1X_2$, and finally $X_2X_3X_4$. Notice that this code does not protect against all single qubit errors since, for example, $X_0$ produces the same syndrome as $X_3$.
    
\end{example}

\begin{table}[]
    \centering
    \begin{tabular}{cc} \toprule
        Syndrome & Error \\ \midrule
        0000 & $I$\\
        0001 & $Z_3$\\
        0010 & $Z_0$\\
        0011 & $Z_2$\\
        0100 & $X_1$\\
        0101 & $Y_4$\\
        0110 & $Y_1$\\
        0111 & $X_1Z_2$\\ \bottomrule
    \end{tabular}
    \hspace{2em}
    \begin{tabular}{cc} \toprule
        Syndrome & Error \\ \midrule
        1000 & $X_0$\\
        1001 & $Y_3$\\
        1010 & $Y_0$\\
        1011 & $X_0Z_2$\\
        1100 & $X_2$\\
        1101 & $X_2Z_3$\\
        1110 & $X_2Z_0$\\
        1111 & $Y_2$\\ \bottomrule
    \end{tabular}
    \caption{Syndrome look-up table for $m=1$, $n=2$ planar code. For each syndrome, we choose a representative of the error class with minimum weight.}
    \label{tab:m1n2syndromes}
\end{table}

Let us close this section with a quick analysis of the planar code. The distance of the planar code on an $m\times n$ lattice is again $d=\min(m,n)$, since $\Bar{X}$ is a product of $m$ Pauli-$X$ gates and $\Bar{Z}$ is a product of $n$ non-trivial Pauli-$Z$ gates. On the other hand, the number of qubits in the code is $2nm+n-m$, which is fewer than the number of physical qubits in the toric code when $m>n$. Thus, the planar code has the same distance as the toric code but at a slightly smaller overhead when $m>n$. Although, it produces only one logical qubit. Putting this information together shows that the planar code is a $[[2mn+n-m,1,\min(m,n)]]$ code with rate $R=\frac{1}{2mn+n-m}$.

\subsection{Discussion and Further Reading}

We have given only a small taste of topological codes in the preceding discussion, and the reader may wonder whether these codes are merely a mathematical curiosity or in fact useful for doing error correction. The planar codes in particular, usually confusingly called surface codes in the literature, are one of the best known classes of codes for several reasons, not the least of which is its simple and mathematically pleasing construction. In addition to this, it turns out that planar codes have comparatively high error thresholds around $p_{th}=1\%$. Indeed, if the physical error rate (per gate, measurement, or qubit) is below this threshold, then the logical error rate will exponentially decrease with the distance of the code. This does, however, come at the cost of more qubits since the number of physical qubits increases quadratically with the distance in a planar code. When the physical error rate is above this threshold, it is not beneficial to run this code in practice because the errors accumulate at a faster rate than they can be corrected and the code becomes overwhelmed by noise. There is a variant of the planar code we have presented here called the rotated surface code \cite{fowler2012,terhal2015}, which uses fewer physical qubits to achieve the same distance.

We have one further problem if our goal is to create a fault-tolerant universal device using something like a planar code, which is that the planar code encodes only one logical qubit and tells us only how to construct logical Pauli gates. Thus, we currently have no means of interacting the logical qubits through a 2-qubit logical gate. There have been several approaches to implementing 2-qubit gates on a planar code. Perhaps the first was the transversal method in which planar codes are stacked on top of each other and, for example, a logical $\cnot$ gate between the logical qubits is implemented by applying a $\cnot$ gate from the physical qubit on one planar code to the corresponding physical qubit on the remaining planar code. This construction, however can introduce a number of problems, and more recently, the lattice surgery approach, which uses only nearest-neighbor operations, has been favored. The surgery approach departs from the topological perspective we have introduced here by introducing a discontinuous operation on the lattice which is analogous to the surgery theory of geometric topology. Two such operations are introduced, called \textbf{merging} and \textbf{splitting}. We perform merging by introducing enough qubits to combine two planar codes along a rough boundary, and we perform splitting by removing a row of physical qubits from the lattice using measurements, leaving two independent planar codes. The process of merging and then splitting codes is shown to implement entangling gates, such as the $\cnot$ operation. For the details, we refer the reader to \cite{horsman2012,campbell2017}. In addition to lattice surgery, other methods include code deformation \cite{bombin2009} and twist-based encodings \cite{fowler2012}, where defects or domain walls are moved across the lattice to braid logical qubits. These approaches are deeply rooted in the topological nature of the code and can also be used to implement logical Clifford or even non-Clifford gates in conjunction with magic state distillation \cite{campbell2017}.

The performance of surface codes in practice also depends on the efficiency and accuracy of decoding algorithms \cite{deMartiiOlius2024}, such as minimum weight perfect matching \cite{edmonds1965,dennis2002,fowler2009,fowler2014} and belief propagation \cite{poulin2008,duclos-cianci2010}, which attempt to infer the most likely error given a syndrome. Machine learning-based decoders, including those using neural networks, are also being developed with the goal of achieving near-optimal decoding while reducing runtime \cite{baireuther2018,varsamopoulos2017,varsamopoulos2020}. These approaches show promise, especially for adapting to correlated and hardware-specific noise, though generalization and scalability remain open challenges. We explore the machine learning approach to decoders further in Section~\ref{sec:decoderML}. Finally, decoder latency and hardware compatibility are becoming increasingly important considerations, as real-time decoding will be essential for fault-tolerant control loops in near-term quantum devices. Recent experiments \cite{Gambetta2017,Krinner2022,willow2024} have demonstrated small-scale logical qubits based on surface codes, with multiple rounds of stabilizer measurement and real-time decoding. These implementations mark a significant step toward realizing scalable fault-tolerant quantum computation.

\section{Machine Learning Based Decoders}\label{sec:decoderML}

In a quantum error correcting code, decoding refers to the task of inferring the most likely error from the observed syndrome and determining the appropriate correction operation to undo it. For small codes, this is as simple as creating a look-up table matching each syndrome to an error, but this quickly becomes infeasible as the number of syndromes scales exponentially with the number of stabilizer generators. For stabilizer codes, especially those with a topological structure (such as surface codes), this problem is clearly computationally demanding. Alternative decoding methods, such as minimum-weight perfect matching (MWPM) \cite{edmonds1965,dennis2002,fowler2009,fowler2014} and belief propagation \cite{poulin2008,duclos-cianci2010}, exploit the structure of the code and the assumed noise model but often struggle to generalize. In recent years, machine learning (ML) has emerged as a promising alternative, offering decoders that can learn to decode directly from syndrome data without requiring an explicit error model.

The motivation for using machine learning in decoding is twofold. First, quantum hardware exhibits noise that is often non-Markovian, spatially correlated, or otherwise poorly modeled by simple Pauli error channels \cite{kam2024,liu2024}. Machine learning approaches have the capacity to learn effective decoding strategies in the presence of such complex noise. Second, once trained, many ML-based decoders can execute inference extremely quickly, making them attractive for real-time decoding in quantum processors where latency is a critical constraint. A variety of ML techniques have been applied to the decoding problem, including supervised learning approaches, which treat decoding as a classification task. Given a syndrome, we wish to predict the best correction to the error that occurred. Convolutional neural networks (CNNs) have been used for this purpose \cite{davaasuren2020,gicev2023}. Recurrent neural networks (RNNs) and transformer-based models have also been proposed to capture temporal correlations in repeated syndrome measurements \cite{baireuther2018,varsamopoulos2020decoding}.

Reinforcement learning (RL) has provided another avenue, especially for cases where the decoding process itself is treated as a sequential decision-making task \cite{sweke2020,matekole2022,Domingo-Colomer2020}. In RL-based decoders, an agent learns to propose correction operations in response to syndromes, receiving a reward signal based on whether the final logical state matches the intended target. This approach has been used to build decoders for stabilizer codes and to address challenges in decoding when noise is highly structured.

Despite their promise, ML-based decoders face several challenges. Training requires large, representative datasets of syndrome-error pairs, which can be computationally expensive to generate, especially for large-distance codes. ML approaches continue to evolve rapidly, and their integration with near-term quantum devices remains an active area of research. The ultimate goal is to produce decoders that are fast, adaptive, and robust enough to operate under realistic hardware conditions and noise models.

In Section~\ref{sec:classical_ml_review}, we provide a brief review of some classical machine learning concepts for the uninitiated. We then discuss the use of various machine learning models in the development of an ML-based decoder, including feedforward and convolutional neural networks in Section~\ref{sec:feedforward}, recurrent neural networks in Section ~\ref{sec:recurrent}, and transformers in Section~\ref{sec:transformer}. We also offer insight into how hybrid ML/non-ML models and variational quantum circuit models could function in Section~\ref{sec:hybrid}.

\subsection{Review of Classical Machine Learning}
\label{sec:classical_ml_review}

Machine learning has become the dominant method for learning functional relationships from data. Unlike traditional programming, where rules are explicitly coded by humans, machine learning systems infer rules directly from observed data through optimization. In the context of quantum error correction, the decoding problem is naturally framed as a supervised learning task, where the goal is to infer the most likely logical error given a syndrome measurement. Because decoding must be fast and accurate, and the mapping from syndrome to correction is often highly nonlinear and degenerate, machine learning offers a promising alternative to heuristic or combinatorial decoding methods.

Supervised learning \cite{bishop2006pattern} begins with a collection of training examples $\{(x_i, y_i)\}_{i=1}^N$, where each $x_i$ is an input (such as a syndrome vector or sequence) and $y_i$ is a target output (such as a logical Pauli operator or correction label). The learning algorithm aims to find a function $f_\theta : \mathcal{X} \to \mathcal{Y}$ from a parameterized family of functions that minimizes the expected loss between predictions $\hat{y}_i = f_\theta(x_i)$ and true labels $y_i$. For many practical purposes, this expected loss is approximated by an empirical loss over the training set. Common loss functions include the mean squared error,
\begin{equation}
L_{\mathrm{MSE}}(\hat{y}, y) = \|\hat{y} - y\|^2,
\end{equation}
and the cross-entropy loss,
\begin{equation}
L_{\mathrm{CE}}(\hat{y}, y) = -\sum_{k=1}^{K} y_k \log(\hat{y}_k),
\end{equation}
used when $y$ is a categorical variable and $\hat{y}$ is a probability distribution over classes \cite{murphy2012machine}. The optimization problem is to minimize the total loss
\begin{equation}
\mathcal{L}(\theta) = \frac{1}{N} \sum_{i=1}^N L(f_\theta(x_i), y_i),
\end{equation}
typically using variants of stochastic gradient descent \cite{bottou2010large}.

The earliest and most simple neural architectures are feedforward neural networks, also called multilayer perceptrons. These consist of a sequence of layers, each performing an affine transformation followed by a non-linear activation function. Mathematically, the output of the $\ell$-th layer is given by
\begin{equation}
h^{(\ell)} = \sigma(W^{(\ell)} h^{(\ell-1)} + b^{(\ell)}),
\end{equation}
where $W^{(\ell)}$ is a weight matrix, $b^{(\ell)}$ is a bias vector, and $\sigma$ is a nonlinearity such as the rectified linear unit (ReLU), sigmoid, or hyperbolic tangent. The input to the first layer is $h^{(0)} = x$, and the output layer typically applies a softmax activation to produce probabilities over the target classes. The universal approximation theorem \cite{hornik1989} guarantees that such networks can approximate any continuous function on a compact domain, given sufficient width and depth. However, in practice, careful regularization and architecture design are necessary to ensure generalization \cite{goodfellow2016deep}.

While feedforward networks treat inputs as unstructured vectors, many types of data, including syndromes on a lattice, have spatial structure that can be exploited. Convolutional neural networks \cite{lecun1998} were designed for image data, where locality and translational symmetry are key. CNNs replace fully connected layers with convolutional layers, in which small kernels are applied across the input via a sliding window. This reduces the number of parameters and captures local spatial patterns. A 1D convolution for input $x \in \mathbb{R}^d$ and kernel $w \in \mathbb{R}^k$ is given by
\begin{equation}
(x * w)_i = \sum_{j=0}^{k-1} x_{i+j} w_j,
\end{equation}
and a similar formula applies in two dimensions. CNNs often include pooling operations that reduce spatial resolution while retaining dominant features, allowing for hierarchical feature extraction \cite{lecun1998}. In QEC, CNNs have proven effective in decoding surface codes and other topological codes, where error patterns are often spatially localized and syndromes resemble images. These networks can learn to recognize common defects, boundaries, and logical operator paths.

Temporal structure is another common feature of QEC data, especially in fault-tolerant settings where syndrome measurements are repeated over time. Recurrent neural networks \cite{elman1990} were developed to model sequences by maintaining an internal state that evolves with the input. A basic RNN processes a sequence $\{x_t\}_{t=1}^T$ using a recurrence relation
\begin{equation}
h_t = \tanh(W x_t + U h_{t-1} + b),
\end{equation}
\begin{equation}
\hat{y}_t = V h_t + c,
\end{equation}
where $h_t$ is the hidden state at time $t$, and $W$, $U$, $V$, $b$, and $c$ are learned parameters. RNNs can, in principle, remember long-range dependencies, but in practice they often suffer from vanishing or exploding gradients \cite{bengio1994}. Long short-term memory (LSTM) networks \cite{hochreiter1997} address this issue with gated mechanisms that regulate the flow of information through memory cells. These gates allow LSTMs to remember or forget past information over long sequences, making them suitable for tasks such as decoding from time-correlated noise or tracking defects through repeated syndrome cycles.

Transformers \cite{vaswani2023} represent a significant departure from RNNs. Originally developed for natural language processing, they rely on attention mechanisms to model dependencies in the input. Attention computes pairwise interactions between all positions in a sequence, allowing the model to focus on relevant parts of the input without relying on recurrence. For a sequence of inputs $\{x_i\}$, the transformer computes queries $q_i = W^Q x_i$, keys $k_i = W^K x_i$, and values $v_i = W^V x_i$, and then uses the scaled dot-product attention formula
\begin{equation}
\alpha_{ij} = \frac{\exp(q_i^\top k_j / \sqrt{d})}{\sum_{l=1}^T \exp(q_i^\top k_l / \sqrt{d})},
\end{equation}
\begin{equation}
z_i = \sum_{j=1}^T \alpha_{ij} v_j,
\end{equation}
where $d$ is the dimensionality of the key vectors. The output $z_i$ for each position is then processed by a feedforward network and further transformed by stacking multiple attention layers. Transformers can be trained efficiently in parallel and achieve state-of-the-art performance on many sequence tasks. In the QEC setting, transformers have been applied to decode syndrome sequences that include both spatial and temporal information. Their ability to learn global patterns and correlations makes them particularly powerful for codes with long-range dependencies or complex error dynamics.

Training all these models involves solving a high-dimensional non-convex optimization problem. Parameters are updated using stochastic gradient descent and its variants. A typical update rule takes the form
\begin{equation}
\theta \leftarrow \theta - \eta \nabla_\theta \mathcal{L}(\theta),
\end{equation}
where $\eta$ is the learning rate. The gradients $\nabla_\theta \mathcal{L}$ are computed using backpropagation, an algorithm that applies the chain rule efficiently through the network layers \cite{rumelhart1986}.

Despite their flexibility, neural networks are prone to overfitting, especially when the training data is limited or unrepresentative of the true distribution. Overfitting occurs when the model learns to memorize the training set rather than generalizing to new examples. This can be mitigated through several techniques. Regularization adds a penalty term to the loss function, such as $\ell_2$ norm on the weights. Dropout randomly disables units during training, which helps prevent co-adaptation. Early stopping monitors the performance on a validation set and halts training when improvement ceases. Data augmentation artificially enlarges the training set by applying transformations that preserve the label, such as flipping or translating syndrome images. In quantum decoding, care must also be taken to train models on data drawn from realistic noise channels, or to adapt to mismatches using transfer learning or domain adaptation.

Classical machine learning offers a spectrum of tools to address decoding challenges in quantum error correction. Feedforward networks provide a generic and simple architecture capable of approximating arbitrary mappings. Convolutional networks exploit spatial locality and are well-suited for planar codes. Recurrent networks model temporal correlations and are useful in repeated measurement scenarios. Transformer models combine the benefits of parallel processing and global attention, allowing them to learn complex structure across space and time. These architectures can also be combined into hybrid models, such as convolutional transformers or recurrent CNNs \cite{gulati2020conformer}, to capture multiple modalities in syndrome data.

\subsection{Feedforward and Convolutional Neural Network Decoders}\label{sec:feedforward}

The earliest ML-based decoder models primarily employed feedforward networks (FFNs) \cite{krastanov2017,varsamopoulos2017}, which map syndrome measurements directly to predicted corrections on the physical qubits. These networks are composed of fully connected layers that process fixed-size input vectors to produce multi-label outputs corresponding to Pauli error corrections on each qubit. While conceptually straightforward, feedforward networks do not exploit the inherent spatial structure of the surface code or other topological codes, where syndrome measurements correspond to vertices on a two-dimensional lattice with local interactions. This lack of inductive bias limits the ability of FFNs to generalize efficiently to larger codes or correlated noise. Nevertheless, Krastanov and Jiang \cite{krastanov2017} demonstrated that even such simple networks could outperform the minimum weight perfect matching decoder in small surface codes under depolarizing noise, achieving higher decoding accuracy with significantly faster inference times.

Building on this foundation, convolutional neural networks were introduced to better capture spatial correlations by leveraging the two-dimensional lattice structure of the syndrome data. In a surface code, each stabilizer corresponds to a localized patch of qubits on a grid, and syndrome measurements naturally form a two-dimensional array. CNNs utilize convolutional filters that slide across this lattice to extract local features, capturing error patterns such as error chains and clusters more effectively than fully connected layers. Weight sharing in convolutions reduces the number of trainable parameters and enforces spatial equivariance, which leads to improved sample efficiency and better generalization performance.

A typical CNN decoder architecture for surface codes consists of multiple convolutional layers with nonlinear activations (such as ReLUs) followed by fully connected layers that output per-qubit Pauli error predictions. The input to the model is a tensor of shape $[B, 1, d, d]$, where $B$ is the batch size and $d$ is the code distance, representing a single round of syndrome measurements. The output is typically a tensor of shape $[B, 3n]$, where $n$ is the number of physical qubits, and the factor 3 corresponds to predictions for $X$, $Y$, and $Z$ error components. For instance, the PyTorch example in Listing~\ref{listing:cnn} defines a simplified CNN decoder architecture.

\begin{lstfloat}
    \centering
    \begin{minipage}{0.65\textwidth}
        \begin{lstlisting}[language=Python]
import torch.nn as nn

class CNNDecoder(nn.Module):
    def __init__(self, code_size, num_qubits):
        super().__init__()
        self.conv_layers = nn.Sequential(
            nn.Conv2d(1, 32, kernel_size=3, padding=1),
            nn.ReLU(),
            nn.Conv2d(32, 64, kernel_size=3, padding=1),
            nn.ReLU(),
            nn.Flatten()
        )
        self.fc_layers = nn.Sequential(
            nn.Linear(64 * code_size * code_size, 256),
            nn.ReLU(),
            nn.Linear(256, num_qubits * 3)
        )

    def forward(self, syndrome):
        features = self.conv_layers(syndrome)
        output = self.fc_layers(features)
        return output\end{lstlisting}
    \end{minipage}
    \caption{Example CNN Decoder Architecture in PyTorch}
    \label{listing:cnn}
\end{lstfloat}

Training such CNNs typically involves simulating noisy syndromes and associated error patterns from well-characterized noise models, such as the depolarizing channel, with labels indicating the actual errors that occurred. The loss function is usually the binary cross-entropy applied element-wise to each predicted Pauli error label. Empirical studies \cite{baireuther2018} demonstrated that CNN decoders outperform MWPM for small and medium-distance surface codes, especially under noise with local correlations. The convolutional inductive bias allows these models to learn local error correlations more efficiently than FFNs. Furthermore, inference with CNN decoders can be implemented with high parallelism on GPUs, offering fast decoding times suitable for real-time error correction.

However, CNN decoders face limitations; their receptive field is constrained by kernel size and network depth. For large codes or highly nonlocal errors, this restricts the decoder's ability to capture long-range correlations. Moreover, these models often operate on syndrome data from a single error correction cycle, ignoring temporal correlations present in multiple rounds of syndrome extraction. Despite these constraints, feedforward and convolutional neural network decoders represent an important class of ML decoders.

\subsection{Recurrent Neural Network Decoders}\label{sec:recurrent}

Realistic quantum error correction protocols consist of repeated rounds of syndrome measurements over time. These sequences of syndromes contain temporal correlations that can be exploited to improve decoding performance, particularly in the presence of temporally correlated noise. Recurrent neural networks (RNNs) are well suited to model such temporal dependencies. By maintaining a hidden state that evolves over a sequence of inputs, RNNs can learn complex patterns in syndrome time series, enabling decoders to correct errors that manifest gradually or correlate across rounds.

A typical RNN decoder accepts a tensor of syndrome sequences with shape $[B, T, S]$, where $B$ is the batch size, $T$ is the number of syndrome measurement rounds, and $S$ is the syndrome vector size per round (the number of generators of the stabilizer subgroup). At each time step, the RNN updates its internal state based on the current syndrome and prior context, and outputs a prediction of the qubit error probabilities, often using the hidden state from the last time step. An illustrative PyTorch implementation of an RNN decoder is shown in Listing~\ref{listing:rnn_decoder}.

\begin{lstfloat}
    \centering
    \begin{minipage}{0.75\textwidth}
        \begin{lstlisting}[language=Python]
import torch.nn as nn

class RNNDecoder(nn.Module):
    def __init__(self, input_size, hidden_size, n_layers, n_qubits):
        super().__init__()
        self.rnn = nn.GRU(input_size=input_size,
                          hidden_size=hidden_size,
                          num_layers=n_layers,
                          batch_first=True)
        self.fc = nn.Linear(hidden_size, n_qubits * 3)

    def forward(self, syndrome_seq):
        # syndrome_seq shape: [batch_size, seq_len, input_size]
        rnn_out, _ = self.rnn(syndrome_seq)
        # Use last hidden state for output
        final_hidden = rnn_out[:, -1, :]
        output = self.fc(final_hidden)
        return output\end{lstlisting}
    \end{minipage}
    \caption{RNN Decoder for Syndrome Time Series}
    \label{listing:rnn_decoder}
\end{lstfloat}

Training RNN decoders requires sequences of syndrome measurements paired with the corresponding error labels over time. The data can be generated by simulating a noise process over multiple error correction cycles, producing syndrome sequences that reflect temporal error dynamics. Cross-entropy loss functions similar to those used in CNN decoders are typically employed to optimize model parameters. Varsamopoulos et al. \cite{varsamopoulos2017,varsamopoulos2020,varsamopoulos2020decoding} demonstrated that RNN-based decoders trained on multi-round syndrome sequences significantly improve logical error rates compared to single-round decoders for surface codes subjected to temporally correlated noise. They reported logical error rate reductions for low distance codes under noise with realistic temporal correlations.

RNNs provide flexibility in handling variable-length syndrome sequences and can capture long-range temporal correlations that other decoders typically neglect. Nonetheless, they do not explicitly model spatial locality within each syndrome round, treating the syndrome vector as a flat input at each time step. This limits their ability to learn spatial error correlations as effectively as CNNs. Training RNNs can also be challenging due to issues like vanishing gradients, especially for long sequences. Techniques such as gradient clipping \cite{pascanu2013}, careful initialization \cite{glorot2010}, and the use of gated units mitigate these difficulties \cite{hochreiter1997}.

\subsection{Transformer-Based Decoders and the AlphaQubit Model}\label{sec:transformer}

Transformers, originally developed for natural language processing tasks \cite{vaswani2023}, have gained prominence in quantum error correction decoding for their remarkable ability to model both spatial and temporal correlations in syndrome data through self-attention mechanisms. Unlike recurrent neural networks, transformers process input sequences in parallel and can capture long-range dependencies without the constraints of fixed receptive fields or sequential processing bottlenecks.

In the context of QEC, syndrome measurements form structured spatiotemporal data; syndromes are arranged on two-dimensional lattices and measured repeatedly over many rounds. Whereas other decoders often struggle to handle the interplay of spatial and temporal correlations, especially under realistic noise featuring correlated errors \cite{zhou2025}, leakage \cite{brown2020}, and measurement imperfections, transformer architectures can directly attend to relevant syndrome features across space and time, enabling robust decoding strategies.

Google Quantum AI’s \textit{AlphaQubit} decoder \cite{bausch2024,willow2024} represents a state-of-the-art implementation of this approach. AlphaQubit integrates convolutional layers to propagate information locally between neighboring stabilizers with transformer encoder layers that model global correlations across the code lattice and across multiple rounds of syndrome measurement. This architecture allows the decoder to maintain and update a latent ``memory'' state per stabilizer, capturing the dynamics of error processes that evolve in time and space. The AlphaQubit model was trained in two stages. Initially, it underwent extensive pretraining on synthetically generated syndrome data simulating a wide range of noise models, including depolarizing noise, leakage, and correlated errors. Subsequently, it was fine-tuned on experimentally collected syndrome data from Google’s Sycamore processor, which includes realistic hardware noise features such as analog syndrome measurements and crosstalk. This transfer learning approach enabled AlphaQubit to generalize from idealized noise to real-world quantum hardware effectively.

Experimental results demonstrate that AlphaQubit outperforms minimum weight perfect matching on surface codes of distance three and five. Moreover, simulations indicate scalability to larger distances (up to eleven), achieving lower logical error rates under practical noise regimes. These results mark a significant milestone toward deploying ML-based decoders in fault-tolerant quantum computing. At a technical level, AlphaQubit’s transformer components operate on a sequence of latent vectors representing stabilizer states over multiple measurement rounds. The self-attention layers compute pairwise interactions between all stabilizers and time steps, weighting the importance of each syndrome input when forming the updated latent representation. This mechanism enables the decoder to focus dynamically on relevant spatiotemporal patterns, such as error chains extending over many qubits and time intervals. Although the full implementation details of AlphaQubit are proprietary, a simplified transformer encoder block for syndrome decoding can be illustrated in PyTorch according to Listing~\ref{listing:transformer_block}.

\begin{lstfloat}[t]
    \centering
    \begin{minipage}{0.86\textwidth}
        \begin{lstlisting}[language=Python]
import torch.nn as nn

class TransformerEncoderBlock(nn.Module):
    def __init__(self, d_model, nhead, dim_feedforward=2048, dropout=0.1):
        super().__init__()
        self.self_attn = nn.MultiheadAttention(d_model, nhead, dropout=dropout)
        self.linear1 = nn.Linear(d_model, dim_feedforward)
        self.dropout = nn.Dropout(dropout)
        self.linear2 = nn.Linear(dim_feedforward, d_model)

        self.norm1 = nn.LayerNorm(d_model)
        self.norm2 = nn.LayerNorm(d_model)
        self.dropout1 = nn.Dropout(dropout)
        self.dropout2 = nn.Dropout(dropout)
        self.activation = nn.ReLU()

    def forward(self, src):
        # src shape: [seq_len, batch_size, d_model]
        src2, _ = self.self_attn(src, src, src)
        src = src + self.dropout1(src2)
        src = self.norm1(src)
        src2 = self.linear2(self.dropout(self.activation(self.linear1(src))))
        src = src + self.dropout2(src2)
        src = self.norm2(src)
        return src\end{lstlisting}
    \end{minipage}
    \caption{Simplified Transformer Encoder Block}
    \label{listing:transformer_block}
\end{lstfloat}

Training such transformer decoders is computationally intensive due to the quadratic complexity of self-attention with respect to sequence length \cite{keles2022}. Nevertheless, the success of AlphaQubit underscores the promise of transformer-based decoders to scale efficiently with code distance and adapt to complex noise profiles, leveraging advances in deep learning architectures originally developed for natural language and vision tasks.

\subsection{Hybrid and Variational Quantum Decoders}\label{sec:hybrid}

While purely classical machine learning decoders, such as convolutional, recurrent, and transformer-based architectures, have achieved impressive results, they are not the only path forward. Hybrid approaches, which combine classical decoding heuristics with machine learning components, and quantum-enhanced models such as variational quantum circuits (VQCs), represent alternative strategies for leveraging machine learning in quantum error correction.

One hybrid approach augments traditional classical decoding algorithms with neural networks that assist in the decoding process. A common example is to pair minimum weight perfect matching or belief propagation with a neural network that learns to predict local error likelihoods or guide the decoder toward better correction paths. For instance, Chamberland and Ronagh \cite{chamberland2018} trained a neural network to estimate the logical error likelihood conditioned on syndrome data, allowing the decoder to make more informed recovery choices. By combining domain-specific decoding heuristics with statistical learning, hybrid decoders gain robustness against correlated and hardware-specific noise without sacrificing computational efficiency.

Other methods inject classical algorithm outputs into neural networks as auxiliary features. For example, one can use partial decoding paths from MWPM or lookup tables as hints, concatenated with the syndrome inputs before being passed into a deep neural network for final classification. This fusion of algorithmic and statistical techniques enables decoders to blend fast approximate solutions with data-driven correction. A minimal illustration of such a hybrid decoder in PyTorch appears in Listing~\ref{listing:hybrid_decoder}. In this example, a classical preprocessing step produces heuristic features, which are concatenated with the raw syndrome vector and passed to a simple feedforward model.

\begin{lstfloat}[t]
    \centering
    \begin{minipage}{0.8\textwidth}
        \begin{lstlisting}[language=Python]
import torch.nn as nn

class HybridDecoder(nn.Module):
    def __init__(self, syndrome_dim, hint_dim, hidden_dim, num_qubits):
        super().__init__()
        self.fc = nn.Sequential(
            nn.Linear(syndrome_dim + hint_dim, hidden_dim),
            nn.ReLU(),
            nn.Linear(hidden_dim, num_qubits * 3)
        )

    def forward(self, syndrome, classical_hint):
        x = torch.cat((syndrome, classical_hint), dim=-1)
        return self.fc(x)\end{lstlisting}
    \end{minipage}
    \caption{Hybrid Neural Decoder Combining Classical Hints}
    \label{listing:hybrid_decoder}
\end{lstfloat}

Here, \texttt{syndrome} is the binary stabilizer outcome vector, and \texttt{classical\_hint} might represent distances to defects, partial correction candidates, or outputs from a fast decoder. The network output is a prediction of Pauli corrections per qubit. These approaches can be trained on realistic hardware-generated syndrome datasets. Because classical heuristics often succeed under idealized noise assumptions but fail in practice due to hardware-specific errors (such as leakage and crosstalk), neural augmentations learn to compensate for deviations from the ideal, improving the overall fidelity.

Another compelling class of models leverages quantum processors themselves to perform parts of the decoding task. Variational quantum decoders, built using parameterized quantum circuits \cite{schuld2020}, operate as hybrid quantum-classical models. These circuits take syndrome data as classical input and transform them into quantum states using an embedding procedure. A trainable quantum circuit processes the embedded data, and measurements on ancillary qubits yield correction predictions. Training such models typically involves optimizing a loss function defined over measurement outcomes, with parameters updated via classical gradient descent. Due to their quantum-native nature, variational decoders are capable of modeling high-dimensional, non-classical correlations that may be inaccessible to classical decoders. They are especially attractive for near-term quantum devices operating in the noisy intermediate-scale quantum (NISQ) \cite{preskill2018} regime.

Recent developments in quantum machine learning \cite{schuld2015,schuld2021machine} suggest further enhancements to this paradigm. We suspect that one particularly promising direction is geometric quantum machine learning (GQML) \cite{ragone2023,larocca2022}, which leverages the symmetries of data to guide the design of more expressive and generalizable quantum models. In the context of decoding, GQML may be used to construct equivariant quantum circuits \cite{bradshaw2025} that respect the symmetries of the underlying error-correcting code. This aligns with work on equivariant quantum classifiers, which have recently seen some success \cite{meyer2023}.

A simplified workflow for a variational decoder implemented using PennyLane or Qiskit might involve the following components:
\begin{enumerate}
    \item A data embedding routine that maps binary syndrome vectors to quantum states using parameterized rotations.
    \item A trainable quantum circuit composed of single-qubit rotations and entangling gates.
    \item A measurement protocol that extracts information from the output state and maps it to classical predictions of Pauli errors.
    \item A classical optimizer that minimizes a fidelity-based loss function.
\end{enumerate}
A high-level pseudocode sketch of such a training loop could be as in Listing~\ref{listing:training_loop}.
Here, \texttt{quantum\_circuit} refers to the quantum function representing the decoder, parameterized by variational angles. The function \texttt{fidelity\_loss} measures the closeness between the predicted correction and the ideal one.

\begin{lstfloat}[t]
    \centering
    \begin{minipage}{0.71\textwidth}
        \begin{lstlisting}[language=Python]
for epoch in range(num_epochs):
    for syndrome_batch, target_corrections in dataloader:
        def loss_fn(params):
            predictions = quantum_circuit(syndrome_batch, params)
            return fidelity_loss(predictions, target_corrections)

        gradients = compute_gradients(loss_fn, params)
        params = optimizer.step(params, gradients)\end{lstlisting}
    \end{minipage}
    \caption{Pseudocode for Training a Variational Quantum Decoder}
    \label{listing:training_loop}
\end{lstfloat}

Despite their theoretical potential, variational quantum decoders remain in early experimental stages. Their limitations include finite sampling noise, optimization instability, and hardware error accumulation \cite{wang2021, depalma2023}. Furthermore, their expressive advantage over classical decoders has not yet been conclusively demonstrated for large-scale error correcting codes. Nevertheless, they represent a promising avenue for integrating quantum computation into the decoding process itself, especially for use cases where classical resources are limited or classical training is infeasible due to system complexity. Continued advances in geometric methods, symmetry-aware models, and quantum-classical hybrid architectures may help close the gap between theoretical expressiveness and practical viability in the near future.

\section*{Acknowledgment}
A portion of this work was done at the Naval Surface Warfare Center in Panama City. The authors acknowledge support from the NSWC Panama City Division’s In-house Laboratory Independent Research (ILIR) program funded by the Office of Naval Research (ONR) under N0001425GI00778. This document is Distribution Statement A, distribution is unlimited. ZPB thanks Ada H. Bradshaw for her invaluable assistance in maintaining a low signal-to-noise ratio throughout the writing process.

\bibliographystyle{unsrt}
\bibliography{refs}

\newpage
\section*{Appendix}

\setcounter{subsection}{0}
\renewcommand{\thesubsection}{\Alph{subsection}}

% -------------------------------------------------------------------------------------------------------------
\subsection{Noisy Communication Channels} \label{sec:appendix:impl:noise}
% -------------------------------------------------------------------------------------------------------------

When demonstrating a quantum code, it becomes necessary to simulate noisy
communication channels in order to see how errors are corrected in the studied code.
To this end, we put ancillary qubits in a state for which the probability of measuring the $\ket{1}$ state is known
and controllable so that we can induce error at a known rate.
This can be accomplished by employing an $R_y$ gate with a well-chosen rotation parameter $\phi$, which is
related to the probability of observing a $\ket{1}$ according to
\begin{equation}
    \phi = \arcsin(2p-1),
\end{equation}
where $p=P(\ket{1})$ is the probability of observing $\ket{1}$. This can be seen directly by applying a Hadamard gate and an $R_y$ gate to the $\ket{0}$ state, producing
\begin{align}
    \ket{\psi} &= R_y(\phi) H \ket{0} \\
               &= \frac{1}{\sqrt{2}} \begin{bmatrix} \cos(\frac{\phi}{2}) - \sin(\frac{\phi}{2}) \\ \sin(\frac{\phi}{2}) + \cos(\frac{\phi}{2}) \end{bmatrix},
\end{align}
from which it follows that the probability of observing the $\ket{1}$ state is
\begin{align}
    P(\ket{1}) &= \left[ \frac{1}{\sqrt{2}}\left( \sin\frac{\phi}{2} + \cos\frac{\phi}{2} \right) \right]^2 \\
               &= \frac{1}{2} \left( \sin^2 \frac{\phi}{2} + \cos^2 \frac{\phi}{2} + 2\sin\frac{\phi}{2}\cos\frac{\phi}{2} \right) \\
               &= \frac{1}{2} \left( 1 + \sin(\phi) \right).
\end{align}
Solving, we obtain $\phi = \arcsin(2p-1)$.

Let $E$ denote a single qubit error that we wish to apply with probability $p$ to one of the physical qubits making up the logical state. After the ancillary qubit above is constructed, we control off of it and target one of the physical qubits with the $E$ operation. Finally, we measure the ancillary qubit so that $E$ is applied to the physical qubit with probability $p$. In order to model independent errors for each physical qubit, we prepare a number of ancillary qubits equal to the number of physical qubits in the logical state and perform the above procedure for each pair. The example of the three qubit code is shown in Figure~\ref{fig:3qubit-qiskit}.

\begin{figure}
    \centering
    \includegraphics[width=6in]{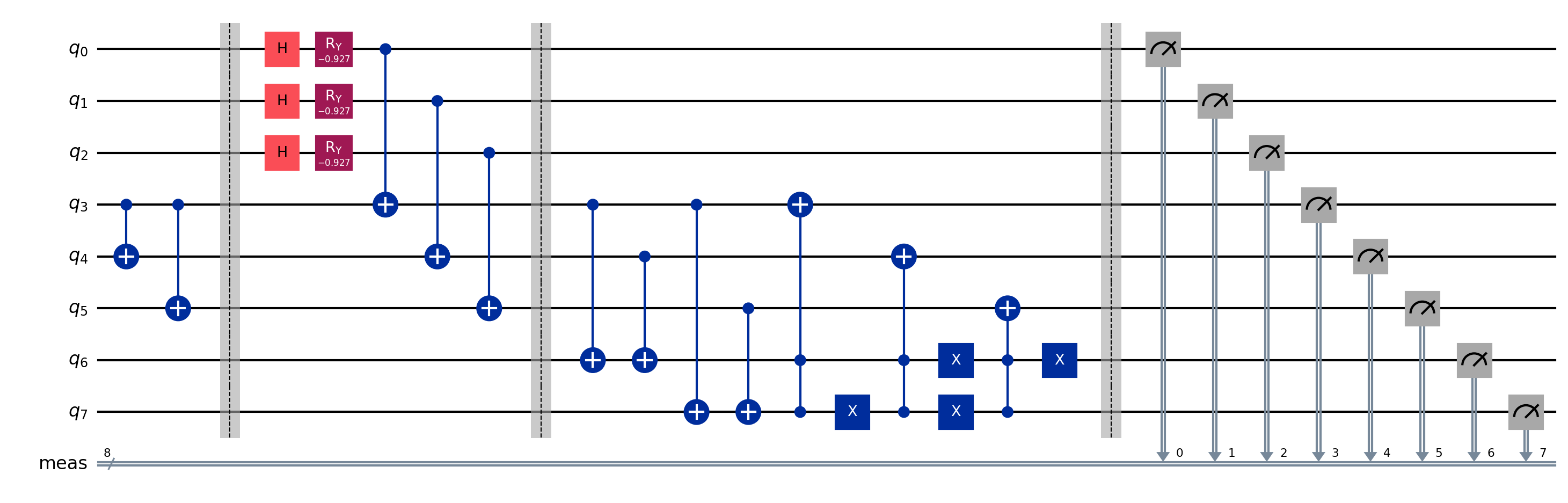}
    \caption{
        The 3-qubit code implemented in a Qiskit circuit diagram with the noise channel shown.
    }
    \label{fig:3qubit-qiskit}
\end{figure}

For the codes which protect against all single qubit unitary errors, we perform the above procedure for $E=X$ and then perform the procedure again for $E=Z$. The probability of observing a single $X$ error is then $np(1-p)^{2n-1}$, where $n$ is the number of physical qubits. Similarly, the probability of observing a single $Z$ error is $np(1-p)^{2n-1}$. The probability that an $X$ and $Z$ are triggered on the same qubit, producing a $Y$ up to a global phase, is $np^2(1-p)^{2n-2}$. For the Shor code in particular, this simulation requires too many ancillary qubits, and so we limit our attention to a subset of the possible errors as seen in Figure~\ref{fig:shor-qiskit}.

\begin{figure}
    \centering
    \includegraphics[width=6in]{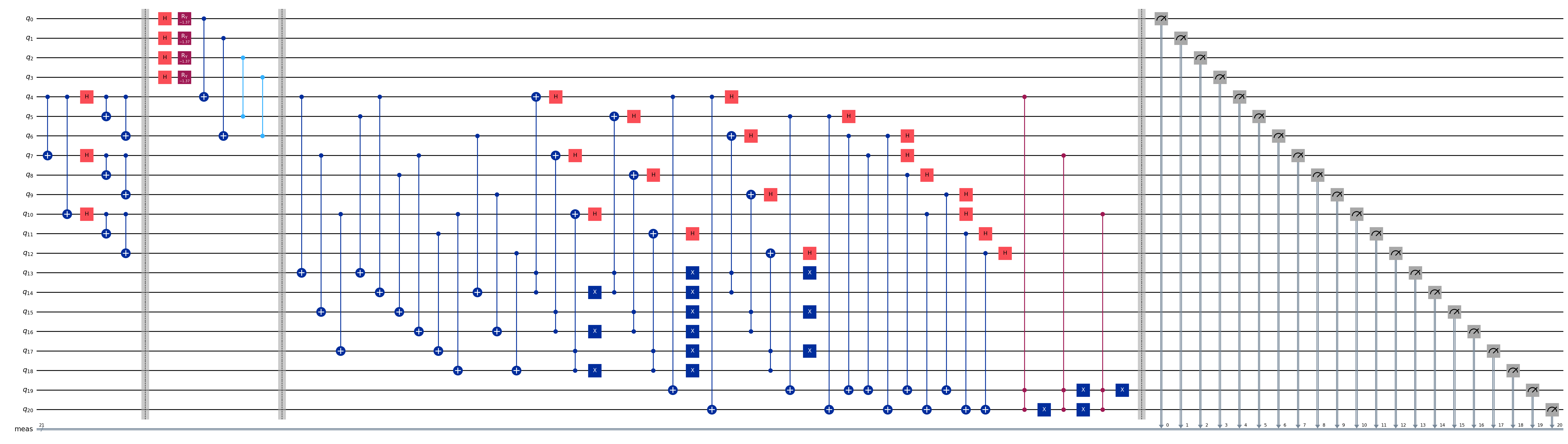}
    \caption{
        The Shor code implemented in a Qiskit circuit diagram with the noise channel shown.
    }
    \label{fig:shor-qiskit}
\end{figure}

\end{document}